\newcommand{\onlyFull}[1]{#1}
\newcommand{\refApp}[1]{#1}
\documentclass[acmsmall,nonacm,screen,prologue,dvipsnames]{acmart}
\newenvironment{DIFnomarkup}{}{} %
\usepackage[utf8]{inputenc} 
\usepackage{hyperref}
\usepackage[capitalize]{cleveref}
\usepackage{thm-restate}
\usepackage{amsmath, mathtools,amsfonts,amsthm}
\usepackage{algorithmic}
\usepackage{stmaryrd}
\usepackage{array}
\usepackage{url}
\usepackage[disable]{todonotes}
\usepackage{subcaption}
\usepackage{enumitem}
\usepackage{shuffle}

\usepackage{pifont}
\newcommand{\cmark}{\ding{51}}%
\newcommand{\xmark}{\ding{55}}%

\usepackage{tikz}
\usetikzlibrary{arrows.meta}
\usetikzlibrary{calc}
\usetikzlibrary{automata}
\usetikzlibrary{decorations.pathmorphing}
\usetikzlibrary{tikzmark}
\tikzstyle{branch}=[fill,shape=circle,minimum size=3pt,inner sep=0pt]

\pgfdeclarelayer{back}
\pgfdeclarelayer{front}
\pgfsetlayers{back,main,front}

\makeatletter
\pgfkeys{%
	/tikz/on layer/.code={
		\pgfonlayer{#1}\begingroup
		\aftergroup\endpgfonlayer
		\aftergroup\endgroup
	},
	/tikz/node on layer/.code={
		\gdef\node@@on@layer{%
			\setbox\tikz@tempbox=\hbox\bgroup\pgfonlayer{#1}\unhbox\tikz@tempbox\endpgfonlayer\egroup}
		\aftergroup\node@on@layer
	},
	/tikz/end node on layer/.code={
		\endpgfonlayer\endgroup\endgroup
	}
}

\def\node@on@layer{\aftergroup\node@@on@layer}
\makeatother

\usepackage{totcount}
\usepackage{refcount}
\usepackage{newfile}

\newcommand{\tuple}[1]{\left(#1\right)}

\newcommand{\mcfg}{MCFG}

\newcommand{\kmcfg}{\ensuremath{k}\textsf{-}\mcfg}

\newcommand{\grammar}{\mathcal{G}}
\newcommand{\initgrammar}{\mathcal{G}}
\newcommand{\Nonterminals}{N}
\newcommand{\Terminals}{\Sigma}
\newcommand{\nt}[2]{#1^{(#2)}}
\newcommand{\Productions}{P}
\newcommand{\StartNonterminal}{A_\text{start}}
\newcommand{\prodrule}{\pi}

\newcommand{\twoicfg}{\ensuremath{2}\textsf{-\mcfg}}

\newcommand{\head}{\textsf{head}}
\newcommand{\body}{\textsf{body}}

\newcommand{\Variables}{X}
\newcommand{\mcfgrule}[2]{#1 \leftarrow #2}
\newcommand{\lograph}{G}
\newcommand{\LOG}{\text{lo-}\text{graph}}
\newcommand{\logset}[2]{\textsc{lo-GraphSet}(#1,#2)}
\newcommand{\sedge}{\textsc{Succ}}
\newcommand{\node}[2]{\textsc{Node}({#1, #2})}
\newcommand{\addsedge}[2]{\sedge({#1, #2})}
\newcommand{\addedge}[3]{\textsc{Edge}({#1, #2, #3})}
\newcommand{\fcolor}[1]{\textsc{Free}({#1})}
\newcommand{\join}{\textsc{Join}}
\newcommand{\ecprg}{\bot}
\newcommand{\vlabel}{\lambda}
\newcommand{\Elabels}{\Gamma}
\newcommand{\MSO}{\textsf{MSO}}
\newcommand{\glangof}[1]{\mathcal{L}_{\mathsf{G}}(#1)}
\newcommand{\wlangof}[1]{\mathcal{L}_{\mathsf{W}}(#1)}
\newcommand{\bglangof}[2]{\mathcal{L}^{#1}_{\mathsf{G}}(#2)}
\newcommand{\bwlangof}[2]{\mathcal{L}^{#1}_{\mathsf{W}}(#2)}
\newcommand{\pglangof}[1]{\mathcal{L}_{\hat{\mathsf{G}}}(#1)}

\newcommand{\phigrid}{\phi_\text{grid}}
\newcommand{\phicopy}{\phi_\text{copy}}

\newcommand{\tup}{\textsf{tup}}
\newcommand{\TB}{\mathcal{B}}

	\newcommand{\matching}{\curvearrowright}
\newcommand{\phiwellnesting}{\phi_\text{well-nesting}}
\newcommand{\maxarity}{d}

\newcommand{\mpath}{\textsf{mpath}}
\newcommand{\maxmpath}{\textsf{maxmpath}}
\newcommand{\Ampath}{\textsf{Ampath}}
\newcommand{\maxAmpath}{\textsf{maxAmpath}}
\newcommand{\eps}{\textsf{eps}}
\newcommand{\prods}{\textsf{prod}}
\newcommand{\scan}{\textsf{scan}}
\newcommand{\nextscan}{\textsf{nextscan}}

\newcommand{\parsetreetikz}{
	\begin{tikzpicture}[node distance=3mm, initial text=,>=latex, thick,	runnode/.style ={fill,circle,inner sep=0pt, outer sep = 0pt, minimum size=5pt}]
		
		\node[runnode,color=highlight-1, label=right:$\StartNonterminal(baaabb)$] (n1) {};
		\node[runnode, color = highlight-2,  label=right:\ensuremath{A(aabb,ba)}, below = of n1] (n2) {};
		\node[runnode, color=highlight-3,label=left:\ensuremath{A(a,aa)}, below left= 3mm and 3mm of n2] (nL) {};
		\node[runnode,color=highlight-4, label=right:\ensuremath{B(b,b,b)}, below right= 3mm and 3mm of n2] (nR) {};
		
		\draw[-{latex}] (n2) edge(n1);
		\draw[-{latex}] (nL) edge(n2);
		\draw[-{latex}] (nR) edge(n2);
		
	\end{tikzpicture}
}
\newcommand{\nestedwordtikz}{
	\begin{tikzpicture}[node distance=3mm, initial text=,>=latex, thick,	runnode/.style ={draw ,circle,inner sep=0pt, outer sep = 0pt, minimum size=10pt}]
		
		\node[runnode,] (n1) {1};

		\foreach \x [remember=\x as \lastx (initially 1)] in {2,...,22}
		{\node[runnode, right = of n\lastx] (n\x) {\x};}

		{\node[, minimum size=0pt, below = 0.05mm of n4] (b4) {   $b$ };}
		{\node[,  below = 0.05mm of n7] (b7) { \vphantom{$b$} $a$};}
		{\node[,  below = 0.05mm of n12] (b12) { \vphantom{$b$} $a$};}
		{\node[, , below = 0.05mm of n13] (b13) { \vphantom{$b$} $a$};}
		{\node[, below = 0.05mm of n16] (b16) { $b$};}
		{\node[,  below = 0.05mm of n19] (b19) {  $b$};}
		
			\draw[ -{Classical TikZ Rightarrow[length=3pt]}, ] 	\foreach \x [remember=\x as \lastx (initially 1)] in {2,...,22}
		{(n\lastx) edge (n\x)} ;

		\draw[->, thin,  highlight-1 ] (n1) edge[out=60, in =120, looseness=.3] (n22);
		
		\draw[->, thin, highlight-2 ] (n10) edge[out=60, in =120, looseness=.4]  (n21)
		(n21)edge[out=240, in=300, looseness= .3]  (n2)
		(n2)edge[out=60, in =120, looseness=.5] (n9);

		\draw[-> ,thin, highlight-3 ] (n6) edge[ out=60, in =120, looseness=1.2]  (n8)
		(n8)edge[ out=300, in=240, looseness= 1 ]  (n11)
		(n11)edge[bend left = 12mm, out=60, in =120, looseness= 1] (n14);
		
		\draw[->,thin, highlight-4 ] (n3) edge[out=60, in =120, looseness=1.2]  (n5)
		(n5)edge[out=300, in=240, looseness= .3 ]  (n18)
		(n18)edge[out=60, in =120, looseness=1.2]  (n20)
		(n20)edge[out=240, in=300, looseness= .6]  (n15)
		(n15)edge[out=60, in =120, looseness=1.2]  (n17);

	\end{tikzpicture}
}

\newcommand{\rungraphletter}{
	\begin{tikzpicture}[node distance=2.5mm, initial text=,>=latex, thick,	runnode/.style ={draw ,circle,inner sep=0pt, outer sep = 0pt, minimum size=10pt}]
		
		\node[runnode, label = below:{\phantom{\ensuremath{\closetwo}}\opentwo}] (n1) {1};
		\foreach \x / \y[remember=\x as \lastx (initially 1)] in { 2/\openone, 3/\openone, 4/\opentwo, 5/\closeone, 6/\closetwo, 
			7/\opentwo,  8/\openone,  9/\closetwo,  10/\closeone, 
			11/\openone,  12/\opentwo,  13/\closeone,  14/\closetwo,  15/\closetwo,  16/\closeone
		}
		{\node[runnode, right = of n\lastx, label = below:{\phantom{\ensuremath{\closetwo}}\y}] (n\x) {\x};}
		
		\draw[-{Classical TikZ Rightarrow[length=3pt]}, ] 	\foreach \x [remember=\x as \lastx (initially 1)] in {2,...,16}
		{(n\lastx) edge (n\x)} ; 
	   \begin{pgfonlayer}{back}
					\draw[->, bend left, color=highlight-1, out=60, in =120] 
			\foreach \x / \y  in {3/5, 8/10, 11/13}
			{(n\x) edge (n\y)} ; 
			\draw[->, bend left, color=highlight-1, out=60, in =120, looseness=.3]  (n2) edge (n16);
			\draw[->, bend right, color=highlight-2, out=300, in=240] 
			\foreach \x / \y  in { 4/6, 7/9, 12/14}
			{(n\x) edge (n\y)} ; 
			\draw[->, bend left, color=highlight-2, out=300, in =240, looseness=.3]  (n1) edge (n15);
		    \end{pgfonlayer}

	\end{tikzpicture}
}

\newcommand{\cO}{\mathcal{O}}
\newcommand{\Img}{\mathrm{Img}}
\newcommand{\poly}{\mathrm{poly}}

\newcommand{\summ}{\sigma}

\newcommand{\image}[1]{\textsf{Img}(#1)}

\newcommand{\sizeof}[1]{|#1|}

\newcommand{\set}[1]{\{1, 2,\dots, #1\}}

\newcommand{\compabstr}{\textsf{piece-summ}}

\newcommand{\summarycomp}{\textsf{piece-smry}}

\newcommand{\Smrycomp}[1]{\textsf{pieceSmrySet}(#1)}

\newcommand{\deltacomp}{\delta^\textsf{piece}}
\newcommand{\TArun}[1]{\llbracket#1 \rrbracket}
\newcommand{\TA}{\mathcal{A}}

\newcommand{\Alphabet} {\Sigma}
\newcommand{\letter}{\textsf{letter}}
\newcommand{\word}{\textsf{word}}

\newcommand{\Nat} {\mathbb{N}}
\newcommand{\lang} {{L}}
\newcommand{\subword} {\preceq}

\newcommand{\downc}[1]{#1\mathord{\downarrow}}

\newcommand{\MPDS} {\mathcal{M}}
\newcommand{\mpds} {MPDA}

\newcommand{\mpdsops} [1]{{\sf ops_{#1}}}
\newcommand{\op}{{\sf op}}
\newcommand{\tgtstate}{\texttt{tgt}}
\newcommand{\srcstate}{\texttt{src}}

\newcommand{\Succ} {{\tt Succ}}
\newcommand{\Match} {{\tt Match}}

\newcommand{\tmap}{\textsf{trans}}

\newcommand{\states} {Q}

\newcommand{\initstate}{q_0}
\newcommand{\trans}{\Delta}
\newcommand{\mpdspush}[2]{{\sf push}(#1,#2)}
\newcommand{\mpdspop}[2]{{\sf pop}(#1,#2)}

\newcommand{\noop}{{\sf nop}}
\newcommand{\numstack}{s}
\newcommand{\alp}{\Sigma}
\newcommand{\stackalp}{\Gamma}
\newcommand{\finstates} {f}
\newcommand{\MPDStuple} {\tuple{\states,\stackalp,\initstate,\trans, \finstates}}

\newcommand{\langof}[1] {\mathcal{L}(#1)}

\newcommand{\vertices} {{V}}

\newcommand{\edges} {{E}}

\newcommand{\colorlab} {\chi}
\newcommand{\colors} {\mathsf{C}}

\newcommand{\expr}{{\tt expr}}

\newcommand{\src}{\textsf{src}}
\newcommand{\tgt}{\textsf{tgt}}

\newcommand{\setof}[1]{\{#1\}}
\DeclarePairedDelimiter{\card}{\lvert}{\rvert}
\newcommand{\embed}{\hookrightarrow}
\newcommand{\id}{\mathrm{id}}
\newcommand{\cG}{\ensuremath{\mathcal{G}}}

\newcommand{\cS}{\ensuremath{\mathcal{S}}}
\newcommand{\yield}{\mathrm{yield}}
\newcommand{\subtree}[2]{#1/#2}
\newcommand{\height}{\mathrm{height}}

\newcommand{\CP}{\mathsf{CP}}
\newcommand{\cC}{\mathcal{C}}
\newcommand{\ap}{\mathfrak{P}}

\definecolor{wong-green}{HTML}{009E73}
\definecolor{wong-blue}{HTML}{0072B2}
\definecolor{wong-orange}{HTML}{D55E00}
\definecolor{wong-pink}{HTML}{CC79A7}

\colorlet{highlight-1}{wong-orange}
\colorlet{highlight-2}{wong-pink!80!black}
\colorlet{highlight-3}{wong-green}
\colorlet{highlight-4}{wong-blue}
	
\colorlet{highlight}{highlight-1}

\newcommand{\hlA}[1]{\textcolor{highlight-1}{#1}}
\newcommand{\hlB}[1]{\textcolor{highlight-2}{#1}}
\newcommand{\hlC}[1]{\textcolor{highlight-3}{#1}}

\newcommand{\coNP}{\mathsf{coNP}}
\newcommand{\EXPSPACE}{\mathsf{EXPSPACE}}
\newcommand{\TWOEXPSPACE}{\mathsf{2EXPSPACE}}
\newcommand{\THREEEXPSPACE}{\mathsf{3EXPSPACE}}
\newcommand{\FIVEEXPSPACE}{\mathsf{5EXPSPACE}}

\AtEndPreamble{
	\newtheorem{fact}[theorem]{Fact}
	\newtheorem{remark}[theorem]{Remark}
	\newtheorem{claim}[theorem]{Claim}
	\newcounter{matchitemctr}
	\newcounter{matchsubitemctr}
	\newcommand\matchitem{\newline\stepcounter{matchitemctr}{\color{violet}\textsc{m}\arabic{matchitemctr})\,\,\,\,}}
	\newcommand\matchsubitem{\newline\indent\stepcounter{matchsubitemctr}{\color{violet}\textsc{m}\arabic{matchitemctr}\alph{matchsubitemctr})\,\,\,\,}}
}

\newcommand{\longcomment}[1]{}
\newcommand{\Lang}{\mathcal{L}}
\newcommand{\statepush}[1]{\langle #1, \textsf{push} \rangle}
\newcommand{\stateswitch}[1]{\langle #1, \textsf{switch}\rangle}
\newcommand{\statepop}[1]{\langle #1, \textsf{pop}\rangle}

\newcommand{\stackset}{[1..\numstack]}
\newcommand{\oper}{\texttt{oper}}
\newcommand{\rungraph}{\rho}

\newcommand{\streewidth}{\textsf{stw}}

\newcommand{\PTIME}{\mathsf{P}}
\newcommand{\NP}{\mathsf{NP}}
\newcommand{\EXPTIME}{\mathsf{EXPTIME}}

\crefname{appsec}{Appendix}{Appendix}
\Crefname{appsec}{Appendix}{Appendix}



\begin{document}

\title{Bounded Treewidth, Multiple Context-Free Grammars, and Downward Closures}

\author{C. Aiswarya}
\affiliation{%
  \institution{Chennai Mathematical Institute and CNRS, ReLaX, IRL 2000}
  \city{Chennai}
  \country{India}
} %
\email{aiswarya@cmi.ac.in} %
\orcid{0000-0002-4878-7581} %

\author{Pascal Baumann}
\affiliation{
  \institution{Max Planck Institute for Software Systems (MPI-SWS)}
  \city{Kaiserslautern}
  \country{Germany}
}
\email{pbaumann@mpi-sws.org}
\orcid{0000-0002-9371-0807}

\author{Prakash Saivasan}
\affiliation{
  \institution{The Institute of Mathematical Sciences, HBNI, and CNRS, ReLaX, IRL 2000}
  \city{Chennai}
  \country{India}
} %
\email{prakashs@imsc.res.in} %
\orcid{0000-0001-5060-0117} %

\author{Lia Sch\"{u}tze}
\affiliation{
  \institution{Max Planck Institute for Software Systems (MPI-SWS)}
  \city{Kaiserslautern}
  \country{Germany}
}
\email{lschuetze@mpi-sws.org} %
\orcid{0000-0003-4002-5491} %

\author{Georg Zetzsche}
\affiliation{
  \institution{Max Planck Institute for Software Systems (MPI-SWS)}
  \city{Kaiserslautern}
  \country{Germany}
}
\email{georg@mpi-sws.org} %
\orcid{0000-0002-6421-4388} %

\begin{abstract}
	The reachability problem in multi-pushdown automata (MPDA), or equivalently, interleaved Dyck reachability, has many applications in static analysis of recursive programs.
	An example is safety verification of multi-threaded recursive programs with shared memory.
	Since these problems are undecidable, the literature contains many decidable (and efficient) underapproximations of MPDA. 

	A uniform framework that captures many of these underapproximations is
	that of \emph{bounded treewidth}: To each execution of the MPDA, we
	associate a graph; then we consider the subset of all graphs that have
	a treewidth at most $k$, for some constant $k$. In fact, bounding treewidth is a generic approach to obtain classes of systems with decidable reachability, even beyond MPDA underapproximations. The resulting systems are also called \emph{MSO-definable bounded-treewidth systems}.

	While bounded treewidth is a powerful tool for reachability and similar types
	of analysis, the word languages (i.e.\ action sequences corresponding
	to executions) of these systems remain far from understood.
	For the slight restriction of bounded \emph{special treewidth}, or
	``bounded-stw'' (which is equivalent to bounded treewidth on MPDA, and even
	includes all bounded-treewidth systems studied in the literature), this work
	reveals a connection with multiple context-free languages (MCFL), a concept
	from computational linguistics. We show that the word languages of MSO-definable bounded-stw
	systems are exactly the MCFL.  

	We exploit this connection to provide an optimal algorithm for
	computing \emph{downward closures} for MSO-definable bounded-stw systems.
	Computing downward closures is a notoriously difficult task that has
	many applications in the verification of complex systems: As an example
	application, we show that in programs with dynamic spawning of
	MSO-definable bounded-stw processes, safety verification has the same
	complexity as in the case of processes with sequential recursive
	processes.
\end{abstract}

\maketitle

\section{Introduction}\label{sec:introduction}
\newcommand{\myparagraph}[1]{\vspace{0.2cm}\noindent\textbf{#1}. }
\myparagraph{Reachability in multi-pushdown automata}
We study the complexity of static analysis tasks for multi-threaded shared-memory (recursive) programs and related systems. Such programs consist of concurrent threads, each of which performs a recursive sequential computation. The threads have access to shared variables. Moreover, we assume that all program variables range over finite domains. Safety verification of such programs can be modeled as reachability in multi-pushdown automata (MPDA). These feature a finite set of control states (which model the shared memory) and have access to multiple pushdown stores (which model the call stacks of the individual threads).

For this reason, the reachability problem in MPDA has received a vast amount of attention in recent decades (see below). Another reason for this strong interest is that MPDA reachability is also equivalent to the problem of \emph{interleaved Dyck reachability}, which can model many important verification tasks beyond safety, but in recursive sequential programs~\cite{DBLP:journals/pacmpl/LiZR21}, such as 
structure-transmitted data-dependence~\cite{DBLP:journals/toplas/Reps00},
typestate analysis~\cite{DBLP:journals/pacmpl/SpathAB19},
alias analysis~\cite{DBLP:conf/pldi/SridharanB06,DBLP:conf/pldi/ChengH00}, %
and taint analysis~\cite{DBLP:conf/pldi/ArztRFBBKTOM14,DBLP:conf/issta/HuangDMD15}. %

\myparagraph{Parameterized underapproximations} Since in full generality, reachability in MPDA is undecidable~\cite{ramalingam2000context,DBLP:journals/toplas/Reps00}, a popular approach to the above tasks is to consider \emph{parameterized underapproximations}. Here, for some parameter, one considers only executions of the MPDA (resp.\ multi-threaded program) where that parameter is bounded by some $k\ge 1$. Imposing such bounds then often makes reachability decidable. An example is \emph{context-bounding}, where one considers only those executions that switch between stacks (i.e.\ threads) at most $k$ times~\cite{QadeerR05}. This turned out to be a remarkably useful approach for finding bugs in practice, since empirically, buggy behavior often manifests in executions where such parameters are small~\cite{QadeerR05,MusuvathiQadeer,InversoTFTP22}.
This general approach inspired numerous underapproximations of MPDA. Each comes with its own tailor-made algorithms and results on expressiveness (i.e.~which behaviors can be captured by the underapproximation). Examples include context-bounded~\cite{QadeerR05}, ordered~\cite{DBLP:journals/ijfcs/BreveglieriCCC96}, phase-bounded~\cite{TorreMP07}, and scope-bounded~\cite{DBLP:conf/atva/AtigBKS12} MPDA.

\myparagraph{Bounded treewidth} Given this plethora of underapproximations, an important contribution was the discovery of a powerful unifying approach: A seminal paper by Madhusudan and Parlato~\cite{TreeAuxMadhu}, together with a line of follow-up work~\cite{CyriacGK12, AiswaryaGK14,Cyriac14,LaTorreSNM14} revealed that most of the decidable underapproximations of MPDA (including the ones mentioned above) are \emph{bounded-treewidth}~\cite{QadeerR05, TorreMP07, TorreNP20, DBLP:journals/ijfcs/BreveglieriCCC96}. This means, the executions can be represented as graphs that have small tree decompositions. This implies that decidability of these underapproximations can be uniformly explained by Courcelle's Theorem~\cite{Courcelle89}, which says that among structures of bounded treewidth, satisfiability of any property definable in monadic second-order logic (MSO) is decidable in polynomial time. 

In fact, the bounded-treewidth framework applies not only to MPDA, but can be viewed as a general framework that yields analysis algorithms in infinite-state systems: It also applies to systems with FIFO queues~\cite{TreeAuxMadhu, AiswaryaGK14} and even timed multi-pushdown systems~\cite{AkshayGK18, AkshayGKS17, AkshayGK16}. For this reason, for infinite-state systems (whether they are MPDA or not) whose valid runs can be represented as bounded-treewidth graphs, we informally use the term \emph{bounded-treewidth systems}.

\myparagraph{Analyzing traces: Downward closures}
However, while Courcelle's theorem yields algorithms for yes-or-no queries (such as safety), it is of limited use for closer analysis of its set of program traces (i.e.\ sequences of actions), also called its \emph{(word) language}.
An example task for which the above decision procedures fall short is the
computation of downward closures.  The \emph{downward closure} of a language
$L\subseteq\Sigma^*$ is the set of all (not necessarily contiguous)
subsequences of members of $L$. By a fundamental result of
Higman~\cite{Higman1952OrderingBD}, the downward closure of \emph{every}
language is regular. For example, the non-regular language $L=\{a^nb^n \mid
n\ge 0\}$ has the regular downward closure $\{a^nb^m\mid m,n\ge 0\}$.
Unfortunately, computing an NFA for the downward closure of a language $L$
(when given some representation of $L$) is a notoriously challenging task:
Algorithms to compute downward closures (especially optimal ones) are usually
non-trivial~\cite{DBLP:conf/lics/ClementePSW16,DBLP:conf/icalp/Zetzsche16,DBLP:conf/icalp/BarozziniCCP20,DBLP:conf/icalp/HabermehlMW10,DBLP:conf/icalp/Zetzsche15}
and require deep insights into the structure of the set of runs of a
system.

Downward closure computations are useful for analyzing systems that consist
of several components, which are infinite-state systems themselves. In such
situations, one can often replace each component by its (finite-state) downward
closure, without affecting important properties such as safety. For example,
\emph{asynchronous programs}~\cite{DBLP:journals/toplas/GantyM12} or, more
generally, \emph{concurrent programs}~\cite{DBLP:conf/icalp/0001GMTZ23},
consist of threads that are modeled by infinite-state systems. Moreover, in these models,
threads can be spawned dynamically during the run. It was observed recently
that any algorithm for computing downward closures for threads, can be used to
decide safety, boundedness, and termination of the entire asynchronous
program~\cite[Section 6]{DBLP:journals/lmcs/MajumdarTZ22} or concurrent
programs~\cite{DBLP:conf/icalp/0001GMTZ23}. Another example is paramerized
verification of shared memory systems with non-atomic reads and
writes~\cite{DBLP:conf/fsttcs/Hague11}: As shown by LaTorre, Muscholl, and
Walukiewicz~\cite{DBLP:conf/concur/TorreMW15}, any algorithm for computing downward closures of the leader process yields a procedure for safety verification.

\myparagraph{Downward closures and bounded treewidth} Therefore, an algorithm for computing downward closures of bounded treewidth
systems would directly yield procedures for verifying various complex systems with
bounded-treewidth components.

Computing downward closures of bounded-treewidth systems is far from
understood. There is an abstract result implying the existence of an
algorithm\footnote{One can deduce this from \cite{DBLP:conf/icalp/Zetzsche15} because the languages of bounded-treewidth systems have semilinear Parikh
images~\cite{TreeAuxMadhu}.}~\cite{DBLP:conf/icalp/Zetzsche15}, but that
algorithm comes with no complexity upper bounds (not even beyond
primitive-recursive): It enumerates potential finite automata and then checks
which one is the downward closure.

\myparagraph{Special treewidth (\streewidth)} A slight restriction of ``bounded treewidth'' is to bound the \emph{special treewidth}~\cite{Cou10} (in short \emph{\streewidth}).
Consider graphs of bounded degree, as in particular, systems with stacks or queues lead to graphs of degree $\le 3$. For such graphs, bounded treewidth is equivalent to bounded \streewidth\ (for degree $\le d$, treewidth $\le t$ implies \streewidth\ $\le 20dt$~\cite{Cou10}). In fact, to our knowledge, all bounded-treewidth infinite-state systems studied in the literature are bounded-\streewidth: Even the  underapproximations of the (unbounded-degree) timed multi-pushdown systems~\cite{AkshayGK18, AkshayGKS17, AkshayGK16} have in fact bounded \streewidth. See \refApp{\cref{app:stw}} for a detailed comparison.%

\myparagraph{Contributions} Our \emph{first main result} reveals a connection between bounded-\streewidth\ systems and a concept from computational linguistics: We show that the languages of bounded-\streewidth\ systems are precisely those generated by \emph{multiple context-free grammars} (MCFG)~\cite{SekiMFK91}.  MCFG are an established concept in computational linguistics. Their main motivation is to extend the expressive power of context-free grammars while retaining their algorithmic properties. In this way, they are an example model of \emph{mildly context-sensitive languages}. As we show, the \streewidth\ is reflected in what we call the \emph{width} of the MCFG. In an MCFG, a nonterminal generates a $k$-tuple of words rather than a single word (as for CFG).
 Here, $k$ is the \emph{arity} of that nonterminal, and the maximal arity occurring in an MCFG is its \emph{dimension}, denoted $d$. The \emph{rank}, denoted $r$, of an MCFG is the maximal number of nonterminal occurrences in a production. Then the width of the MCFG is $k=d\cdot r$. Starting from graphs with \streewidth\ $k$, the resulting MCFG has width $2k$. Conversely, from an MCFG of width $k$, we obtain graphs of \streewidth\ $6k$. 

In our \emph{second main result}, we exploit this new connection for the computation of downard closures. We show that for a given MCFG (of fixed dimension), one can construct a downward closure NFA of at most doubly exponential size. We also show that this is optimal.
In particular, this yields optimal downward constructions for bounded-\streewidth\ underapproximations for MPDA.

\myparagraph{Application I: technology transfer} Coincidentally, a recent paper by Conrado, Kjelstr{\o}m, van de Pol, and Pavlogiannis~\cite{MCFGPOPL25} construct a direct, bespoke underapproximation of MPDA (there phrased as interleaved Dyck reachability) in the form of an MCFG. This appears to be a promising approach: The paper provides efficient (and almost conditionally optimal) bounds for language-theoretic problems for MCFG. Moreover, the authors show that empirically, their underapproximation has remarkably high coverage (of executions) in practice. Our results show that the two types of underapproximation allow transferring algorithmic techniques: (i)~the algorithms by Conrado et.\ al.\ (in fact, any algorithms for MCFG) apply to any bounded-\streewidth\ underapproximation
(and thus all bounded-treewidth systems in the literature), but also (ii)~all algorithms for bounded-\streewidth\ systems based on Courcelle's Theorem apply to the languages of Conrado et.\ al. In particular, the underapproximation of Conrado et.\ al.\ can be viewed as part of the bounded-treewidth framework.

\myparagraph{Application II: Concurrent Programs} As mentioned above, downward
closures can be used to verify systems with dynamic spawning.  A \emph{dynamic network of concurrent
pushdown systems} (DCPS)~\cite{AtigBQ11} consists of recursive
processes, modeled as pushdown systems. These processes can access the shared memory of finite domain, but also spawn new
processes. At each moment, there is one active process, which can be
interrupted up to $K$ times. Thus, DCPS are multi-threaded
shared-memory programs with dynamic spawning and $K$-bounded
context-switching. 
Deciding safety of DCPS is complete for
$\TWOEXPSPACE$~\cite{AtigBQ11,DBLP:conf/icalp/BaumannMTZ20} (i.e.\ two-fold
exponential space). 

Our final application concerns DCPS where the individual processes can not only be
recursive programs, but \emph{any bounded-\streewidth\ system}, such as many
underapproximations of MPDA, but also queue systems, etc. The resulting systems
are then \emph{Concurrent Programs over MCFG} (CPM), since by our first
contribution, the individual processes can be represented by MCFG.  Here, our
result is that for CPM, safety verification is still $\TWOEXPSPACE$-complete.
Hence, although the model of CPM permits much more involved individual
processes than DCPS, our result shows that the complexity of safety
verification remains the same. 

To this end, we rely on the idea that for safety in DCPS, one computes downward
closures of context-free languages, which results in exponentially large finite
automata. These automata are then used to construct an exponential-sized vector
addition system with states (VASS), which is checked for coverability. Since
coverability in VASS is $\EXPSPACE$-complete, this results in a $\TWOEXPSPACE$
procedure overall. Applying our aforementioned downward closure result directly 
would thus only lead to a $\THREEEXPSPACE$ upper bound. Therefore, we require
another trick: We introduce the notion of the \emph{partially permuting
downward closure}. For this, we show that (i)~using an additional
\emph{compression step}, one can compute a \emph{singly-exponential} 
automaton, and that (ii)~these new closures can be used in place of the
classical downward closure in safety of concurrent programs.

\myparagraph{Structure of the paper} In \cref{sec:preliminaries}, we will
introduce and explain basic concepts. In \cref{sec:main-results}, we outline
the main results of this paper. Then, \cref{sec:ta2mcfg,sec:mcfg2mso} present
the translations between MCFG and bounded-\streewidth\ systems. In
\cref{sec:downward-closures}, we present the downward closure construction for
MCFG. Finally, \cref{sec:applications} is about the application to concurrent
programs over MCFG. Due to space constraints, many proof details are only available 
in the \refApp{appendix}.

\myparagraph{Related work}
While this work is about \emph{underapproximations} of MPDA reachability, there has also been a significant amount of work on \emph{overapproximations}. Note underapproximating the set of behaviors of a system can be used to find bugs (any undesirable behavior in an underapproximation also exists in the original system), whereas overapproximations can certify the absence of bugs (a bug-free overapproximation means the original system is correct). Prominent examples include bidirected interleaved Dyck reachability~\cite{DBLP:journals/pacmpl/LiSZ22,DBLP:journals/pacmpl/LiZR21,DBLP:journals/pacmpl/KrishnaLPT24,DBLP:conf/icalp/GanardiMPSZ22,DBLP:journals/pacmpl/KjelstromP22}, LCL-reachability~\cite{DBLP:conf/popl/ZhangS17}, overapproximations based on integer linear programming~\cite{DBLP:journals/pacmpl/LiZR23}, synchronized pushdown reachability~\cite{DBLP:journals/pacmpl/SpathAB19}, and techniques for refining overapproximations~\cite{DBLP:journals/sttt/ConradoP25,DBLP:conf/sas/DingZ23}. See \cite{DBLP:journals/siglog/Pavlogiannis22} for a survey on algorithmic aspects.

Another application of treewidth in static analysis of recursive programs is to exploit the empirical observation that many programs in practice have control-flow graphs of small treewidth~\cite{DBLP:journals/iandc/Thorup98,DBLP:conf/atva/ChatterjeeGP17}. This motivates the study of parameterized complexity of static analysis problems under bounded treewidth~\cite{DBLP:conf/esop/ChatterjeeGIP20,DBLP:conf/esa/ChatterjeeIP16,DBLP:conf/popl/ChatterjeeIPG15,DBLP:journals/toplas/ChatterjeeIGP18,DBLP:conf/vmcai/GoharshadyZ23}. However, this approach is fundamentally different from bounded treewidth underapproximations:
The latter achieve decidability by restricting the \emph{search space}, i.e.\ the set of \emph{program executions}, to those with bounded treewidth. 
 In the parameterized complexity setting, however, one studies how much more efficient the (already decidable) analysis becomes if the \emph{input program}, specifically its control-flow graph, has small treewidth.

\section{Preliminaries}\label{sec:preliminaries}

\newcommand{\arity}{\textsf{arity}}
\newcommand{\rel}{R}
\newcommand{\Relof}[2]{\mathcal{R}_{#2}(#1)}
\newcommand{\atomicTermsof}[2]{#1(#2)}
\newcommand{\rfsubset}{\subseteq_{\textsf{rf}}}
\newcommand{\tname}[1]{{\color{red}\small \ensuremath{\tau_{#1}}}}
\newcommand{\openone}{\ensuremath{a}}
\newcommand{\closeone}{\ensuremath{\bar{a}}}
\newcommand{\opentwo}{\ensuremath{b}}
\newcommand{\closetwo}{\ensuremath{\bar{b}}}
\newcommand{\transopenone}{\tname{1}}
\newcommand{\transcloseone}{\tname{2}}
\newcommand{\transopentwo}{\tname{3}}
\newcommand{\transclosetwo}{\tname{4}}
\myparagraph{Multi-pushdown automata and languages}\label{sec:mpds}
A \emph{multi-pushdown automaton} (\mpds) with $\numstack$ stacks, over a finite alphabet $\Sigma$ is a tuple of the form %
$\MPDS = \MPDStuple $, where %
$\states$ is a finite set of control states, $\stackalp$ is a finite stack alphabet, $\initstate \in \states$ is the initial state, $\finstates \in \states$ is the final state and $\trans \subseteq \states \times \alp \cup \{ \varepsilon \} \times \mpdsops\numstack \times \states$, where $\mpdsops\numstack = \{\mpdspush{i}{a},\mpdspop{i}{a} \mid i \in \stackset, a \in \stackalp\} \cup \{ \noop\}$. It accepts a word $w \in \Sigma^\ast$  if there is a run on $w$ that ends in the final state with all stacks empty. 

\begin{example}\label{example:mpds-dyck}
Recall that the Dyck-1 language $D_1\subseteq\{a,\bar{a}\}^*$ consists of all ``well-bracketed'' words, i.e.\ words $w$ with $|w|_a=|w|_{\bar{a}}$ and every prefix $u$ of $w$ admits $|u|_a\ge |u|_{\bar{a}}$.
Consider the \mpds\ with $2$ stacks depicted in Figure~\ref{fig:mpds-dyck} (left). It accepts all interleavings of two Dyck-1 languages, one over the alphabet $\{\openone, \closeone\}$ and the other over $\{\opentwo, \closetwo\}$.
For instance, it accepts the word $\opentwo \openone \openone \opentwo \closeone \closetwo
\opentwo \openone \closetwo \closeone
\openone \opentwo \closeone \closetwo \closetwo \closeone
$. An accepting run on this word is graphically represented in Fig~\ref{fig:mpds-dyck} (right),  where %
the matching push-pop relations are depicted by curved arrows. The matching relation of stack 1 (resp. stack 2) is shown by an upward (resp. downward) curve. 
\end{example}

\begin{figure}
	\scalebox{0.9}{
		\begin{tikzpicture}[node distance=2.5cm, initial text=,>=latex, thick]
			\node[ draw, rounded corners,  initial, accepting ] (q0) {$q$};
			
			\draw[->] 
			(q0) edge [loop right, min distance =1cm, out=30, in = 330] node[right] { $\begin{array}{l}
					\transopenone: \openone, \mpdspush{1}{\openone}\\
					\transcloseone: \closeone, \mpdspop{1}{ \openone}\\   
					\transopentwo: \opentwo, \mpdspush{2}{\opentwo }\\
					\transclosetwo: \closetwo, \mpdspop{2}{ \opentwo}\\   
				\end{array}$ }
				(q0)
			;
		\end{tikzpicture}
	}
	\hfil
	\scalebox{.9}{	\rungraphletter
	}
	\caption{A 2-stack MPDA $\MPDS_\textsf{int-Dyck}$ on the left and an accepting rungraph of  $\opentwo \openone \openone \opentwo \closeone \closetwo
		\opentwo \openone \closetwo \closeone
		\openone \opentwo \closeone \closetwo \closetwo \closeone
		$ on the right.  }\label{fig:mpds-dyck}
\end{figure}

\noindent	

A \emph{rungraph} is a structure $\rungraph = \tuple{\vertices, \Succ,\Match,\vlabel}$  where   $\vertices = \{v_1, \cdots, v_n\}$ is a finite set of vertices, $\Succ$ is the successor relation of a total order on $\vertices$,  $\Match \subseteq \vertices \times \vertices$ is a binary matching relation and the  labelling $\vlabel: \vertices \mapsto \Sigma$  maps  vertices  to letters. %
A rungraph is accepted by an MPDA if there is a labelling  $\tmap: \vertices \mapsto \trans$ that is consistent with $\vlabel$, $\Succ$ and $\Match$.  Consistency of $\tmap$ with $\Match$ also requires that the last-in-first-out access policy of stacks is respected (well-nesting). See \refApp{\cref{app:MPDS-examples}} for details.  Further,  the first (resp. last) event must be labelled by an initial (resp. final) transition. 

The \emph{(run)graph language} of an \mpds\ $\MPDS$, denoted $\glangof{\MPDS}$, is  the  set of all  rungraphs that are accepted by $\MPDS$. 
For a rungraph $\rungraph =  \tuple{\vertices, \Succ,\Match,\vlabel}$ let $\word(\rungraph)$ be the word in $\Alphabet^\ast$ given by the $\Succ$ relation  and $\vlabel$. 
The \emph{word language }of an MPDA, denoted $\wlangof{\MPDS}$ or simply $\Lang(\MPDS)$,  is given by $ \Lang(\MPDS) = \{ w  \mid  w = \word(\rungraph) $ for some $\rungraph \in \glangof{\MPDS}\}$. %

\myparagraph{Multiple context-free grammars and languages} Multiple context-free grammars (MCFG) generalize context-free grammars (CFG) by allowing each nonterminal to derive a tuple of words, rather than a single word. A production $A\to BC$ in a CFG tells us that if we take any word $u$ derivable by $B$, and any word $v$ derivable by $C$, then $uv$ is considered derivable by $A$. Thus, we combine words derivable by $B,C$ to obtain words derivable by $A$. In an MCFG, we do the same for tuples. A rule is of the form 
\begin{equation}
	 A(s_1,\ldots,s_n)\leftarrow \{B_1(x_{1,1},\ldots,x_{1,n_1}),\ldots,B_m(x_{m,1},\ldots,x_{m,n_m})\}. \tag{Prod} \label{eqn:production-rule-form}
\end{equation}
 Here, the $x_{i,j}$ are variables, and $s_1,\ldots,s_n$ are words over both variables and terminal letters. Such a rule tells us that if the tuple $(u_{i,1},\ldots,u_{i,n_i})$ is derivable by $B_i$ for each $i$, then the tuple $(s'_1,\ldots,s'_n)$ is derivable by $A$. Here, $s'_j$ is obtained from $s_j$ by replacing each occurrence of $x_{i,\ell}$ by $u_{i,\ell}$. Importantly, every variable can occur at most once on each side; and every variable that occurs on the left also occurs on the right. This means, the words in the tuples $(u_{i,1},\ldots,u_{i,n_i})$ can be rearranged (or forgotten), but not repeated. 
 The number $n$ is the \emph{arity} of $A$, and $m$ is the \emph{rank} of the rule~(\ref{eqn:production-rule-form}). 

Formally, an MCFG is a tuple $\grammar= (\Terminals, \hat\Nonterminals, \Variables, \Productions, \StartNonterminal)$ where $\Terminals$ is a finite alphabet, $\hat\Nonterminals = (\Nonterminals, \arity)$ is a finite set of nonterminals disjoint from $\Terminals$ along with their arities given by $\arity: \Nonterminals \to \Nat$,    $\Variables$ is a finite set of variables disjoint from $\Nonterminals$ and $\Terminals$,  $\Productions$ is a finite set of \emph{productions} of the form (\ref{eqn:production-rule-form}) and ${\StartNonterminal} \in \Nonterminals$ is the start nonterminal with $\arity(\StartNonterminal) =1$. 

The \emph{dimension}, denoted $d$, of an MCFG $\grammar= (\Terminals, \hat\Nonterminals, \Variables, \Productions, \StartNonterminal)$ is the maximum arity of the nonterminals. The \emph{rank}, denoted $r$, of $\grammar$ is the maximum rank of its production rules; but we define it as $1$ if all rules have rank $0$. The \emph{width} of the MCFG is defined as $k=d\cdot r$. A $d$-MCFG is an MCFG of dimension $\le d$; and a $d$-MCFG($r$) is one that also has rank $\le r$.  Note that the $1$-MCFG are precisely the context-free grammars (CFG). 
 
Next we define the language of an MCFG. First, a nonterminal $A$ of arity $n$ generates a set of  $n$-tuples of words, denoted by $\rel(A)$. We also write $A(u_1,\dots, u_n)$ to mean $(u_1,\dots, u_n) \in \rel(A)$. We define $\rel(A) \subseteq (\Sigma^*)^n$  inductively by the following deduction rules. \\
 
\noindent \begin{minipage}{.33\textwidth}
\begin{equation}
	\frac{A(s_1,\ldots,s_n)\leftarrow \emptyset }{A(s_1,\ldots,s_n)} \label{eqn:base-case-deduction-rule}
\end{equation}
 \end{minipage}
 \hfill
\begin{minipage}{.57\textwidth}
\begin{equation}
	\frac{(\text{\ref{eqn:production-rule-form}}) \qquad B_i(u_{i,1},\ldots,u_{i,n_i})\,  \forall i\in \set{m}}{A(\eta(s_1), \dots, \eta(s_n))} \label{eqn: inductive-step-deduction-rule}\\
\end{equation}
\end{minipage}\\

Note that in the deduction rule~(\ref{eqn:base-case-deduction-rule}), each $s_i \in \Terminals^\ast$. In the deduction rule~(\ref{eqn: inductive-step-deduction-rule}) $\eta$ is a homorphism $\eta: (\Variables \cup \Terminals) \to \Terminals^\ast$ given by $\eta(x_{i,j}) = u_{i,j}$ and $\eta(a) = a $ for $a \in \Terminals$. 

Notice also that if $(u_1, \dots u_n) \in \rel(A)$ then the deduction rules naturally give  a \emph{derivation tree}. These notions are more rigorously explained in \refApp{\cref{appendix:mcfg}}.
\smallskip

The \emph{language} of an MCFG $\grammar= (\Terminals, \hat\Nonterminals, \Variables, \Productions, \StartNonterminal)$, denoted $\langof{\grammar}$, is just $\langof{\grammar} = \rel(\StartNonterminal)$.
A language $L \subseteq \Terminals^\ast$ is a \emph{multiple context-free language (MCFL)} if there is an MCFG $\grammar$ with $L = \langof{\grammar}$.

\begin{example} \label{ex:mcfg-parsetree}
	Consider the \mcfg\  $\initgrammar = (\{a,b\}, \hat\Nonterminals, \Variables, \Productions, \StartNonterminal)$ where $\Nonterminals = \{A, B, {\StartNonterminal} \} $ with $\arity(A) =2, \arity(B) =3, \arity(\StartNonterminal) = 1$, $\Variables = \{x_1, \dots, x_5\}$ and  $\Productions$ listed below. 	Figure~\ref{fig:parsetree} shows a parse tree with yield 
	$baaabb$.\\ 
	\noindent	\begin{minipage}{0.60\textwidth}
		\begin{align*}
			\StartNonterminal( x_2x_1 ) &\leftarrow \{A(x_1,x_2)\} \\
			A( x_2x_5x_4, x_3x_1) &\leftarrow \{ A(x_1,x_2), B(x_3,x_4,x_5)\} \\
			A(a, aa) & \leftarrow \emptyset \\
			B(b, b,b) & \leftarrow \emptyset
		\end{align*}
	\end{minipage}\hfill
	\begin{minipage}{0.35\textwidth}
		\scalebox{1}{
			\parsetreetikz
		}
		
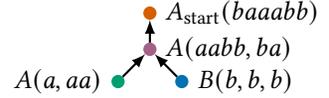
\captionof{figure}{A parse tree}\label{fig:parsetree}
	\end{minipage}\\ 
\end{example}

\myparagraph{About MCFL} MCFL are considered a promising candidate for a class of \emph{mildly context-sensitive languages} as postulated by Joshi~\cite{Joshi1985}. One aspect that makes them appealing is that there are several characterizations of MCFL by different formalisms, such as %
linear context-free rewriting systems~\cite[Lemma 2.2(f3)]{SekiMFK91} (see also the comment on \cite[p.~193]{SekiMFK91}), 
hyperedge replacement grammars and deterministic tree-walking transducers~\cite{DBLP:journals/jcss/EngelfrietH91,DBLP:conf/acl/Weir92}, 
minimalist grammars~\cite{DBLP:conf/lacl/Michaelis01,DBLP:conf/lacl/Harkema01}, and 
multi-component tree-adjoining grammars~\cite{DBLP:conf/acl/Vijay-ShankerWJ87}.
In terms of complexity, the membership problem\label{preliminaries:complexity-membership} for MCFG is $\EXPTIME$-complete in general~\cite{DBLP:journals/ci/KajiNSK94,kaji1992a}, $\NP$-complete for fixed dimension $d\ge 2$~\cite{DBLP:journals/ci/KajiNSK94,kaji1992b}, and $\PTIME$-complete if both dimension and rank are fixed~\cite{DBLP:journals/ci/KajiNSK94,kaji1992b}. The membership problem was (among other problems) also recently studied by Conrado, Kjelstr{\o}m, van de Pol, and Pavlogiannis~\cite{MCFGPOPL25}, who provide close-to-optimal fine-grained complexity bounds.

Regarding their relationship with other language classes studied in verification, the MCFL are incomparable with the \emph{indexed languages}~\cite{aho1968indexed}, which are exactly the languages accepted by second-order pushdown automata~\cite[Theorem~2.2]{parchmann1980deterministic}. First, the indexed languages are not included in the MCFL, since all MCFL have semilinear Parikh images~\cite{DBLP:conf/acl/Vijay-ShankerWJ87} (see also \cite[Theorem 3.7]{SekiMFK91}), whereas there are indexed languages, such as $\{a^{2^n} \mid n\in\Nat\}$, which have a non-semilinear Parikh image. Conversely, the MCFL are not included in the indexed languages: As observed by Kanazawa and Salvati~\cite{DBLP:conf/lata/KanazawaS10}, the ``triple copying theorem'' by Engelfriet and Skyum~\cite[Theorem~2]{DBLP:journals/ipl/EngelfrietS76} implies that for every $K\subseteq\Sigma^*$, if the language $\{w\#w\#w\mid w\in K\}$ is indexed, then $K$ must be an EDT0L language (see \cite{DBLP:journals/jcss/EngelfrietRS80} for a definition).
Now, for example the Dyck-1 language $D_1\subseteq\{a,\bar{a}\}^*$ (see \cref{example:mpds-dyck}) is not EDT0L~\cite[Theorem 5]{ehrenfeucht1977some}, but context-free.
Because of the latter, we can easily construct an \mcfg{} for the language $\{w\#w\#w\mid w\in D_1\}$, which then cannot be indexed.

\myparagraph{Downward closures}
	Given two words $u,v \in \Sigma^*$ for some finite alphabet $\Sigma$, we say $u$ is a \emph{subword} of $v$, denoted as $u \subword v$ if $u$ can be obtained by dropping some letters from $v$, that is, if $u = u_1\cdots u_n$ and $v=v_0u_1v_1\cdots u_nv_n$ for some words $v_0,u_1,\ldots,u_n,v_n\in\Sigma^*$.
For a language $\lang \subseteq \Sigma^*$, its \emph{downward closure} $\downc\lang$ is the set of all subwords of words in $L$, in symbols:
$ \downc\lang := \{u\in\Sigma^* \mid \exists v \in \lang\colon u \subword v \} $.

\myparagraph{Graphs and MSO definable classes of graphs}
Let $\Alphabet$ and $\Elabels$ be alphabets. An \emph{$\LOG$ over $(\Alphabet,\Elabels)$} is a directed graph with an underlying linear order, where (i)~the vertices carry labels in $\Alphabet\cup\{\varepsilon\}$ and (ii)~edge relations have names in $\Elabels$. Formally, it is a tuple $\lograph = \tuple{\vertices, \sedge, (\edges_\gamma)_{\gamma \in \Elabels}, \vlabel}$ where (i)~$\vertices$~is a finite set of vertices, (ii)~$\edges_\gamma \subseteq \vertices \times \vertices$ is an  edge relation (directed) for each $\gamma \in \Elabels$, (iii)~$\sedge \subseteq \vertices \times \vertices$ is the successor relation of a linear order on the vertices, and (iv)~$\vlabel : \vertices \to \Alphabet \cup \{\varepsilon\}$ is a labeling function. The set of all $\LOG$s over $(\Alphabet, \Elabels)$ is denoted $\logset{\Alphabet}{\Elabels}$.

 The labels of the sequence of vertices visited according to the linear order of an \LOG\ $\lograph = \tuple{\vertices, \sedge,  (\edges_\gamma)_{\gamma \in \Elabels}, \vlabel} $ naturally give a word in $\Alphabet^\ast$. We call this word $\word(G)$. That is, if $v_1 , v_2  \dots  v_n$ is the sequence of all vertices in $\vertices$ according to the linear order (i.e., $(v_i, v_{i+1}) \in \sedge$ for all $1 \le i < n $), then $\word(G) = \vlabel(v_1)\vlabel(v_2)\dots\vlabel(v_n)$. 
 
 Let $S $ be a set of $\LOG$s over $ (\Alphabet, \Elabels)$. We define $\word(S) = \{ \word(\lograph) \mid \lograph \in S\}$.
 \medskip
 
 \textit{Monadic second-order logic} (MSO) on directed node-labelled graph with an underlying linear order over a finite alphabet $\Alphabet$, denoted $\MSO(\Alphabet, \Elabels)$, is defined in the standard way. An $\MSO(\Alphabet, \Elabels)$ formula is given by the grammar
 \[\phi = a(x) \mid x \in X \mid \sedge(x, y) \mid  \edges_\gamma(x,y)  \mid \neg \phi \mid \phi \vee \phi \mid \exists x \phi \mid \exists X \phi \]
 where $a \in \Alphabet$, $\gamma \in \Elabels$, $x, y$ are first-order variables ranging over vertices and $X$ is a second-order variable ranging over subsets of vertices. The semantics is as expected. 

 Every $\MSO(\Alphabet, \Elabels)$ sentence $\phi$  defines a set of $\LOG$s over $(\Alphabet, \Elabels)$, denoted by $\glangof{\phi}$. It also defines a set of words over $\Alphabet$, denoted $\wlangof{\phi}$, given by \[\wlangof{\phi} = \word(\glangof{\phi}) =  \{\word(G) \mid G \in \glangof{\phi}\}.\]

 One can show the following by declaring a second-order variable for each transition of an \mpds:
 
 \begin{claim}
 	For every \mpds\ $\MPDS$ over the alphabet $\Sigma$, there is an $\MSO(\Alphabet, \{\Match\})$-formula $\phi_\MPDS$ such that $\glangof{\phi_\MPDS} = \glangof{\MPDS}$. 
 \end{claim}
 
 \begin{example}Consider the class of all  directed grids over the node labels $\{a,b\}$. Linear order is given by rows in the successive order. There is an MSO sentence $\phigrid \in \MSO(\{a,b\}, \{\rightarrow, \downarrow\})$ for this. Consider the set of grids in which the node-label is the same in every column. This set is now definable by the following sentence. $\phicopy = \phigrid \wedge \forall x \forall y (\edges_\downarrow(x,y) \implies ((a(x) \wedge a(y)) \vee (b(x) \wedge b(y)))$. Note that $\wlangof{\phicopy} = \{
u^k \mid k \ge 1, u \in \{a,b\}^\ast\}$.
\end{example}

 \newcommand{\cpclog}{cp-\LOG}
  \newcommand{\cpcgraph}{\hat\lograph}
  \newcommand{\UOP}{\textsf{UOP}}
    \newcommand{\BOP}{\textsf{BOP}}
    \newcommand{\ecpcgraph}{\bot}

\myparagraph{Bounded (special) treewidth languages}
A \emph{tree decomposition} of a (directed) graph is a rooted tree where each tree node is labelled by a bag of  graph vertices such that 1) each graph vertex appears in at least one bag, 2) the bags containing a graph vertex form a connected subtree, and 3) every adjacent pair of vertices of the graph must belong together in at least one bag.

A tree decomposition is \emph{special} if the condition 2 above is strengthened as follows: $2^\prime$) the bags containing a graph vertex form a path along a branch of the rooted tree. In particular, this forbids two sibling bags to share a graph vertex. The width of a special tree decomposition is the maximum size of its bags. The \emph{special treewidth} (\streewidth) of a graph is the minimum width over all special tree decompositions. 

\begin{example}\label{example:tree-decomposition} A tree decomposition of the rungraph in Figure~\ref{fig:mpds-dyck} is given in Figure~\ref{fig:tree-decom}. It is a special tree decomposition, its width is $4$ and hence the special treewidth of this rungraph is at most $4$. 
\end{example}

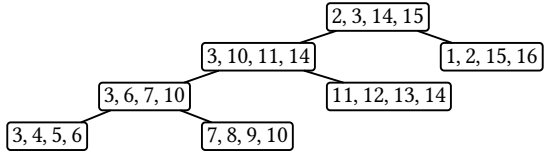
\begin{figure}
	\scalebox{.9}{	\begin{tikzpicture}[node distance=2.5mm, initial text=,>=latex, thick,	bagnode/.style ={draw ,rounded corners=0.5mm,inner sep=2pt, outer sep = 0pt, minimum size=10pt}]
			
			\node[bagnode, ] (root) {$2, 3, 14, 15$};
			\node[bagnode, below right=of root] (r) {$1, 2, 15, 16$};
			\node[bagnode, below left=of root] (l) {$3, 10, 11, 14$};
			\node[bagnode, below right=of l] (lr) {$11, 12, 13, 14$};
			\node[bagnode, below left=of l] (ll) {$3, 6, 7, 10$};		
			\node[bagnode, below right=of ll] (llr) {$7,8, 9, 10$};
			\node[bagnode, below left=of ll] (lll) {$3, 4, 5, 6$};												
			
			\draw (root) edge (l) 
			(root) edge (r)
			(l) edge (ll)
			(l) edge (lr)
			(ll) edge (lll)
			(ll) edge (llr);
			
		\end{tikzpicture}
	}
	\caption{A special tree decomposition of the rungraph in Figure~\ref{fig:mpds-dyck}.  }\label{fig:tree-decom}
\end{figure}

Given an $\MSO(\Alphabet, \Elabels)$ sentence $\phi$ and a bound $k$ on special treewidth, we define the $k$\emph{-bounded language of $\phi$ }as follows:  $\bglangof{k}{\phi}$ the set of all $\LOG$s that satisfy $\phi$ and have \streewidth\ $\le k$. Accordingly, we have  
$\bwlangof{k}{\phi} =  \word(\bglangof{k}{\phi})$. A language $L\subseteq\Sigma^*$ is an \emph{MSO-definable bounded-\streewidth\
	language} if $L=\bwlangof{k}{\phi}$ for some $\MSO(\Sigma,\Gamma)$-formula $\phi$ for some $\Gamma$ and $k\in\Nat$. 
	
	Similarly, we define the $k$-bounded rungraph language and word language of an MPDA as well:  $\bglangof{k}{\MPDS}$ be  the  set of all  rungraphs of the MPDA $\MPDS$ with \streewidth\ at most $k$, $\bwlangof{k}{\MPDS} =  \word(\bglangof{k}{\MPDS})$.

\myparagraph{Tree automata} A width-$k$ tree decomposition of a graph can be
viewed as a tree with finitely many node labels: for each possible bag, we use
a separate node label. It is well known that the sets of trees that correspond
to an MSO-definable bounded-\streewidth\ set of graphs are precisely those
recognized by \emph{finite (bottom-up) tree automata}.  Essentially, such a
tree automaton works like a deterministic finite automaton on words: Initially
all leaves are assigned an initial state, and then the state of an inner node
$v$ is determined by the states of $v$'s children states and by $v$'s label.
Finally, a tree is accepted if the root node is assigned a final state.
For now informally, we will call a deterministic bottom-up tree automaton that
accepts width-$k$ special tree decompositions a \emph{width-$k$ tree automaton}.
This informal definition should be enough to understand our results; a formal
definition will follow in \cref{sec:ta2mcfg}.

	\begin{fact}\cite{Courcelle89,RabinTree} \label{fact:msotoTA}
	Given an  $\MSO(\Alphabet, \Elabels)$ sentence $\phi$ and a bound $k\in \Nat$, there is a width-$k$ tree automaton $\TA$ such that $\glangof{\TA} = \bglangof{k}{\phi}$.
	\end{fact}

\section{Main results}\label{sec:main-results}
Our first main result is the characterization of MSO-definable bounded-\streewidth\ languages by MCFL:
\begin{theorem}\label{main-equivalence}
The MSO-definable bounded-\streewidth\ languages are precisely the MCFLs.
\end{theorem}
\cref{main-equivalence} will follow from \cref{thm:TA2MCFL,thm:MCFL2MSO}, which provide complexity bounds. 
First, for the translation from an MSO-definable bounded-\streewidth\ language to an MCFG, we need to translate an MSO formula to an MCFG. However, in algorithmic settings, MSO-definable bounded-\streewidth\ languages are usually represented by tree automata \textcolor{blue}{(cf. Fact~\ref{fact:msotoTA})}, which have much better algorithmic properties (e.g.\ emptiness is decidable in polynomial time, compared to non-elementary complexity for MSO). 
Hence, the following implies that MSO-definable bounded-\streewidth\ languages are MCFL:
\begin{theorem}\label{thm:TA2MCFL}
	Fix $k\ge 0$. Given a width-$k$ tree automaton $\TA$, we can construct a $k$-\mcfg(2)
for $\wlangof{\TA}$ in time $|\TA|\cdot 2^{O(k\log k)}$. 
\end{theorem}
Hence, the resulting MCFG is of size $|\TA|\cdot 2^{O(k\log k)}$ and width $2k$.
For the converse translation:
\begin{theorem}\label{thm:MCFL2MSO}
	Given a $d$-$\mathrm{MCFG}(r)$ $\initgrammar$, we can construct an MSO
	formula $\phi$ with $\bwlangof{6dr}{\phi}=\langof{\initgrammar}$.
\end{theorem}
In particular, starting from an MCFG of width $k$, we obtain a special treewidth bound of $6k$.
Since \cref{thm:MCFL2MSO} is not used in our applications (and rather shows
that \cref{main-equivalence} is an exact characterization), 
complexity is not crucial, so we present that translation using MSO.

\myparagraph{Downward closures of MCFL}
\Cref{thm:TA2MCFL} implies that a procedure for downward closure computation of MCFL yields such a procedure for MSO-definable bounded-\streewidth\ languages. This motivates:
\begin{restatable}{theorem}{mcfgdownwarddoubleexp}\label{thm:dcimcfg}
	Fix $d\ge 1$. Given a $d$-MCFG $\initgrammar$, we can compute an NFA  for $\downc{\langof{\initgrammar}}$ in time $2^{2^{\poly(|\initgrammar|)}}$.
\end{restatable}
In particular, the resulting NFA will be of size $2^{2^{\poly(|\initgrammar|)}}$.
Note that here, the dimension is fixed, but the rank of $\initgrammar$ is part of the input.
We also show that the doubly exponential upper bound in \cref{thm:dcimcfg} is optimal, by constructing a $\twoicfg(2)$ $\initgrammar_n$ of size linear in $n$ for the language $L_n=\{ww \mid w\in\{a,b\}^{2^n}\}$ (see \refApp{\cref{app:MCFL-lowerbound}} for details):
\begin{theorem}\label{claim:LBmcfl}
	There is a family of $\twoicfg(2)$ $(\initgrammar_n)_{n \in \Nat}$ with $|\initgrammar_n| = \mathcal{O}(n)$ such that any NFA for $\downc{\langof{\initgrammar_n}}$ has at least $2^{2^n}$ states. 
\end{theorem}
		
\myparagraph{Multi-pushdown automata}
An important application domain of our results concerns bounded-treewidth underapproximations of MPDA: By imposing a treewidth bound on the rungraphs of MPDA, one obtains an MSO-definable bounded-\streewidth\ language. One can thus use \cref{thm:TA2MCFL} and \cref{thm:dcimcfg} to compute the downward closure of a bounded-treewidth underapproximation of a given MPDA (this is the setting where our later applications will be used). This raises the question of whether the translation MPDA $\leadsto$ width-$k$ tree automaton $\leadsto$ MCFG $\leadsto$ downward closure NFA introduces an unnecessary blow-up.
We will now see that this is not the case, provided that we fix the \streewidth\ (equivalently, the treewidth~\cite{Cou10}) and the number $\numstack$ of stacks. First, as shown in \cite{TreeAuxMadhu, CyriacGK12, Cyriac14, AiswaryaGK14,AkshayGK18}, given an MPDA, we can efficiently construct a width-$k$ tree automaton: 
\begin{fact} \label{fact:MPDS2TA} %
Given an MPDA $\MPDS$ with $\numstack$ stacks and a bound $k$, one can construct a width-$k$ tree automaton $\TA$ for $\bglangof{k}{\MPDS}$ in time $|\MPDS|^k \times 2^{\text{poly}(\numstack, k)}$.
\end{fact}
			
From \cref{fact:MPDS2TA},  \cref{thm:TA2MCFL} and \cref{thm:dcimcfg}, we can conclude a doubly exponential upper bound for downward closures of MSO-definable bounded-treewidth MPDA languages:
\begin{corollary}\label{cor:mpds-down}
Fix $\numstack$ and $k$. Given an  MPDA $\MPDS$ with $\numstack$ stacks, we can compute in doubly exponential time an NFA for $\downc{\bglangof{k}{\MPDS}}$. 
\end{corollary}
The upper bound in \cref{cor:mpds-down} is optimal (see \refApp{\cref{app:mpda-lower-bound}}): 
\begin{theorem}\label{claim:LBMPDS}
There is a bound $k$ and a family of $2$-stack MPDA $(\MPDS_n)_{n \in \Nat}$ with $|\MPDS_n| = \mathcal{O}(n)$ and \streewidth\ at most $k$ such that any NFA for the downward closure $\downc{\Lang(\MPDS_n)}$ has at least $2^{2^n}$ states. 
\end{theorem}

\myparagraph{Concurrent Programs}
As mentioned in \cref{sec:introduction}, an important underapproximation of MPDAs is that of $k$-bounded context switching: A run is allowed to switch between stacks only $k$ times.
In the same section we also mention that this can be lifted to dynamic pushdown systems:
Stacks can dynamically spawn during execution, and each stack may switch $k$ times.
Finally we lift this even further so that each dynamically spawned process is not a mere pushdown,
but itself an MPDA underapproximation, or some other MSO-definable bounded-\streewidth\ system (however, we still require that each process as a whole, is only interrupted at most $k$ times).
With our previous results, we assume all processes have been transformed into $\mcfg$, that
have fixed dimension $d$ and rank $r$.
This gives rise to \emph{concurrent programs} ($\CP$) over $d$-\mcfg$(r)$.
With some modifications to our downward closure computation, we obtain an optimal
complexity upper bound for $k$-context-bounded safety:

\begin{theorem}
  Fix $d,r \geq 0$. The following problem is $\TWOEXPSPACE$-complete: Given a $\CP$ over the $d$-\mcfg$(r)$, one of its global states $q_f$, and a unary encoded $k \geq 1$, is $q_f$ reachable under $k$-bounded context switching?
\end{theorem}

The lower bound already follows from the case of $\CP$ over the CFG for a fixed $k \geq 1$ \cite{DBLP:conf/icalp/0001GMTZ23}.

\section{Bounded-special-treewidth languages are MCFL}\label{sec:ta2mcfg}
Let us give an overview of the proof of \cref{thm:TA2MCFL}. Details can be found in \refApp{\cref{app:ta2mcfg}}. We begin with a formal definition of width-$k$ tree automata, which were only sketched in \cref{sec:preliminaries}. This definition relies on graph algebras~\cite{Cou10}, which allow us to express the graphs of bounded special treewidth by terms (or trees) in the algebra.
	
	\myparagraph{A graph algebra for bounded-stw graphs \cite{Cou10}}
	We need some terminology. We first define  \emph{colored piecewise linearly ordered} graphs (\cpclog{}s).
  
  A structure $(V, \sedge)$ is called \textit{piecewise linearly ordered} if $V$ and $\sedge$ can be partitioned into a number of parts $V = \uplus_i V_i$ and $\sedge = \uplus_i S_i$ such that each part $S_i$ is the successor relation of a total order on $V_i$. Each part $(V_i, S_i)$ is called a piece. That is, each piece is totally ordered by $\sedge$, but there is no order between the pieces. 
  
  A \cpclog\ does not require the vertices to be totally ordered by $\sedge$ like in an $\LOG$. Instead they only need to be piecewise linearly ordered.  Further, some of the vertices of a \cpclog\ are uniquely colored. In particular we require the extreme vertices of every piece to be colored. 

  \newcommand{\remove}[1]{}
  \newcommand{\colorless}[1]{{\scalebox{1.2}{$\diamondsuit$}}(#1)}
   Formally, a \emph{\cpclog} is a tuple  $\cpcgraph = \tuple{\vertices, \sedge,(\edges_\gamma)_{\gamma \in\Elabels},\vlabel, \colorlab}$, where  $\sedge \subseteq \sedge'$ for some $\sedge'$ such that  $\tuple{\vertices, \sedge',(\edges_\gamma)_{\gamma \in\Elabels},\vlabel}$ is  a \LOG.   A maximal $\sedge$-connected component is called a \textit{piece}.  The coloring $\colorlab: \vertices \to \colors$ is a partial injective function from vertices to colors, and must keep colored the  vertices  that either have no incoming $\sedge$ edge or no outgoing $\sedge$ edge. 

	Given a  \cpclog\  $\cpcgraph = \tuple{\vertices, \sedge,(\edges_\gamma)_{\gamma \in\Elabels},\vlabel, \colorlab}$, we denote by $\colorless{\cpcgraph}$ the (colorless) piecewise-\LOG\ $\tuple{\vertices, \sedge,(\edges_\gamma)_{\gamma \in\Elabels},\vlabel}$ where the last entry $\colorlab$ of  $\cpcgraph$ is omitted. Note that an \LOG\ is also a  piecewise-\LOG\ with only one piece.
	
	We will define some operations to iteratively construct {\cpclog}s. These operations can produce structures $\tuple{\vertices, \sedge,(\edges_\gamma)_{\gamma \in\Elabels},\vlabel, \colorlab}$ which need not be \cpclog{}s in a strict sense. For instance, $\sedge$ can be an arbitray edge relation instead of the successor relation of a piecewise linear order, or the extreme vertices of a piece may be uncolored. However, even the {relaxed~}{\cpclog}s insist that coloring $\colorlab: \vertices \to \colors$ is a partial injective function from vertices to colors. 
	We call such structures \emph{relaxed \cpclog}.
\remove{	A \emph{relaxed \cpclog} relaxes the requirement that $\sedge$ form a partial linear order and that unconnected vertices are colored.}
	
	  First we have a set of unary operations $\UOP$ and a single binary operation $\join(\cdot)$. The set of unary operators is $ \UOP = \{\node{c}{a} \mid c \in \colors, a \in\Alphabet \cup \{\varepsilon\} \} \cup \{\addsedge{c_1}{c_2} \mid c_1, c_2 \in \colors\} \cup \{\addedge{\gamma}{c_1}{c_2} \mid c_1, c_2 \in \colors, \gamma \in \Elabels\} \cup \{\fcolor{c} \mid c \in \colors\}$. For each $\op \in \UOP$, we define $\op(  \tuple{\vertices, \sedge,(\edges_\gamma)_{\gamma \in\Elabels},\vlabel, \colorlab}) =  \tuple{\vertices', \sedge', (\edges'_\gamma)_{\gamma \in\Elabels},\vlabel',  \colorlab'}$ as follows.
	\begin{itemize}
		\item $\op \equiv \node{c}{a}(\cdot)$ : Requires that $c \not\in \image{\colorlab}$. The resulting graph contains a new vertex $v$ colored $c$ and labelled $a$. Notice that $a$ may also be $\varepsilon$.  Note that $v$ will be a piece by itself. That is, $\vertices' = \vertices \uplus \{v\}$, $\sedge' = \sedge$, $\edges_\gamma' = \edges_\gamma$ for each ${\gamma \in\Elabels}$, $\vlabel' = \vlabel \uplus \{(v, a)\}$, $\colorlab' = \colorlab \uplus \{(v, c)\}$.
		
		\item $\op\equiv \addsedge{c_1}{c_2}(\cdot)$: Adds a $\sedge$ edge from the unique $c_1$-colored vertex $u$ to the unique $c_2$-colored vertex $v$. That is, $\vertices' = \vertices$,   $\edges_\gamma' = \edges_\gamma$ for each ${\gamma \in\Elabels}$,  $\vlabel' = \vlabel $, $\colorlab' = \colorlab $, $\sedge' = \sedge \uplus \{(u,v) \mid \colorlab(u) = c_1, \colorlab(v) = c_2\}$.
		
		\item $\op \equiv \addedge{\gamma}{c_1}{c_2}(\cdot)$: Adds a $\gamma$-labelled edge from the unique $c_1$-colored vertex $u$ to the unique $c_2$-colored vertex $v$. That is, $\vertices' = \vertices$,  $\sedge' = \sedge$,  $\vlabel' = \vlabel $, $\colorlab' = \colorlab $,  $\edges_\gamma' = \edges_\gamma \uplus \{(u,v) \mid \colorlab(u) = c_1, \colorlab(v) = c_2\}$, $\edges'_{\gamma'} = \edges_{\gamma}$ for each ${\gamma' \in \Elabels \setminus \{\gamma\}}$.
		
		\item $\op \equiv \fcolor{c}(\cdot)$: This frees the color $c$. That is, $\vertices' = \vertices$,   $\sedge' = \sedge$, $\edges_\gamma' = \edges_\gamma$ for each ${\gamma \in\Elabels}$, $\vlabel' = \vlabel $,  $\colorlab' = \colorlab \setminus \{(u,c) \mid \colorlab(u) = c\} $. 
	\end{itemize}
	The binary  operation $\join$ takes two relaxed~{\cpclog}s with disjoint colors, and constructs their disjoint union.
	For $i \in \{1,2\}$, let $\cpcgraph_i = \big(\vertices_i, \sedge_i,(\edges^i_\gamma)_{\gamma \in \Elabels},\vlabel_i, \colorlab_i\big)$.
	Their join is $\cpcgraph =$ $\big(\vertices,$ $\sedge,$ $(\edges_\gamma)_{\gamma \in\Elabels},$ $\vlabel,$ $\colorlab\big)$ defined with $\vertices = \vertices_1 \uplus \vertices_2$, $\sedge = \sedge_1 \uplus \sedge_2$, $\edges_{\gamma} = \edges^1_{\gamma} \uplus \edges^2_{\gamma}$ for each $\gamma \in \Elabels$, $\vlabel = \vlabel_1 \uplus \vlabel_2$, and $\colorlab = \colorlab_1 \uplus \colorlab_2$.
	
	An empty \cpclog\ is one with no vertices. It is denoted $\ecpcgraph$. We are interested in structures obtained by  (repeated) application of the above operations  starting from $\ecprg$. This can be described by expressions  given by the  following grammar.
	\begin{align*}
		\psi :=&\,\, \ecpcgraph \mid \node{c}{a}(\psi) \mid \addsedge{c_1}{c_2} (\psi) \,\,	\mid \addedge{\gamma}{c_1}{c_2}(\psi) \mid \fcolor{c}(\psi) \mid \join(\psi, \psi)
	\end{align*}
	where $c, c_1, c_2 \in \colors$, $a \in \Alphabet \cup \{\epsilon\}$ and $\gamma \in \Elabels$. The set of all expressions generated by the above grammar is denoted $\expr(\colors, \Alphabet, \Elabels)$. The structure defined by an expression $\psi$ is denoted $\cpcgraph_\psi$. This is defined inductively as expected. The special treewidth of $\cpcgraph_\psi$ is bounded by $\sizeof{\colors}$:
	\begin{remark}
		Note that the term $\psi$ is a special tree decomposition of $\cpcgraph_\psi$. Indeed, if a node $x$ is the root of subterm $\psi'$ of $\psi$, then the bag at $x$ will be the colored vertices of the  graph $ \cpcgraph_{\psi'}$. 
	\end{remark}
	
	The \emph{special treewidth} (\streewidth) of a \cpclog\ $\cpcgraph$ is defined to be $m$ where  $ m = \inf \big\{\, \sizeof{\colors} ~\big|~ \cpcgraph = \cpcgraph_\psi$ for some $\psi \in \expr(\colors,\Alphabet, \Elabels) \,\big\}$. Note that due to the $\fcolor{\cdot}$-operation, $m$ may be larger than the final set of colors used in $\cpcgraph$.

	\myparagraph{Width-$k$ tree automata} A \emph{width-$k$ tree automaton} is a deterministic bottom-up tree automaton over expressions from  $\expr(\set{k}, \Alphabet, \Elabels)$, i.e., with the color set $\colors = \set{k}$.   It  is given by a tuple $\TA = \tuple{Q, \delta, q_\bot, Q_\text{accept}}$ where $Q$ is a finite set of states, $Q_\text{accept} \subseteq Q$ is the set of accepting states and  $\delta =  \bigcup_{\op \in \UOP}\delta_\op \cup \delta_\join $ is the set of transition functions:  $\delta_\op  : Q  \to Q$ for all $\op \in \UOP$, and $\delta_\join : Q \times Q  \to Q$. 
	
	A deterministic bottom-up tree automaton  $\TA = \tuple{Q, \delta, q_\bot, Q_\text{accept}}$ defines a (partial) function $\TArun{\TA}: \expr(C, \Alphabet, \Elabels) \to Q$, inductively as follows
	\begin{itemize}
	\item $\TArun{\TA}(\bot)=  q_\bot$
	\item if $\psi_1 \equiv \op(\psi_2)$ for some $\op \in \UOP$, then $\TArun{\TA}(\psi_1) = \delta_\op (\TArun{\TA}(\psi_2))$
	\item if $\psi_1 \equiv \join(\psi_2, \psi_3)$, then $\TArun{\TA}(\psi_1) = \delta_\join (\TArun{\TA}(\psi_2), \TArun{\TA}(\psi_3))$
	\end{itemize}
	
	The language of a tree automaton $\TA = \tuple{Q, \delta, q_\bot, Q_\text{accept}}$, denoted $\langof{\TA}$, is defined as $\langof{\TA} = \{\psi  \mid  \TArun{\TA}(\psi)  \in  Q_\text{accept} \}$.
	We also define $\glangof{\TA}$ as the set of structures defined by $\langof{\TA}$ that are {\LOG}s, $\glangof{\TA} = \{\colorless{\cpcgraph_\psi} \mid \psi \in \langof{\TA} \} \cap \logset{\Alphabet}{\Elabels}$, and  $\wlangof{\TA}$ as $\wlangof{\TA} = \{\word(\lograph) \mid \lograph \in \glangof{\TA} \}$.
  Recall that $\colorless{\cdot}$ removes the coloring $\colorlab$.

\myparagraph{Challenges}
Consider a tree automaton $\TA = \tuple{Q, \delta^\TA, q_\bot, Q_\text{accept}}$ over $\expr(C, \Alphabet, \Elabels)$.  We defined $\langof{\TA} = \{\psi  \mid  \TArun{\TA}(\psi)  \in  Q_\text{accept} \}$. Indeed for any state $q \in Q$, we can define  $\langof{\TA, q} = \{\psi  \mid  \TArun{\TA}(\psi)  = q \}$.  Let $\pglangof{\TA} = \{\cpcgraph_\psi \mid \psi \in \langof{\TA} \} $. Note that $\pglangof{\TA}$ is different from $\glangof{\TA}$ as we are not intersecting it with $ \logset{\Alphabet}{\Elabels}$.  In particular, $\cpcgraph_\psi$ may be a relaxed~{\cpclog}, not a strict~{\cpclog}, since the $\sedge$-edges can branch or form cycles. On the other hand, if $\cpcgraph$ is a strict~{\cpclog}, $\sedge$ still need not induce a total order but only a piecewise  linear order. Here we extend the definition of $\word(\cdot)$ to $\word(\cpcgraph)$, but now it defines a multiset of words, with one word corresponding to each \remove{$\sedge$-connected component}piece. 

In our \mcfg, we will have a nonterminal $A_q$ for each state $q$ of the tree automaton. For $\psi \in\langof{\TA, q} $, if $\cpcgraph_\psi$ is a  \cpclog, then  the nonterminal  $A_q$  should be able to generate $\word(\cpcgraph_\psi)$. There are several challenges: 1)~For any $\psi \in\langof{\TA, q}$,  we need to check whether $\cpcgraph_\psi$ is a  \cpclog. 2)~For a $\cpcgraph_\psi$ that is a \cpclog,  $\word(\cpcgraph_\psi)$ is a multiset of words, but $\langof{A_q}$ generates a tuple of words. Hence, we need to order the words in the multiset, so as to correspond to the tuple. 3)~For $\psi_1, \psi_2 \in \langof{\TA, q}$, the number of pieces in $\cpcgraph_{\psi_1}$ may be different from that of $\cpcgraph_{\psi_2}$. Hence fixing an arity for the nonterminal $A_q$ is problematic. One possible solution is to have copies of $A_q$ one of each possible arity. 

\myparagraph{Solution}
Since there are only a fixed number of colors in any expression, we can address these issues by maintaining a finite \emph{summary} for each expression.  The summary contains the start and end colors of each piece (note that they need to be colored, if a $\sedge$-edge is to be added in the future). This already has the information about the number of pieces (cf.\ challenge 3 above). We assume an ordering between the colors, and order the pieces based on their start colors (cf.\ challenge 2).  Storing the start and end colors will also allow us to forbid  cycles or adding multiple successors/predecessors with $\sedge$-edges (cf.\ challenge 1).

These summaries can be computed by a tree automaton. We then consider the Cartesian product $\TB$ of the tree automaton $\TA$ and the summary-computing automaton. The states in the product can be thought of as annotating the states of $\TA$ with the summary information. 

\myparagraph{Syntactic translation} Now $\TB$ translates synctactically into an \mcfg. Indeed, the nonterminals are the states, with the number of pieces in the summary as the arity. The transitions translate to productions. For a $\join$-transition $(B,C) \mapsto A $ of $\TB$, we will have a production of the form $A(s_1, \dots s_\ell) \mapsfrom \{B(x_1,\dots, x_m), C(y_1, \dots y_n)\}$ with $\{s_1, \dots, s_\ell\} = \{x_1, \dots, x_m\} \cup \{y_1, \dots, y_m\}$.   Corresponding to a $\node{c}{a}$-transition $B \mapsto A$, there will be a production of the form $A(s_1, \dots s_\ell) \mapsfrom \{B(x_1,\dots, x_m)\}$  where $\ell = m+1$. Every $s_i$ is some $x_j$, except for one which is the letter $a$. The nonterminal $A$ keeps track of the fact that this $a$ is a piece by itself and is colored $c$.  Corresponding to a $\addsedge{c_1}{c_2}$-transition, there will be a production of the form $A(s_1, \dots s_\ell) \mapsfrom \{B(x_1,\dots, x_m)\}$  where $\ell = m-1$. Every $s_i$ is some $x_j$, except for one which is now a concatenation of the form $x_k x_{k'}$. The indices $k, k'$ are indeed determined by the summary kept at $B$: The end vertex of $x_k$ is colored $c_1$ and the start vertex of $x_{k'}$ is colored $c_2$. The nonterminal $A$ keeps track of the fact that this new piece gets the pair $(c_{3}, c_{4})$ of colors if $c_3$ was coloring the start vertex of $x_k$ and $c_4$ was coloring the end vertex of $x_{k'}$. Corresponding to an $\addedge{\gamma}{c_1}{c_2}$-transition there will be a production of the form $A(x_1, \dots, x_m) \mapsfrom \{B(x_1,\dots, x_m)\}$ as the pieces are not changed, only the states are updated. 

The ordering of the arguments $s_1 , \dots, s_\ell$ follows the ordering we fixed on the colors: Each piece $s_i$ is summarized by a pair of colors of their start and end vertices. The pieces are ordered according to their start-vertex color. For details, see \refApp{\cref{app:ta2mcfg}}.

\section{MCFL are MSO-definable bounded-\streewidth\  languages}\label{sec:mcfg2mso}

Towards \cref{thm:MCFL2MSO}, we show that every MCFL can be described by an MSO sentence over words with an additional binary relation $\matching$ that we call matching. The matching relation $\matching$ is, in  a sense, a generalisation of the nesting/matching relation for context-free languages\cite{LautemannST94,AlurM06}. We call these \LOG\ over $(\Terminals, \{\matching\})$ \emph{nested words}. 

\myparagraph{Nested words of an MCFG} Consider an \mcfg\ $\initgrammar=(\Terminals, \hat\Nonterminals, \Variables, \Productions, \StartNonterminal)$ and a word $ w \in \langof{\initgrammar}$. For the construction, we assume that $\initgrammar$ has the ``information-lossless'' property, meaning for each production \ref{eqn:production-rule-form} (see p.~\labelcpageref{eqn:production-rule-form}), every variable that occurs on the right also occurs on the left. This can be achieved using a translation due to Seki, Matsumura, Fuji, and Kasami~\cite[Lemma 2.2(f3)]{SekiMFK91}.

Suppose the parse tree for $w$ has a subtree for $A (u_1, u_2,\dots, u_k)$ where $u_i \in \Terminals^\ast$ for $i \in \set{k}$. Then we can write $ w$ as $ w = w_0 u_{i_1} w_1 u_{i_2} \dots u_{i_k} w_k$ such that $\{i_1, i_2 , \dots i_k \} = \{1, \dots, k\}$. That is, the entries $u_1, \dots , u_k$ in the tuple generated by $A$ eventually appear as non-overlapping infixes in the final word. We would like to mark the endpoints of these infixes and connect them, indicating they were related together as one tuple in the derivation tree. 
Towards this, we will introduce extra vertices in $w$, labelled $\epsilon$ as they should not change the word $w$, bordering each 
entry $u_i$. We then connect these extra vertices using the additional $\matching$ edge. This would give us an $\LOG$ $\lograph_w$ with $\word(\lograph_w) = w$ and  $\lograph_w$ would correspond to a parsetree-structure of $w$. See \refApp{\cref{app:NestedWord}} for details. We call such $\LOG$s \emph{nested words} of $\initgrammar$.

	\begin{figure}
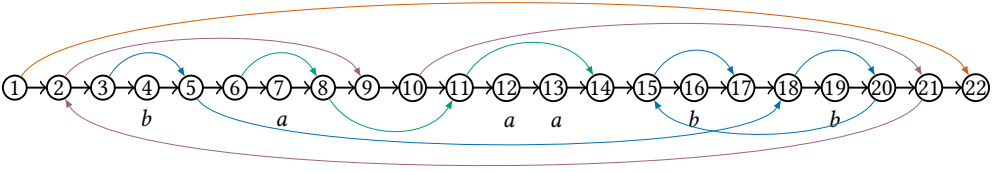

	\scalebox{.9}{
		\nestedwordtikz
	}
	\caption{The nested-word corresponding to the parse tree in Fig~\ref{fig:parsetree} generating  $baaabb$.}\label{fig:nestedword}
\end{figure}
	
	\begin{example} 
	Consider the parse tree  in Figure~\ref{fig:parsetree}, generating the word 
	$baaabb$ (cf. Example~\ref{ex:mcfg-parsetree}).  The graph with $\matching$ relation corresponding to this parsetree is depicted in Figure~\ref{fig:nestedword}. 
	The label of a node, if not specified, is $\epsilon$. The $\epsilon$-labelled nodes participate in the $\matching$ relation, and the $\Alphabet$ labelled nodes do not. There are $4$ different maximal $\matching$-connected paths in this graph, each colored differently. Each of these paths correspond to one node of the parse-tree.  We can notice the well-nesting: The  open interval $(1,22)$ contains $\tuple{(10,21), (2,9)}$ which in turn contains $\tuple{(6,8), (11,14)}$ and $\tuple{(3,5), (18,20), (15,17)}$.  Note that every maximal $\matching$-path corresponds to a node in the parsetree. 
	The number of nodes in a $\matching$-connected path is exactly twice the arity of the respective nonterminal. 
	\end{example}
	
	Note that,  every $\LOG$ over $(\Terminals, \{\matching\})$ is not necessarily a nested word of some \mcfg\ $\initgrammar$. The purpose of this section is to show that, for every $d$-\mcfg(r) $\initgrammar$  we can construct an $\MSO(\Terminals, \{\matching\})$ sentence $\phi$ such that $\langof{\phi} = \langof{\initgrammar}$.  The sentence $\phi$ has two parts. The first part (called the well-nesting property) is independent of $\initgrammar$, but depends on the dimension $d$. The second part assumes well-nesting and states that it corresponds to a valid parsetree of $\initgrammar$.

	We will use $\le$ to mean the reflexive transitive closure of $\sedge$. Note that this is expressible in $\MSO(\sedge)$. Once we assume the definition of $\le$, the rest of the formula can be written in \textit{existential} $\MSO$ (\textsf{EMSO}).	 We will first declare some more shorthands. We will also write $pEq$ instead of $(p,q) \in E$, and $p {\rightarrow} q$ instead of $\Succ(p,q)$. Further we write $p_1 E_1 p_2 E_2 p_3$ to mean $p_1 E_1 p_2 \wedge p_2 E_2 p_3$.
	
\myparagraph{Well-nesting matching relation $\matching$} 
 Recall that  $\maxarity$ is the dimension, i.e.\ the maximum arity among nonterminals of the \mcfg. The matching relation $\matching$ must satisfy the following properties.
\begin{enumerate}
	\item The length (i.e.\ the number of nodes) of any $\matching$-connected path is at most $2\maxarity$.
	\item The length of any maximal  $\matching$-connected path is even. 
	\item A node can have at most one $\matching$-successor and at most one $\matching$-predecessor.%
	\item %
	A vertex $u$ participates in the matching relation $\matching$ iff it is labelled $\epsilon$.
	\item Let $v_1 \matching v_2 \matching \dots \matching v_{2k}$ be a maximal $\matching$-connected path. Then
	\begin{enumerate}
		\item the matching edges at odd positions (which enclose an entry) must agree with the linear order. That is, for all $i: 1 \le i \le k$, $v_{2i-1}  \le v_{2i}$. The matching edges at even positions have no restrictions, and are used to represent the right-neighbor entry in the tuple.
		\item  entries do not intersect. That is, for  $i,j: 1 \le i, j \le k, i \neq j$, $v_{2j} < v_{2i-1}  $ or $v_{2i} < v_{2j-1}$. 
	\end{enumerate} 
	\item If a nesting path (representing a tuple $t$) intersects an entry (corresponding to a tuple~$t'$) then the tuple $t$ must be completely contained inside tuple $t'$. Formally, if  $v_1 \matching v_2 \matching \dots \matching v_{2k}$ and $u_1 \matching u_2 \matching \dots \matching u_{2k'}$ both are maximal $\matching$-connected paths and if ($u_{2i-1} < v_{2j} < u_{2i}$ or $u_{2i-1} < v_{2j-1} < u_{2i}$) for some $i: 1 \le i \le k'$ and $j: 1 \le j \le k$, then for all $j: 1 \le j \le k$ there exists $i: 1 \le i \le k'$  such that $u_{2i-1} < v_{2j-1} < v_{2j} <  u_{2i}$. 
\end{enumerate}
A matching relation $\matching$ is called \emph{well-nesting} if it satisfies all the above properties. There is sentence $\phiwellnesting$ in $\MSO(\Sigma,\{\matching\})$  (in fact in $\textsf{FO}(\Sigma, <, \matching )$) which asserts that the relation $\matching$ is well-nesting. See \refApp{\cref{app:MSOwellnesting}} for the concrete formula. 

\myparagraph{Ensuring a valid parsetree}
We assume  well-nesting of $\matching$.  %
Given an  \mcfg\ $\initgrammar= (\Terminals, \Nonterminals, \Variables, \Productions, \StartNonterminal)$, our $\MSO(\Alphabet, \{\matching\})$ formula
asserts the following. 
\begin{enumerate}
	\item There is a partition of the $\epsilon$-labelled vertices into $|\Nonterminals|$ many parts $(Y_A)_{A \in \Nonterminals}$. All vertices in a $\matching$-connected path will belong to the same part. 
\end{enumerate}
 Hence it makes sense to talk about $A$-paths, for $A\in \Nonterminals$.
\begin{enumerate}[resume]
	\item The length of a maximal $\matching$-connected path  labeled $X_A$ (called an $A$-path) will be $2\times \arity(A)$.
	\item For each nonterminal $A$ with $\arity(A) = n$, we require the following condition. 
		 For all $A$-paths $y_1^b \matching y_1^e \matching y_2^b \matching y_2^e \dots y_n^b \matching y_n^e$ (i.e.,  we use the first-order variable $y_i^b$ to indicate the left boundary vertex of the $i$th entry, and $y_i^e$ to indicate the right boundary vertex) the following conditions must hold for at least one production rule $\prodrule$.  For a production rule~$\prodrule$ of the form ~(\ref{eqn:production-rule-form}) we require the existence of the boundaries for the arguments on the RHS, 
	$\exists	x_{1,1}^b, x_{1,1}^e, \dots x_{1,n_1}^b,x_{1,n_1}^e, \dots, x_{m,1}^b, x_{m,1}^e, \dots x_{m,n_m}^b,x_{m,n_m}^e $, such that 
		\begin{enumerate}
			\item for all $i: 1 \le i \le n$, the the path from $y_i^b$ to $y_i^e$ that takes the nesting edge $x_{j,j'}^b \matching x_{j,j'}^e$ as a shortcut whenever possible (that is, this path corresponds to the topmost level  within the nesting edge  $y_i^b \matching y_i^e$ in our pictorial depictions), should be identical to   $s_i$ after the following slight modification:  for all $j: 1 \le j \le m$, for all $j': 1 \le j' \le n_j$, if $x_{j,j'}$ occurs in $s_i$ it  is replaced with $x_{j,j'}^b \matching x_{j,j'}^e$. For example, if $s_4 = bx_{2,1}abx_{1,3}$ in the production rule ~(\ref{eqn:production-rule-form})  then we will have the following as one of the conjuncts:
			$ \exists z_1 \exists z_2 \exists z_3 (
			b(z_1) \wedge a(z_2) \wedge b(z_3) \wedge			
			y_4^b  {\rightarrow} z_1 {\rightarrow} x_{2,1}^b \matching x_{2,1}^e {\rightarrow} z_2 {\rightarrow} z_3 {\rightarrow} x_{1,3}^b \matching x_{1,3}^e {\rightarrow} y_4^e		
\remove{			\Succ(y_4^b, z_1)  \wedge \Succ(z_1, x_{2,1}^b) \wedge \Succ(x_{2,1}^e, z_2) \wedge \Succ(z_2, z_3) \wedge \Succ(z_3, x_{1,3}^b) \wedge \Succ(x_{1,3}^e, y_4^e)}
			)
			 $.

			\item for all $j: 1 \le j \le m$, the path $x_{j,1}^b\matching x_{j,1}^e \matching x_{j,2}^b \dots \matching x_{j,n_j}^e$ is a maximal $B_j$-path. 
		\end{enumerate}
	\item The entire nested-word is wrapped within an $\StartNonterminal$-path.
	\end{enumerate}

All of the above conditions are expressible in the existential fragment of $MSO(\Sigma,\{\matching\})$, see \refApp{\cref{app:MSOparsetree}} for details. Call that sentence $\phi_\initgrammar $. 
Our final formula would then be $$\phi \equiv \phiwellnesting \wedge\phi_\initgrammar. $$

The correctness of the above formula, as well as the bound on special treewidth stated in \cref{thm:MCFL2MSO} are proved in \refApp{\cref{appendix:MSOcorreectness}}.

\section{Downward closures of MCFGs}\label{sec:downward-closures}
\newcommand{\accelL}[1]{\Sigma_{#1,\mathrm{left}}}
\newcommand{\accelR}[1]{\Sigma_{#1,\mathrm{right}}}
\newcommand{\tikzscale}{0.8}
In this section, we prove \cref{thm:dcimcfg}. For an intuition, 
we first give a rough description of Courcelle's downward closure construction for context-free grammars~\cite{DBLP:journals/eatcs/Courcelle91} (to be precise, its formulation in \cite{DBLP:conf/concur/AnandZ23}).
Imagine we are provided with a grammar recognizing the language $\setof{(ab)^n \# (ba)^n \mid n \in \Nat}$:
\begin{align*}
	S &\to A ~|~ \# & A &\to aBa &	B &\to bSb
\end{align*}
We can see that by repeating the derivation sequence $S \Rightarrow A \Rightarrow aBa \Rightarrow abSba$, 
we can produce an arbitrary number of $ab$ to the left of an $S$, and an arbitrary 
number of $ba$ to the right. Since the downward closure can drop arbitrary letters, 
we can furthermore ignore the ordering of letters and the equal number of 
repetitions on either side. In other words, a sufficient number of repetitions 
allow us to produce arbitrary words from $\{a,b\}^*$ to the left and right of $S$.

Courcelle's construction can be viewed as introducing \emph{accelerating productions} that 
capture this behavior. In our example, the productions $S \to A$, 
$A \to aBa$ and $B \to bSb$ would justify the introduction of the accelerating production $S \to \{a,b\}^* S \{a,b\}^*$ (see also \cref{fig:courcelle}). These are technically not context-free, but their semantics is clear. More generally, for every non-terminal $A$, we introduce a production $A\to\accelL{A}^*A\accelR{A}^*$, where $\accelL{A}$ (resp.\ $\accelR{A}$) is the set of letters that can appear in $u$ (resp.\ $v$) in derivations $A\Rightarrow^* uAv$.

Such acceleration productions are \emph{complete}, in the sense that one can show that  every derivation tree can be cut down (by replacing pieces of the derivation tree by accelerating productions) to one of bounded height. 
Conversely, they are also \emph{sound}: 
a production of the form $A\to\accelL{A}^*A\accelR{A}^*$ will only produce words in the downward closure.
Soundness and completeness allows us to construct an NFA that just simulates bounded-height trees. This NFA is of exponential size and its language has the same downward closure as the input grammar. One can then easily modify the NFA so as to accept exactly the downward closure.

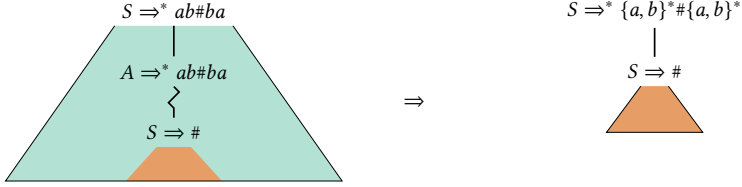
\begin{figure}
	\begin{center}
		\scalebox{\tikzscale}{
		\begin{tikzpicture}
			\begin{scope}
				\node (s) at (0,0) {$S \Rightarrow^* ab\#ba$};
				\node (a) at (0,-1) {$A \Rightarrow^* ab\#ba$};
				\node (a1) at (0,-2) {$S \Rightarrow \#$};
				\draw [-,thick] (s) -- (a);
				\draw [-,thick,decorate,decoration={zigzag}] (a) -- (a1);
				\draw [-] (s.south west) -- (-2.8,-2.8) -- (2.8,-2.8) -- (s.south east);
				\path [on layer=back,fill=highlight-3!30] ($(s.south west)$) -- (-2.8,-2.8) -- (2.8,-2.8) -- ($(s.south east)$);
				\path [on layer=back,fill=highlight-1!60] ($(a1.south)!.5!(a1.south west)$) -- (-.8,-2.8) -- (.8,-2.8) -- ($(a1.south)!.5!(a1.south east)$);
			\end{scope}
			
			\node at (4,-1.5) {$\Rightarrow$};
			\begin{scope}[shift={(8,0)}]
				\node (s) at (0,0) {$S \Rightarrow^* \{a,b\}^*\#\{a,b\}^*$};
				\node (a) at (0,-1) {$S \Rightarrow \#$};
				\draw [-,thick] (s) -- (a);
				\path [on layer=back,draw=black,fill=highlight-1!60] ($(a.south)!.4!(a.south west)$) -- (-.8,-2) -- (.8,-2) -- ($(a.south)!.4!(a.south east)$);
			\end{scope}
		\end{tikzpicture}
		}
	\end{center}
	\caption{Courcelle's construction for CFG: A repeated nonterminal ($S$) gets contracted by replacing the blue subtree with the Kleene closure of a suitable alphabet ($\{a,b\}^*$). The orange subtree is left unchanged.}
	\label{fig:courcelle}
\end{figure}

Ideally, we would like to transfer this approach to the setting of {\mcfg}s. However, there are several reasons that make this a significant challenge.
\begin{enumerate}
	\item \label{it:challenge:subset} First of all, as nonterminals have multiple entries, a repetition of nonterminals may support acceleration only in some entries. As an example, consider the grammar
	\begin{align*}
		A(xy, z) &\leftarrow A(x, y), C(z) &
		A(a, a)  &\leftarrow \emptyset &
		C(c)	 &\leftarrow \emptyset
	\end{align*}
	and the following derivation tree. Note that repetition of $A$s allow an acceleration of $c$s in the first entry, but the second entry is fixed to a single $c$.
	\begin{center}
		\scalebox{\tikzscale}{
		\begin{tikzpicture}
			\node (a1) at (0,0) {$A(aac,c)$};
			\node (a2) at (-1,-1) {$A(aa, c)$};
			\node (c1) at (1, -1) {$C(c)$};
			\node (a3) at (-2,-2) {$A(a, a)$};
			\node (c2) at (0,-2) {$C(c)$};
			\draw [-] (a1) -- (a2);
			\draw [-] (a1) -- (c1);
			\draw [-] (a2) -- (a3);
			\draw [-] (a2) -- (c2);
			
			\node at (2.25,-1) {$\Rightarrow$};
			
			\begin{scope}[shift={(4,-.35)}]
				\node (a1) at (0,0) {$A(aac^*, c)$};
				\node (a2) at (0,-1) {$?$};
				\draw [-] (a1) -- (a2);
			\end{scope}
		\end{tikzpicture}
		}
	\end{center}
	
	\item \label{it:challenge:restriction} Second, even when acceleration is restricted to a subset of the entries, the act of acceleration may restrict possible values in the remaining entries. Consider, for example, if we augment our grammar from the previous point by the productions $\pi : A(x, y) \leftarrow D(x), D(y)$ and $D(d) \leftarrow \emptyset$. In general, this allows the second entry of an $A$ to contain the letter $d$. However, after applying $\pi$, we will not be able to repeat any $A$s. Therefore in any sequence of productions that permits acceleration, the second entry must in the end contain a $c$. We thus must take care to not permit such constructions.
	\begin{center}
		\scalebox{\tikzscale}{
		\begin{tikzpicture}
			\begin{scope}[shift={(0,0)}]
				\node at (0,0.5) {\textcolor{highlight-3}{\cmark}};
				\node (a1) at (0,0) {$A(ddc^*, c)$};
				\node (a2) at (0,-1) {$?$};
				\draw [-] (a1) -- (a2);
				
				\begin{scope}[shift={(2,0)}]
					\node at (0,0.5) {\textcolor{red}{\xmark}};
					\node (a1) at (0,0) {$A(ddc^*, d)$};
					\node (a2) at (0,-1) {$?$};
					\draw [-] (a1) -- (a2);
					
					\begin{scope}[shift={(2,0)}]
						\node at (0,0.5) {\textcolor{red}{\xmark}};
						\node (a1) at (0,0) {$A(aad^*, c)$};
						\node (a2) at (0,-1) {$?$};
						\draw [-] (a1) -- (a2);
						
						\begin{scope}[shift={(2,0)}]
							\node at (0,0.5) {\textcolor{red}{\xmark}};
							\node (a1) at (0,0) {$A(aac^*, d)$};
							\node (a2) at (0,-1) {$?$};
							\draw [-] (a1) -- (a2);
						\end{scope}
					\end{scope}
				\end{scope}
			\end{scope}
		\end{tikzpicture}
		}
	\end{center}
	
	\item \label{it:challenge:shuffle} Lastly, we observe that production rules are allowed to shuffle entries, e.g.\ $\pi : A(x, y) \leftarrow A(y, x)$. Our acceleration procedure needs to account for such productions.
	\begin{center}
		\scalebox{\tikzscale}{
		\begin{tikzpicture}
			\node (a1) at (0,0) {$A(c,aa)$};
			\node (a2) at (0,-1) {$A(aa,c)$};
			\node (a3) at (-1,-2) {$A(a,a)$};
			\node (c1) at (1,-2) {$C(c)$};
			\draw [-] (a1) -- (a2);
			\draw [-] (a2) -- (a3);
			\draw [-] (a2) -- (c1);
			
			\node at (2.25,-1) {$\Rightarrow$};
			
			\begin{scope}[shift={(4,-.35)}]
				\node (a1) at (0,0) {$A(c, aac^*)$};
				\node (a2) at (0,-1) {$?$};
				\draw [-] (a1) -- (a2);
			\end{scope}
		\end{tikzpicture}
		}
	\end{center}
\end{enumerate}

To approach these challenges, we search not only for repeating nonterminals, but also look at their \emph{embedding maps}. These embedding maps track the flow of information through the entries in a derivation tree. For example, in the following derivation tree, the embedding map of $x$ into $y$ is $1 \mapsto 2$ and the embedding map from $x$ into $z$ is $1 \mapsto 1$.

\begin{center}
	\scalebox{\tikzscale}{
	\begin{tikzpicture}
		\node (a1) at (0,0) {$z:A(aac,c)$};
		\node (a2) at (-1,-1) {$y:A(aa, c)$};
		\node (c1) at (1, -1) {$C(c)$};
		\node (a3) at (-2,-2) {$A(a, a)$};
		\node (c2) at (0,-2) {$x:C(c)$};
		\draw [-] (a1) -- (a2);
		\draw [-] (a1) -- (c1);
		\draw [-] (a2) -- (a3);
		\draw [-] (a2) -- (c2);
		\draw [-{latex},bend right,highlight-1] (c2) edge node [right] {$1 \mapsto 2$} (a2);
		\path [-{latex},bend right=70,highlight-1] (c2.east) edge node [right] {$1 \mapsto 1$} (a1.east);
		\path [-{latex},bend left,highlight-2] ($(a2.north west)!.5!(a2.north)$) edge node [left] {$
			\begin{aligned}
				1 &\mapsto 1 \\[-.4em]
				2 &\mapsto 1
			\end{aligned}
		$} (a1.west);
	\end{tikzpicture}
	}
\end{center}

These maps help us identify the acceleratable entries of \cref{it:challenge:subset}: acceleration will happen naturally along the information flow indicated by the map as additional letters join this flow from the sides. Informally, we will call this the \emph{watershed} of the embedding map, and our tree cutting procedure will cut out (part of) the watershed, but keep entries derived \emph{in parallel} to the watershed intact, see \cref{fig:watershed}. We formalize this notion when we give the explicit construction of accelerated grammars in \cref{sec:accelerated-grammars}. 

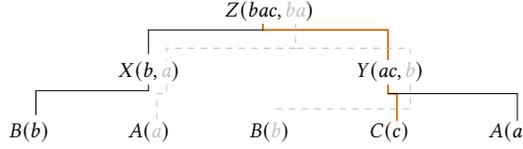
\begin{figure}
	\begin{center}
		\scalebox{\tikzscale}{
		\begin{tikzpicture}[remember picture]
			\node (a1) at (0,0) {$Z(\subnode{a1-1}{bac}, \subnode{a1-2}{\textcolor{lightgray}{ba}})$};
			\node (a2) at (-2,-1) {$X(\subnode{a2-1}{b},\subnode{a2-2}{\textcolor{lightgray}{a}})$};
			\node (a3) at (2,-1) {$Y(\subnode{a3-1}{ac},\subnode{a3-2}{\textcolor{lightgray}{b}})$};
			\node (a4) at (-4,-2) {$B(\subnode{a4-1}{b})$};
			\node (a5) at (-2,-2) {$A(\subnode{a5-1}{\textcolor{lightgray}{a}})$};
			\node (a6) at (0,-2) {$B(\subnode{a6-1}{\textcolor{lightgray}{b}})$};
			\node (a7) at (2,-2) {$C(\subnode{a7-1}{c})$};
			\node (a8) at (4,-2) {$A(\subnode{a8-1}{a})$};
			
			\draw [-,thick,color=highlight-1] (a1-1) -- ++(0,-.35) -| (a3-1);
			\draw [-] (a1-1) ++(0,-.35) -| (a2-1);
			\draw [dashed,lightgray] (a1-2) -- ++(0,-.65) -| (a3-2);
			\draw [dashed,lightgray] (a1-2) ++(0,-.65) -| (a2-2);
			\draw [-] (a2-1) -- ++(0,-.35) -| (a4-1);
			\draw [dashed,lightgray] (a2-2) -- ++(0,-.35) -| (a5-1);
			\draw [-,thick,color=highlight-1] (a3-1) -- ++(0,-.35) -| (a7-1);
			\draw [-] (a3-1) ++(0,-.35) -| (a8-1);
			\draw [dashed,lightgray] (a3-2) -- ++(0,-.65) -| (a6-1);
		\end{tikzpicture}
		}
	\end{center}
	\caption{The embedding map of $C$ into $Z$ (orange, bold), and its watershed (black, solid). The derivation of the second entries of $Z$ (light gray, dashed) is what we consider to occur ``alongside'' or ``in parallel'' to the watershed.}
	\label{fig:watershed}
\end{figure}

As a simplified example of how we use embedding maps, consider \cref{fig:sigma-grammar-visualization}. We identify an embedding map of $\setof{1,2,3} \mapsto 1$ (not shown) between the $A$ nonterminals. Acceleration justifies the removal of parts of the purple and green subtrees where entries are part of the watershed (i.e.\ they themselves embed into the first entry of the root nonterminal). The root nonterminal is also where we introduce our accelerated production, specifically in the first entry. The orange subtree may have non-acceleratable content. To keep it intact, we project the green subtree to the first entry (containing, by the embedding map, all entries of the orange subtree) and lift it below the root $A$.

A particularly interesting instance is the case where the embedding map is \emph{idempotent}, that is $f \circ f = f$. Such maps have the property that there is at least one fixpoint $x$ with $f(x) = x$, and all other points are mapped to one of these fixpoints. This eliminates the shuffling we noted in \cref{it:challenge:shuffle}. We will later see that restricting ourselves to these idempotent embeddings only adds a constant factor to the size of our downward closure automaton.

Our most complicated adjustments come from the dependencies mentioned in \cref{it:challenge:restriction}. We are retaining the parts of the trees that are not part of the watershed. This requires some modification of the underlying grammar, and we need to make sure that this modification only allows us to generate infixes that are generated \emph{alongside} the watershed. For this, we introduce the notion of ``decorated'' nonterminals. A decorated nonterminal behaves almost identical to its plain version from the original grammar, but is augmented slightly. In our case, we are primarily interested in ensuring that a decorated nonterminal produces a derivation tree if and only if the plain version would have produced a derivation tree with a specific embedding.

An additional challenge is that decorated nonterminals will themselves also need decorated variants (so that we can repeatedly cut trees until all paths are short), and so forth. However, since decoration will always strictly decrease arity (decorated variant of $A$ will only simulate non-fixpoint entries of $A$, and an idempotent has at least one fixpoint), and we assume that the dimension of the input \mcfg\ is fixed, we will only need to perform decoration a constant number of times, ensuring a merely polynomial blow-up.

\begin{figure}[t]	
	\begin{center}
		\scalebox{\tikzscale}{
		\begin{tikzpicture}
			\node (root) at (0,-.3) {$A(~u'uwvv'~,~\bullet~,~\bullet~)$};
			\node (b) at (0,-1.5) {$B(~\bullet~,~\bullet~,~\bullet~)$};
			\node (a) at (0,-2.3) {$A(~uwv~,~{\bullet}~,~{\bullet}~)$};
			\node (a2) at (0,-3.5) {$A(~w~,~{\bullet}~,~{\bullet}~)$};
			\draw [-,thick,decorate,decoration={zigzag}] (root) -- (b);
			\draw [-] (b) -- (a);
			\draw [-,thick,decorate,decoration={zigzag}] (a) -- (a2);
			\draw [-] ($(root.south)!.65!(root.south west)$) -- (-2.6,-3.5) -- (a2);
			\draw [-] ($(root.south)!.65!(root.south east)$) -- (2.6,-3.5) -- (a2);
			\path [on layer=back,fill=highlight-2!30] (b.south west) -- (-1.9,-3.5) -- (-1,-3.5) -- ($(a)!.45!(a.west)$) -- ($(a)!.45!(a.east)$) -- (1,-3.5) -- (1.9,-3.5) -- (b.south east);
			\path [on layer=back,fill=highlight-3!30] ($(a)!.45!(a.west)$) -- (-1,-3.5) -- (1, -3.5) -- ($(a)!.45!(a.east)$);
			\draw [-,fill=highlight-1!60] (a2) -- (-0.5,-4.3) -- (0.5, -4.3) -- (a2);
			
			\node at (3,-2) {$\Rightarrow$};
			
			\begin{scope}[shift={(0,-.5)}]
				\node (root) at (8,-.3) {$A(~\Sigma_u^*uwv\Sigma_v^*~,~\bullet~,~\bullet~)$};
				\begin{scope}[shift={(1,0)}]
					\node (b) at (8,-1.5) {$\hat{B}(~\bullet~)$};
					\draw [-,thick,decorate,decoration={zigzag}] (root) -- (b);
					\draw [-] (root.south) -- (6,-2.9) -- (7.5,-2.9) -- ($(b.south) + (-0.2,-0.3)$) -- ($(b.south) + (0.2,-0.3)$) -- (8.5,-2.9) -- (10,-2.9) -- (root.south east);
					\path [on layer=back,fill=highlight-2!30] (b.south west) -- (6.9,-2.9) -- (7.5,-2.9) -- ($(b.south) + (-0.2,-0.3)$) -- ($(b.south) + (0.2,-0.3)$) -- (8.5,-2.9) -- (9.1,-2.9) -- (b.south east);
				\end{scope}
				\node (a3) at (5.4,-.9) {$\hat{A}(~uwv~)$};
				\node (a4) at (5.4,-2.2) {$A(~w~,~{\bullet}~,~{\bullet}~)$};
				\draw [-] (root.south west) -- (a3);
				\draw [-,thick,decorate,decoration={zigzag}] (a3) -- (a4);
				\path [on layer=back, fill=highlight-3!30] ($(a3.south west)!.30!(a3.south)$) -- (4.1,-2.2) -- (6.7,-2.2) -- ($(a3.south east)!.30!(a3.south)$);
				\draw [-] ($(a3.south west)!.30!(a3.south)$) -- (4.1,-2.2) -- (a4);
				\draw [-] ($(a3.south east)!.30!(a3.south)$) -- (6.7,-2.2) -- (a4);
				\draw [-,fill=highlight-1!50] (a4) -- (4.9,-2.9) -- (5.9,-2.9) -- (a4);
			\end{scope}
		\end{tikzpicture}
		}
	\end{center}

	\caption{Downward closure construction for {\mcfg}s. Assume the embedding $f$ between the nonterminals $A$ has a single fixpoint in the first entry (in general, there may be more). The teal subtree rooted at the middle $A$ is moved to the root, but projected to the parts that are relevant for the contents of the fixpoint. In particular, this guarantees that all the (potentially non-acceleratable) entries of the orange subtree are fully contained. The subtree rooted at $B$ is projected to contain only those nodes that do not contribute to the fixpoint. (Our actual construction is slightly more involved, but only to simplify the proof that the tree gets smaller.)}
	\label{fig:sigma-grammar-visualization}
\end{figure}
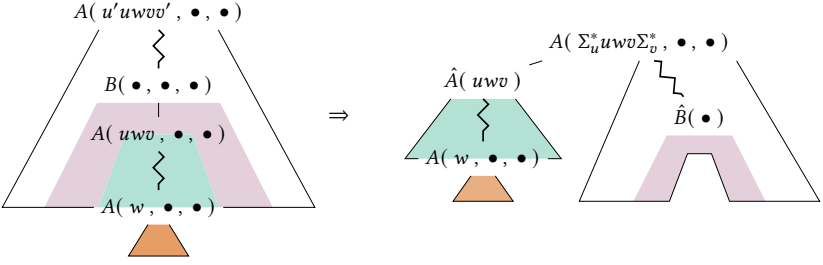

\subsection{Nonterminal Embeddings}
Let us formalize embedding maps now. We will initially define them within individual productions for convenience, though the notion lifts quite naturally to nodes in derivation trees, which is the setting for the maps that we will be primarily concerned with for the remainder of the section.

Let $\pi = A(s_1, \ldots, s_n) \leftarrow S$ be a 
production and $B(x_1, \ldots, x_m) \in S$. We say that $B(\vec x)$ embeds into $A(\vec s)$ with map $f : [1, m] \to [1, n]$, 
written $B(\vec x) \embed_{\pi,f} A(\vec s)$ if $\tgt(x_i) = s_{f(i)}$. Where $\pi$ is clear, we
may omit it. As these embeddings are closed under composition, we can extend 
this notion to derivation trees as expected.

Given two maps $f : [1,m] \to [1,n]$ and $g : [1,\ell] \to [1,n]$ we write $f/g$ for the set of functions $h : [1,m] \to [1,\ell]$ satisfying $f = g \circ h $.

\subsection{Decorating Nonterminals}

As we have mentioned above, we will extend the grammar in order to support the cutting of trees; in order to make sure this extension is not overly broad we ``decorate'' nonterminals in order to derive more restrictive behavior. Decorated nonterminals and productions are syntactically derived from existing ones and largely share the semantics of their original counterparts. For example, one of our decorations will be the projection of nonterminals to a subset of their entries. The other is obtained by annotating an existing nonterminal with a target nonterminal and target embedding, ensuring that the decorated nonterminal is productive if and only if there is a derivation sequence from the original nonterminal leading to the target nonterminal with the specified embedding.

\myparagraph{Projective Decorations} 
We will call the first kind of decoration a \emph{projective decoration}. Given a nonterminal $A$ of arity $n$ and a subset of indices $I = \setof{i_1, \ldots, i_m} \subseteq [1,n]$ we introduce the projected nonterminal $A^I$ of arity $\card I$. The semantics will be such that $A^I \rightarrow^* (w_{i_1}, \ldots, w_{i_m})$ if and only if $A \rightarrow^* (w_1, \ldots, w_n)$. 

For every production $\pi = A(s_1, \ldots, s_n) \leftarrow S$ we introduce a production $\pi^I = A^I(s_{i_1}, \ldots, s_{i_m}) \leftarrow S'$ where $B^J(z_{j_1}, \ldots, z_{j_{m'}}) \in S'$ if and only if $B(z_1, \ldots, z_{n'}) \embed_{\pi,f} A(\vec s)$ and $J = f^{-1}(I) = \setof{j_1, \ldots, j_{m'}}$. We also extend this notion to derivation trees in the expected way: If $t$ is a tree with root $(A, \pi)$ then $t^I$ is the tree with root $(A^I, \pi^I)$ whose children are obtained by the appropriate projection of the children of $t$. As a visual guide, refer to \cref{fig:example-projection}

\begin{figure}
	\begin{center}
		\scalebox{\tikzscale}{
		\begin{tikzpicture}[remember picture]
			\node (tA-1) at (0,0) {$A(\subnode{ta1a1}{a}\subnode{ta1a2}{a}\subnode{ta1b1}{b}, \subnode{ta1b2}{b})$};
			\node (tA-2) at (-1, -2) {$B(\subnode{ta2a1}{a},\subnode{ta2b1}{b})$};
			\node (tA-3) at (1,-2) {$B(\subnode{ta3a1}{a},\subnode{ta3b1}{b})$};
			\draw [-] (ta2a1.north) -- ++(0,1.30) -| (ta1a1.south);
			\draw [-] (ta2b1.north) -- ++(0,.95) -| (ta1b2.south);
			\draw [-] (ta3a1.north) -- ++(0,.25) -| (ta1a2.south);
			\draw [-] (ta3b1.north) -- ++(0,.50) -| (ta1b1.south);

			\node (tB-1) at (5,0) {$A^{\setof{1}}(\subnode{tb1a1}{a}\subnode{tb1a2}{a}\subnode{tb1b1}{b}, \subnode{tb1b2}{\textcolor{lightgray}{b}})$};
			\node (tA-2) at (4, -2) {$B^{\setof{1}}(\subnode{tb2a1}{a},\subnode{tb2b1}{\textcolor{lightgray}{b}})$};
			\node (tA-3) at (6,-2) {$B^{\setof{1,2}}(\subnode{tb3a1}{a},\subnode{tb3b1}{b})$};
			\draw [-] (tb2a1.north) -- ++(0,1.30) -| (tb1a1.south);
			\draw [-,lightgray] (tb2b1.north) -- ++(0,.95) -| (tb1b2.south);
			\draw [-] (tb3a1.north) -- ++(0,.25) -| (tb1a2.south);
			\draw [-] (tb3b1.north) -- ++(0,.50) -| (tb1b1.south);
			
		\end{tikzpicture}
		}
	\end{center}
	\caption{A derivation tree and its projection. Lines between nodes represent the embedding maps. We project the root nonterminal to its first entry. This means that the righthand $B$ is kept intact, as all its entries map to the first entry of $A$. The lefthand $B$ is projected to its first entry, as the second is no longer used. In a larger tree, this would cause the children of the left $B$ to be projected in turn.}
	\label{fig:example-projection}
\end{figure}
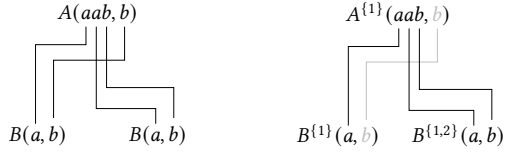

\myparagraph{Searching Decorations}
The second kind of decoration ensures the presence of a particular nonterminal and embedding map inside a derivation tree.
As such, we call these decorations \emph{searching}. 
As we will talk about the relationship between subtrees a lot in this section,
we will find it useful to introduce a notion of addressing subtrees by the path
required to reach them.
We address subtrees using paths $p \in \Nat^*$. 
If $p = \varepsilon$, then the subtree of a tree $t$ at $p$, 
written $\subtree{t}{p}$, is $t$. 
If $p = i \cdot p'$, then $\subtree{t}{p}$ is $\subtree{t_i}{p'}$, 
where $t_i$ is the $i$-th child of $t$.

Let $A$ and $B$ be nonterminals with arity $n_0$ and $m$ respectively, and let $f : [1,m] \to [1,n_0]$. Let $I = \Img(f) = \setof{i_1, \ldots,i_m}$ and $\bar I = [1,n_0]\setminus I = \setof{\bar{i}_1,\ldots,\bar{i}_{n_1}}$. We introduce a new nonterminal $A^{B,f}$ of arity $\card{\bar I} = n_1$. Intuitively, $A^{B,f}$ will faithfully recreate those entries of $A$ that are generated alongside the watershed of a $B$-rooted subtree with embedding map $f$, while ignoring those in the image of the map. In other words, the semantics are such that an $A^{B,f}$-rooted derivation tree $t$ with $\yield(t) = (w_{\bar{i}_1}, \ldots, w_{\bar{i}_{n_1}})$ exists if and only if an $A$-rooted derivation tree $r$ with $\yield(r) = (w_1, \ldots, w_{n_0})$ and a $B$-rooted subtree $\subtree{r}{p}$ with $\subtree{r}{p} \embed_f r$ exists for some path $p$.

To go with our decorated nonterminals, we need decorated productions. Each nonterminal embeds trivially into itself with the identity embedding, so we have $A^{A,\id}() \leftarrow \emptyset$ for all nonterminals~$A$. For more complicated embeddings $f$, we introduce for every production $\pi = A(\vec s) \leftarrow S$ a number of productions that each choose one nonterminal $C_i$ from the body $S$ and a suitable embedding and pursues a search in that branch of the derivation tree. Among all such productions, the choice of $C_i$ is then nondeterministic. More specifically, for every $C_i(\vec x) \in S$ with $C_i(\vec x) \embed_{\pi,g} A(\vec s)$, and every $h \in f/g$, we first introduce the production $\pi^{B,f,i,h} = A^{B,f}(\vec s) \leftarrow S'$. For every $C_j(\vec y) \in S$ with $j \neq i$ we add $C_j(\vec y)$ to $S'$, and we also add $C_i^{B,h}(\vec x)$. However, these productions currently use the wrong arity for $A^{B,f}$, and we are only interested in entries that do not contain information from $B$ anyway. Therefore, we finally replace $A^{B,f}$ and each of its productions by the version projected to $\bar I$.

Let us now explain how to construct a corresponding tree of searching nonterminals, when given a concrete derivation tree $t_0$ where we know that a subtree at some path $p$ embeds with $f$ into the root of $t_0$. As we know the exact position $p$ in this case, we will use it to parameterize the following operation: Let $t_0$ be an $A$-rooted tree with $n$ children and $\subtree{t_0}{j} \embed_{f_j} t_0$ for $j\leq n$. Let $t_1 = \subtree{t_0}{p}$ be a $B$-rooted subtree with $t_1 \embed_f t_0$. 
We construct the tree $t_0^{p,f}$ as follows: 
If $p = \varepsilon$, then $f = \id$ and $t_0^{\varepsilon,\id}$ is just $A^\emptyset() \leftarrow \emptyset$. 
Otherwise, $p = i \cdot p'$. Let $h$ be the embedding such that $t_1 \embed_h \subtree{t_0}{i}$. 
The tree $t_0^{p,f}$ is given by the root node labeled $A^{B,f}$ using the production $\pi^{B,f,i,h}$. 
The children are, for $j \neq i$, $(\subtree{t_0}{j})^{J}$ with $J = f_j^{-1}(\bar I)$. 
For $j = i$ we use $((\subtree{t_0}{i})^{p',h})^J$ as the $i$th child node, where again $J = f_i^{-1}(\bar I)$. 
As a helpful visualization, refer to \cref{fig:example-search}.

\begin{figure}
	\begin{center}
		\scalebox{\tikzscale}{
		\begin{tikzpicture}[remember picture]
			\node (tA-1) at (0,0) {$A(\subnode{t1au}{u}\subnode{t1av}{v}, \subnode{t1ax}{x}\subnode{t1ay}{y}\subnode{t1az}{z})$};
			\node (tA-2) at (-1,-2) {$D(\subnode{t1dz}{z})$};
			\node (tA-3) at (1,-2) {$B(\subnode{t1bu}{u},\subnode{t1bv}{v},\subnode{t1bx}{x})$};
			\node (tA-4) at (0,-4) {$C(\subnode{t1cv}{v})$};
			\draw [-,color=highlight-1] (t1bu) -- ++(0,.5) -| (t1au);
			\draw [-,color=highlight-1] (t1bv) -- ++(0,.75) -| (t1av);
			\draw [-,color=highlight-1] (t1bx) -- ++(0,1) -| (t1ax);
			\draw [-] (t1dz) -- ++(0,1.25) -| (t1az);
			\draw [-,color=highlight-2] (t1cv) -- ++(0,1) -| (t1bv);
			\draw [dashed,color=highlight-3] (tA-4) -- ++(-2,0) |- (tA-1);
			\node at (1.75, -.6) {$\hlA{f_B = \left\{\begin{aligned}
				1 \mapsto 1 \\[-4pt]
				2 \mapsto 1 \\[-4pt]
				3 \mapsto 2 \\
			\end{aligned}\right.}$};
			\node at (-1.1,-3.65) {$\hlC{f = 1 \mapsto 1}$};
			\node at (1.1,-3.35) {$\hlB{f_C = 1 \mapsto 2}$};

			\begin{scope}[shift={(6,0)}]
				\node (tA-1) at (0,0) {$A^{C,f}(\textcolor{lightgray}{\subnode{t2au}{u}\subnode{t2av}{v}}, \subnode{t2ax}{x}\subnode{t2ay}{y}\subnode{t2az}{z})$};
				\node (tA-2) at (-1,-2) {$D(\subnode{t2dz}{z})$};
				\node (tA-3) at (1,-2) {$B^{C,f_C}(\textcolor{lightgray}{\subnode{t2bu}{u},\subnode{t2bv}{v}},\subnode{t2bx}{x})$};
				\node (tA-4) at (0,-4) {$C^{C,\id}(\textcolor{lightgray}{\subnode{t2cv}{v}})$};
				\draw [-,color=lightgray] (t2bu) -- ++(0,.5) -| (t2au);
				\draw [-,color=lightgray] (t2bv) -- ++(0,.75) -| (t2av);
				\draw [-] (t2bx) -- ++(0,1) -| (t2ax);
				\draw [-] (t2dz) -- ++(0,1.25) -| (t2az);
				\draw [-,color=lightgray] (t2cv) -- ++(0,1) -| (t2bv);
			\end{scope}
		\end{tikzpicture}
		}
	\end{center}
	\caption{A visualization of a derivation tree and a search-decoration of it. Lines represent embedding maps. The subtree with root $C$ embeds into the root with the map $f = 1 \mapsto 1$. On the right, we construct the tree $t^{p,f}$, with $p=21$ being the path to $C$. To this end, the root is decorated with the nonterminal, $C$, and the embedding, $f$. As the image of $f$ is $1$, we project away the first entry. We repeat the same procedure with the children. Note that $f_C \in f/f_B$ and therefore the constructed tree is valid according to the decorated production rules we derived.}
	\label{fig:example-search}
\end{figure}
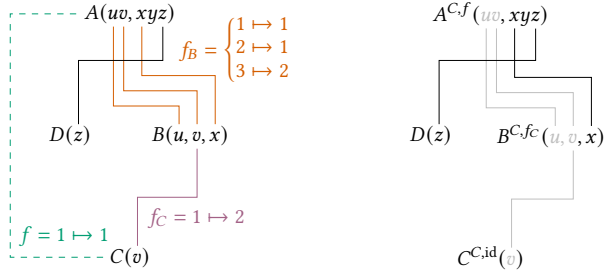

\myparagraph{Termination} 
Given a grammar of size $n$ whose arity is bounded by $k$, our decorations 
introduce $\cO(n^2 \cdot k^k)$ new nonterminals and for each such nonterminal, up to 
$\cO(n^2 \cdot k^k)$ new productions and correspondingly many new variables. For fixed 
$k$, the size of the new grammar is therefore $\cO(n^4)$. However, as we introduced new 
nonterminals, we need to iterate the decoration procedure. 
We may assume that every decorated nonterminal has strictly smaller arity than its corresponding nondecorated version. Therefore we need to repeat the construction at most $k$ times, leading to a final grammar size of $\cO(n^{4^k})$, polynomial in $n$ for fixed $k$.

\iffalse
Assume a non-fixed $k$ instead. Observing that $k \leq n$, we obtain $n^{k+2}$ new nonterminals and $n^{k+2}$ productions per nonterminal, for a total of $n^{2k+4}$. Iterating this construction $k$ times, we obtain a grammar size of $n^{2^{\cO(k)}}$.
\fi

We conclude with the following technical lemma that shows the correspondence between languages generate by the projected and original nonterminals.

\begin{restatable}{lemma}{projectedntyield}
	\label{lem:projected-nt-yield}
	Let $A$ be a nonterminal of arity $n$ and $I = \setof{i_1, \ldots, i_m} \subseteq [1,n]$ a subset of indices. We have $(w_1, \ldots, w_n) \in L(A)$ if and only if 
	$(w_{i_1}, \ldots, w_{i_m}) \in L(A^I)$.
\end{restatable}

\begin{proof}
	The only-if direction follows directly from the construction of $A^I$.
	
	For the other direction, if $A^I$ can generate $(w_{i_1}, \ldots, w_{i_m})$ directly via some terminal production~$\pi^I$, then there is a corresponding terminal production $\pi$ for $A$ and the statement immediately holds.
	
	Let therefore $t$ be an $A^I$-rooted derivation tree whose first level corresponds to some production $\pi^I = A^I(s_{i_1}, \ldots, s_{i_m}) \leftarrow S'$. Each nonterminal in $S'$ is of the form $C_i^{f_i^{-1}(I)}(\vec x_i)$. Let $\vec u_i = (u_{i,i_1}, \ldots, u_{i,i_{m'}})$ be the output of the $i$-th subtree of $t$. By induction, we have $\vec u_i$ in $L(C_i)$. 
	
	We construct a tree $\tilde t$ for $A$, using at the first level the production $\pi$. For the $i$-th child, we use the tree that yields the tuple $\vec u_i$. Let $(w'_1, \ldots, w'_n)$ be the output of $\tilde t$. If $i \in I$, then by construction the set of variables from the $j$-th child used in $s_i$ is contained in $f_j^{-1}(I)$, and therefore $w'_i = w_i$.
\end{proof}

Observing in this proof that $\tilde t$ projected to $I$ is exactly $t$, we obtain the following as a corollary.

\begin{corollary}
	Let $t$ be a derivation tree with yield $(u_1, \ldots, u_n)$. Let $I = \setof{i_1, \ldots, i_m}$ be a set of indices. Then $\yield(t^I) = (u_{i_1}, \ldots, u_{i_m})$.
\end{corollary}

\subsection{Accelerated Grammars}
\label{sec:accelerated-grammars}

As the acceleration of entries introduces languages of the shape~$\Gamma^*$ for $\Gamma \subseteq \Sigma$, we introduce here a generalization of {\mcfg}s, called accelerated {\mcfg}s: In an accelerated \mcfg~$\cG$, construction rules for entries may not only use symbols $a \in \Sigma$ and variables 
$x \in X$, but also languages of the shape~$\Gamma^*$ for $\Gamma \subseteq \Sigma$. The yield of a tree of $\cG$ is therefore a tuple of \emph{languages}, 
and we say that a tuple of words $(w_1, \ldots, w_n)$ is in the language of a 
nonterminal~$A$ if $A \rightarrow^* L_1, \ldots, L_n$ and $w_i \in L_i$.

We can construct an accelerated {\mcfg}~$\cG_\Sigma$ from a base \mcfg~$\cG$. Intuitively, for any nonterminal~$A$ we would like to be able to guess an idempotent map $f$ and a path $p$ such that there exists an $A$-rooted tree $t$ such that the subtree $\subtree{t}{p}$ is also $A$-rooted and embeds with $f$. As we have discussed above, we then introduce productions that 

\begin{itemize}
	\item project away the watershed,
	\item introduce a searching decoration along the first step of $p$, and
	\item introduce a copy of $A$ to provide the possibly non-acceleratable parts of the fixpoints.
\end{itemize}

Formally, for every arity-$n$ nonterminal $A$, every idempotent $f : [1, n] \to [1,n]$, every production $\pi = A(s_1, \ldots, s_n) \leftarrow S$, every nonterminal $C_i(\vec x)$ with $C_i(\vec x) \embed_{\pi,f_i} A(\vec s)$ and every $g \in f/f_i$, we introduce a production $\pi^{f,i,g} = A(s'_1, \ldots, s'_n) \leftarrow S'$ (for fixed dimension, this adds only polynomially many productions). Recall that $I = \Img(f)$ and $\bar I = [1,n] \setminus I$.
The set $S'$ will contain the following:

\begin{itemize}
	\item \emph{project away the watershed:} for every nonterminal $C_j(\vec y) \in S$ with $j \neq i$, we introduce a projected version $C_j^{I_j}(\vec y)$, where $I_j = f_j^{-1}(\bar I)$;
	\item \emph{introduce a searching decoration:} for the case of $j = i$ specifically, we introduce the projected searching nonterminal $(C_i^{A,g})^{I_i}(\vec x)$, which searches for an $A$-rooted subtree that embeds into $C_i$ with $g$, and is projected to $I_i = f_i^{-1}(\bar I)$;
	\item \emph{lift $A$:} to derive the non-acceleratable parts of the fixpoints, we introduce $A^I(z_{i_1}, \ldots, z_{i_m}) \in S$. 
\end{itemize}

We also need to modify the construction rules for the individual entries. Since the entries $i \in \bar I$ are those that are not accelerated, we set $s'_i = s_i$. For the accelerated entries, that is for $i \in I$, we set $s'_i = \Sigma_{i,\mathrm{left}}^* z_i \Sigma_{i,\mathrm{right}}^*$.  Intuitively, these alphabets should contain exactly those letters that can be produced on the watershed of an idempotent embedding and hence we call them \emph{watershed alphabets}.

\myparagraph{Watershed Alphabets}
\begin{figure}
	\begin{center}
		\begin{subfigure}[t]{.30\textwidth}
			\vskip 0pt
			\centering
			\scalebox{\tikzscale}{
			\begin{tikzpicture}[remember picture,tikzmark prefix=watershed-alph-1]
				\node (a1) at (0,0) {$A(\subnode{a1-1}{dyzx}, \subnode{a1-2}{c})$};
				\node (a2) at (0,-1) {$B(\subnode{a2-1}{dyz}, \subnode{a2-2}{x})$};
				\node (a3) at (2,-1) {$C(\subnode{a3-1}{c})$};
				\node (a4) at (0,-2) {$A(\subnode{a4-1}{y},\subnode{a4-2}{z})$};
				\node (a5) at (-2,-2) {$D(\subnode{a5-1}{d})$};
				\node (a6) at (2,-2) {$X(\subnode{a6-1}{x})$};
				\node (a7) at (-1,-3) {$Y(\subnode{a7-1}{y})$};
				\node (a8) at (1,-3) {$Z(\subnode{a8-1}{z})$};
				
				\draw [-,thick,color=highlight] (a1-1) -- ++(0,-.5) -| (a2-1);
				\draw [-,thick,color=highlight] (a1-1) ++(0,-.5) -| (a2-2);
				\draw [-] (a1-2) -- ++(0,-.5) -| (a3-1);
				
				\draw [-] (a2-1) ++(0,-.5) -| (a4-2);
				\draw [-] (a2-1) ++(0,-.5) -| (a5-1);
				\draw [-,thick,color=highlight] (a2-2) -- ++(0,-.5) -| (a6-1);
				\draw [-,thick,color=highlight] (a2-1) -- ++(0,-.5) -| (a4-1);
				
				\draw [-] (a4-2) -- ++(0,-.5) -| (a8-1);
				\draw [-] (a4-1) -- ++(0,-.5) -| (a7-1);
			\end{tikzpicture}
			}
			\caption{The first entry of $\subtree{t}{12}$ and the first entry of $\subtree{t}{13}$ meet in the first entry of $t$.}
		\end{subfigure}%
		\quad%
		\begin{subfigure}[t]{.65\textwidth}
			\vskip 0pt
			\centering
			\scalebox{\tikzscale}{
			\begin{tikzpicture}[remember picture,tikzmark prefix=watershed-alph-2]
				\node (a1) at (0,0) {$A(\subnode{a1-1-}{dyzx}, \subnode{a1-2-}{c})$};
				\node (a2) at (0,-1) {$B(\subnode{a2-1-}{dyz}, \subnode{a2-2-}{x})$};
				\node [fill=highlight-2!20](a3) at (2,-1) {$C(\subnode{a3-1-}{c})$};
				\node (a4) at (0,-2) {$A(\subnode{a4-1-}{y},\subnode{a4-2-}{z})$};
				\node [fill=highlight-3!20] (a5) at (-2,-2) {$D(\subnode{a5-1-}{d})$};
				\node [fill=highlight-2!20] (a6) at (2,-2) {$X(\subnode{a6-1-}{x})$};
				\node (a7) at (-1,-3) {$Y(\subnode{a7-1-}y)$};
				\node (a8) at (1,-3) {$Z(\subnode{a8-1-}z)$};
				
				\draw [-] (a1-1-) -- ++(0,-.5) -| (a2-1-);
				\draw [-] (a1-1-) ++(0,-.5) -| (a2-2-);
				\draw [-] (a1-2-) -- ++(0,-.5) -| (a3-1-);
				
				\draw [-] (a2-1-) ++(0,-.5) -| (a4-2-);
				\draw [-] (a2-1-) ++(0,-.5) -| (a5-1-);
				\draw [-] (a2-2-) -- ++(0,-.5) -| (a6-1-);
				\draw [-] (a2-1-) -- ++(0,-.5) -| (a4-1-);
				
				\draw [-] (a4-2-) -- ++(0,-.5) -| (a8-1-);
				\draw [-] (a4-1-) -- ++(0,-.5) -| (a7-1-);
			\end{tikzpicture}
			}
			\caption{The first entry of $A$ is the fixpoint of the embedding. $d$ will be in $\Sigma_{1,\mathrm{left}}$ (it joins from the left in $\subtree{t}{1}$). $x$ will be in $\Sigma_{1,\mathrm{right}}$ (it joins from the right in $t$). $c$, despite not joining in the displayed tree, will be in $\Sigma_{1,\mathrm{right}}$ as we can easily construct a larger tree where it does. Note that $z$ is not in $\Sigma_{1,\mathrm{right}}$ despite joining from the right, because it can only ever occur rooted below the lowest $A$.}
		\end{subfigure}
	\end{center}
	\caption{The construction of watershed alphabets.}
	\label{fig:watershed-alphabet}
\end{figure}
To formalize watershed alphabets, we introduce the notion of a \emph{join point} of two entries in a derivation tree, that is, the point in the tree where they first occur in the same entry of the same node. If a derivation tree with an idempotent embedding exists and some letter joins the the fixpoint somewhere inbetween the two nonterminals, we add that letter to the watershed alphabet of that fixpoint. However, since we cannot guarantee that anything rooted below the bottom nonterminal can be accelerated, we need to ignore those letters (see \cref{fig:watershed-alphabet}). All these conditions can be checked with a simple graph search in the connectivity graph of the grammar.

Formally, let $t$ be a tree, $\subtree{t}{p_1}$ an $A_1$-rooted and $\subtree{t}{p_2}$ an $A_2$-rooted subtree of arity $n_1$ and $n_2$ respectively and $\ell_1 \in [1, n_1]$, $\ell_2 \in [1, n_2]$. We say that the $\ell_1$-th entry of $\subtree{t}{p_1}$ and the $\ell_2$-th entry of $\subtree{t}{p_2}$ \emph{join at $\subtree{t}{q}$ in entry $\ell$} if $\subtree{t}{p_1} \embed_{f_1} \subtree{t}{q} \hookleftarrow_{f_2} \subtree{t}{p_2}$ with $f_1(\ell_1) = \ell = f_2(\ell_2)$ and $q$ is the longest path for which this condition holds.

To decide whether a letter should be a part of the left or the right watershed alphabet, $\Sigma_{i,\mathrm{left}}$ or $\Sigma_{i,\mathrm{right}}$, we also need a notion of order. Let $B(\vec s) \leftarrow S$ be the root production of to $\subtree{t}{q}$. Let $C_i(\vec x)$ and $C_j(\vec y) \in S$ be the nonterminals of the (not necessarily distinct) $i$-th and $j$-th child of $\subtree{t}{q}$ such that $\subtree{t}{p_1} \embed_{g_1} \subtree{t}{(q \cdot i)}$ and $\subtree{t}{p_2} \embed_{g_2} \subtree{t}{(q \cdot j)}$. We say that the $\ell_1$-th entry of $\subtree{t}{p_1}$ joins the $\ell_2$-th entry of $\subtree{t}{p_2}$ \emph{from the left} if $x_{g_1(\ell_1)}y_{g_2(\ell_2)} \subword s_\ell$, or joins \emph{from the right} if $y_{g_2(\ell_2)}x_{g_1(\ell_1)} \subword s_\ell$.

We may assume that terminal letters are only introduced individually and by leaf nodes. We denote by $L_a(a)$ a nonterminal introducing the letter $a$ as such. We say that a letter $a$ is in $\Sigma_{i,\mathrm{left}}$ if there exists an $A$-rooted tree $t$ with an $A$-rooted subtree $\subtree{t}{p}$ and an $L_a$-rooted subtree $\subtree{t}{q}$ such that there is an idempotent embedding $\subtree{t}{p} \embed_f t$ and (the only entry of) $\subtree{t}{q}$ joins the $i$-th entry of $\subtree{t}{p}$ from the left. Additionally, we require that $\subtree{t}{q}$ is not a child of $\subtree{t}{p}$.

\subsection{Concise Overapproximations}

Having introduced the notion of accelerated grammars and described how to construct an accelerated grammar $\cG_\Sigma$ from a standard \mcfg~$\cG$, we show that $\cG_\Sigma$ does not only overapproximate $\cG$, but that the two are in fact equivalent with respect to their downward closure.
Let $\preceq^n$ denote the subword ordering lifted to tuples of words of dimension $n$.

\begin{restatable}{lemma}{sigmagrammartight}
	
	$L(\cG) \subseteq L(\cG_\Sigma)$, and for every tree $t'$ of $\cG_\Sigma$ and every word $(w_1, \ldots, w_n) \in L(t')$ there is a tree $t$ of $\cG$ such that $(w_1, \ldots, w_n) \preceq^n \yield(t)$.
\end{restatable}

\begin{proof}
	\label{proof:sigma-grammar-tight}
	The first statement follows from the fact that every tree of $\cG$ is also a valid tree of $\cG_\Sigma$.
	
	For the second statement, observe that if $t'$ contains no accelerated productions, then it will also not contain any decorated nonterminals. In that case $t'$ is also a tree of $\cG$ and the statement holds.
	
	We therefore show how to eliminate accelerated productions. WLOG assume that the root production is an accelerated rule $\pi^{f,i,g} = A(s_1, \ldots, s_k) \leftarrow S'$, derived from a production $\pi = A(s_1, \ldots, s_k) \leftarrow S \in \cG$. Let $(w_1, \ldots, w_n) \in \yield(t')$.
	
	From \cref{lem:projected-nt-yield} we know that if the $i$-th child of $t'$ is a $B_i^I$-rooted tree with yield $(u_i)_{i \in I}$ then there exists a $B_i$-rooted tree with yield $(u_1, \ldots, u_m)$. Similarly, if $(v_i)_{i \in J}$ is the yield of the $C^{A,g}$-rooted child, then there exists a $C$-rooted tree whose yield is $(v_1, \ldots, v_n)$. We can therefore construct $\hat t = A[t_1, \ldots, t_{n-1}, t_n]$ with $t_i$ being the $B_i$-rooted tree and $t_n$ being the $C$-rooted tree.
	
	Observe that for all $i \not\in \Img(f)$, this construction already satisfies the requirement by having $\yield(\hat t)[i] = w_i$. Furthermore, there is a path $p$ in $t_n$ such that the subtree $r = \subtree{t_n}{p}$ is $A$-rooted and embeds with $g$ into $t_n$ and hence with $f$ into $\hat t$. We will therefore aim to substitute $r$ with an $A$-rooted subtree with yield $(z_1, \ldots, z_n)$ such that for $i \in \Img(f)$, $w_i \preceq z_i$.
	
	To this end, let $\hat r$ be the $A^I$-rooted child of $t'$ that corresponds to the nonterminal $A^I((y_i)_{i \in I})$ that was newly added to the body of $\pi^{f,i,g}$ during its construction. We have $\yield(\hat r) = (o_{i_1}, \ldots, o_{i_k})$. Additionally, we know that $w_i = u o_i v$ for $u \in \Sigma_{i,\mathrm{left}}^*, v \in \Sigma_{i,\mathrm{right}}^*$. 
	
	We first construct an $A$-rooted tree $\hat r_0$ using \cref{lem:projected-nt-yield}. This tree satisfies $o_i \preceq \yield(\hat r_0)[i]$. After that, we will iteratively construct trees $\hat r_1, \hat r_2$ and so on that will incorporate successive letters from $u$ or $v$ into the $i$-th entry.
	
	Let $u = u_1 \cdots u_n$. By the construction of $\Sigma_{\mathrm{left}}$, this means that there exists an $A$-rooted tree $\tau$ and an $A$-rooted subtree $\tau'$ of $\tau$ with $\tau' \embed_f \tau$ such that the $i$-th entry is of the form $u' u_n u'' x_i v'$, where $x_i$ is traced back to the $i$-th entry of $\tau'$. Therefore substituting $\hat r_0$ for $\tau'$ in $\tau$ yields a new tree, $\hat r_1$, whose output in the $i$-th entry is $z$ with $u_n o_i \preceq z$. We iterate this process until we have eliminated all letters from $u$, and repeat a similar process with the letters of $v$ (note that for $v$ we must proceed from the first letter in order to obtain the correct ordering). The resulting subtree will be $r$. Finally, we substitute $r$ in $\hat t$ at position $n \cdot p$ to obtain $t$.
\end{proof}

All that remains to show is that accelerated grammars even allow for \emph{concise} representations of downward closures. We show this by showing that every ``large'' tree of $\cG_\Sigma$ is already subsumed by a smaller tree. Some of the proof details can be found in \refApp{\cref{app:conciseness-step}}.

\begin{restatable}{lemma}{sigmagrammarconcise}
	\label{lem:sigma-grammar-conciseness}
  Fix a bound $d$ on the dimension of $\cG_\Sigma$.
	There exists a constant $c$ such that for every tree $t$ of $\cG_\Sigma$ there exists a tree $t'$ of $\cG_\Sigma$ such that $\downc{L(t)} \subseteq \downc{L(t')}$ and $\height(t') \leq c\card{\cG_\Sigma}$.
\end{restatable}

\begin{proof}[Proof]
	Assume first that $c$ exists and we know its value. We proceed by induction on $\card{t}$. For the empty tree, the statement trivially holds. Assume therefore that the statement holds for all trees with less than $m$ nodes and let $t$ have $m + 1$ nodes. If $\height(t) \leq c\card{\cG_\Sigma}$, the statement holds. Otherwise, our choice of $c$ will be such that we are able to find a path $p$ in $t$ such that along $p$ we find four distinct $A$-rooted subtrees $t_0, t_1$, $t_2$, and $t_3$ such that $t_i \embed_f t_{i-1}$ for some idempotent $f$. In that case, we can use the accelerated rules to cut out part of the tree. The resulting tree may not be shorter, as $p$ may not be the unique longest path in the tree, but it will contain less nodes. For the details of the tree cutting construction, please refer to \refApp{\cref{app:conciseness-step}}. 
	
	Our goal therefore is to chose $c$ in a way that guarantees the existence of such subtrees with idempotent embeddings. Let $t$ be a tree, $p$ be a path on $t$ and $\cS_A \subseteq \mathrm{prefixes}(p)$ the set of all occurrences of $A$-labeled nodes along $p$. We consider the complete graph obtained by using $\cS_A$ as a vertex set and color each edge $\{p',p'q\}$ with the embedding $f$ from $\subtree{t}{p'q}$ into $\subtree{t}{p'}$. If we find a monochromatic 4-clique $\setof{p_0, p_0p_1, p_0p_1p_2, p_0p_1p_2p_3}$, this implies $f \circ f = f$ and therefore $f$ is idempotent. Ramsey's theorem \cite{Ramsey:1930:ramsey-theorem} tells us that there exists some value $c(d)$ such that every $d$-colored complete graph of at least $c(d)$ vertices contains such a monocolored $4$-clique. In particular, since we assume the dimension $k$ and therefore the number of embeddings between nonterminals is fixed, $c(d)$ becomes a constant. We conclude by observing that on any path of length $c(d) \card{\cG_\Sigma}$, at least one nonterminal occurs $c(d)$ many times.
\end{proof}

Having established the conciseness of accelerated grammars, all that remains is to show how to construct a small NFA for the downward closure.

\begin{proof}[Proof of \cref{thm:dcimcfg}]
	 Assume that $\card{\initgrammar} = n$ and the dimension is some fixed
	 $k$. We construct the corresponding accelerated \mcfg~$\initgrammar'$. The size 
	 of $\initgrammar'$ is $\card{\initgrammar'} = n^{2^{\cO(k)}} = \poly(n)$. 
	 
	 It follows that to 
	 compute the downward closure of $\initgrammar$, it suffices to consider trees 
	 of $\initgrammar'$ of height at most $\poly(n)$. Such trees have a size of 
	 $\poly(n)^{\poly(n)} = 2^{\poly(n)}$. 
	 Therefore the yield of such a tree is also of length $2^{\poly(n)}$. 
	 As the language of an accelerated grammar is a simple concatenation of 
	 individual letters and languages of the shape $\Gamma^*$ for 
	 $\Gamma \subseteq \Sigma$, for a given tree of size $m$ we can compute an $\cO(m)$-sized NFA accepting its downward closure. 
	 Therefore, our final NFA for the downward closure is the union of the NFAs for the trees of $\initgrammar'$ of height at most $\poly(n)$. There are $\poly(n)^{\poly(n)^{\poly(n)}} = 2^{2^{\poly(n)}}$ many of those, each of size $2^{\poly(n)}$, yielding a final size of $2^{2^{\poly(n)}}$.
	 \iffalse
	 Now in the case of non-fixed $k$, we get a tree height of $2^{n^{2^{\cO(k)}}}$ and therefore a size of $2^{2^{n^{2^{\cO(k)}}}}$ and an NFA of size $2^{2^{2^{n^{2^{\cO(k)}}}}}$, i.e. a five-fold exponential upper bound.
	 \fi
\end{proof}

Note that our polynomial bound on derivation tree heights only yields
a doubly exponential NFA upper bound, compared to a singly exponential
bound for CFG (because for CFG, all polynomial-height trees can be
simulated by a polynomial-stack-height pushdown automaton). As our
lower bound shows (\cref{claim:LBmcfl}), the additional blowup for MCFG is
unavoidable.

\section{Application: Context-Bounded Safety in Concurrent Programs}\label{sec:applications}

\newcommand{\mmap}{\mathbf{m}}
\newcommand{\nmap}{\mathbf{n}}
\newcommand{\multiset}[1]{{\mathbb{M}[ #1 ]}}
\newcommand{\multi}[1]{{\llbracket #1 \rrbracket}}
\newcommand{\parikh}{{\Psi}}

\newcommand{\pref}{\mathsf{pref}}

\newcommand{\Runs}[1]{\mathsf{Runs}(#1)}
\newcommand{\Runsk}[1]{\mathsf{Runs}_k(#1)}
\newcommand{\Preruns}[1]{\mathsf{Preruns}(#1)}
\newtheorem{problem}[definition]{Problem}

\newcommand{\spdownc}[1]{{#1}\mathord{\Downarrow}}
\newcommand{\spdownck}[2][k]{{#2}\mathord{\Downarrow_{#1}}}
\newcommand{\parikhspdownck}[2][k]{{#2}\mathord{\Downarrow^{\parikh}_{#1}}}
\newcommand{\alphdownck}[3][k]{{#2}\mathord{\downarrow_{#1,#3}}}
\newcommand{\parikhalphdownck}[3][k]{{#2}\mathord{\downarrow^{\parikh}_{#1,#3}}}
\newcommand{\cA}{\mathcal{A}}
\newcommand{\Natomega}{\Nat_\omega}

Let us provide an application of our results beyond the immediate settings of
bounded-\streewidth\ systems and \mcfg{}s.
We consider the model of concurrent programs,
where each individual concurrent process already computes an MCFL,
and more processes can be created dynamically during execution.
Here we can adapt our results on downward closures to reason about safety verification.
To be precise, we consider safety for concurrent programs over MCFL under the restriction
of bounded context switching.
This restriction is important, because without it safety becomes undecidable
(even if each process only computes a CFL \cite{QadeerR05,ramalingam2000context}).

To properly talk about this setting, we first need to recall some auxiliary notions.

\myparagraph{Multisets and the Parikh image}
Since we assume no relations between our concurrent processes, it makes sense to model the collection
of their individual configurations as some unordered structure.
However, a mere set does not suffice, since that would not properly track duplicate configurations.
As such we use the common generalization to multisets, where each potential element is
assigned a natural number that specifies how often it occurs.
A set can then be seen as a multiset, where every assignment is either to $0$ or to $1$.

Formally, a \emph{multiset} $\mmap\colon X \rightarrow\Nat$ over a set $X$ maps each
element of $X$ to a natural number.
The size $|\mmap|$ of a multiset $\mmap$ is defined as $|\mmap|=\sum_{x \in X}\mmap(x)$.
We write $\multiset{X}$ to denote set of all multisets over $X$.
For a subset $Y \subseteq X$, we identify it with the multiset in $\multiset{X}$
mapping each $x \in Y$ to $1$ and everything else to $0$.

In our model, when several processes are spawned in a row without interruptions in-between,
the precise order of these spawns does not matter.
As such, let us recall how to obtain an unordered multiset from a sequence of elements:
The Parikh image $\parikh(w) \in \multiset{\Sigma}$ of a word $w \in \Sigma^*$
is the multiset such that for each letter $a \in \Sigma$ we have $\parikh(w)(a) = |w|_a$,
where $|w|_a$ denotes the number of occurrences of $a$ in $w$.
When denoting Parikh images, we also identify multisets in $\multiset{\Sigma}$
with vectors in $\Nat^{\sizeof{\Sigma}}$ (for some arbitrary, but fixed, total order on $\Sigma$).

\myparagraph{Concurrent Programs}
We consider the model of concurrent programs, which are defined over an arbitrary class of languages $\cC$
(with some mild restrictions on $\cC$)~\cite{DBLP:conf/icalp/0001GMTZ23},
though here we look at the case where $\cC$ is the class of MCFLs.
The complete definition of the model and its related notions can be found in \refApp{\cref{app:programs-formal-def}}; we only give the relevant intuition here.

A \emph{concurrent program} ($\CP$) $\ap$ over MCFLs consists
of a set of \emph{global states} $D$, an alphabet of \emph{process names} $\Sigma$, 
and for each $a \in \Sigma$ a \emph{process language} $L_a$ that is an MCFL
(given e.g.\ as an \mcfg).
For each process name $a \in \Sigma$, its corresponding language is of the form
$L_a \subseteq aD \big(\Sigma \cup (D \times D)\big)^* (D \cup D \times D)$, where the letters in $D$
denote the \emph{initial} and a possible \emph{final state} of a particular process' execution,
the letters in $\Sigma$ denote its \emph{spawns},
and the letters in $D \times D$ denote its \emph{context switches}.

The program $\ap$ starts in some initial global state $d_0$ with some initial multiset
of processes.
To run a process from current state $d$, $\ap$ nondeterministically schedules a process
whose execution either has not yet started and can start in $d$,
or whose execution when it was last scheduled ended in a context switch
$(e,d)$ for some global state $e \in D$.
The execution is then extended until either a context switch $(d',d'')$ or a final state
$d'$ occurs, upon which the current global state transitions to~$d'$.
All spawns encountered along the way are added to the multiset of processes,
and whenever a process reaches the final state in its execution, it is removed
from said multiset.
Note that any spawned process with name $a \in \Sigma$ is restricted to executions from the
process language $L_a$.

\myparagraph{Bounded context switching}
Most decision problems are undecidable for concurrent programs,
if we impose no restrictions on their behavior \cite{QadeerR05,ramalingam2000context}.
As such we consider the popular restriction of bounded context switching.
Like before, we do not give the full formal definitions from~\cite{DBLP:conf/icalp/0001GMTZ23},
which can also be found in \refApp{\Cref{app:programs-formal-def}}.
Rather, we only give the relevant intuition.

For a number $k \in \Nat$, also called a \emph{context bound},
a run of a $\CP$ is $k$-context-bounded, if each individual process is scheduled
at most $k+1$ times.
Alternatively, one says that each process \emph{experiences} $k$ context switches.
Note that one considers only the part of a process' execution that is actually scheduled as experienced,
and if the symbol reached after $k+1$ schedulings is in $D \times D$, it is not counted
as an experienced $(k+1)$th context switch. 

Now we can define the decision problem of interest.

\begin{problem}[Context-bounded safety for concurrent programs]\ \label{problem:cb-reach}
\begin{description}
  \item[Given] A unary-encoded context-bound $k \in \Nat$, a concurrent program $\ap$ over the MCFLs, and one of its global states $d_f\in D$, where each process language $L_a$ of $\ap$ is given as an \mcfg\ of constant dimension and rank.
  \item[Question] Is there a $k$-context-bounded run of $\ap$ that reaches the global state $d_f$?
\end{description}
\end{problem}

\begin{theorem} \label{thm:CBSafety2EXPSPACE}
  Problem~\ref{problem:cb-reach} is $\TWOEXPSPACE$-complete.
\end{theorem}

Note that $\TWOEXPSPACE$-hardness already holds already when $k=1$
and $\ap$ is over the CFLs \cite{DBLP:conf/icalp/BaumannMTZ20}.

The rest of this section is dedicated to proving \Cref{thm:CBSafety2EXPSPACE}.
To give a brief idea, we want to replace each \mcfg{} process with a finite-state process,
where the latter overapproximates the former in a way that preserves context-bounded safety,
by computing a partially permuting variant of its downward closure (see below).
Then, in the finite-state setting, we can use a known algorithm
to decide context-bounded safety in $\EXPSPACE$.
Since the aforementioned replacement of processes constitutes an exponential blow-up,
we get a $\TWOEXPSPACE$-procedure in total.

We require constant dimension and rank for the \mcfg{}s in the input,
so that we can perform constructions for certain closure properties in polynomial time
(see \Cref{lem:mcfg-closure-poly} below).
However, these requirements are quite natural:
For example, when we invoke \Cref{thm:TA2MCFL} to turn a constant-width tree automaton into an \mcfg,
the latter will also have constant dimension, and even a constant rank of $2$.

\myparagraph{Special downward closures}
To prove the above theorem, we need to define (and later compute)
some special downward closure variants.
The first variant essentially allows us to retain information about all $k$ context switches,
but potentially lose any spawns occurring between these switches.
For a language $L \subseteq \Gamma^*$, a subalphabet $\Theta \subseteq \Gamma$, and a number $k$
let $\alphdownck{L}{\Theta}$ denote its \emph{alphabet-preserving downward closure},
where subwords are not allowed to drop letters from $\Theta$,
and have to contain exactly $k$ of them.
This is slightly different from the definition in \cite{DBLP:journals/pacmpl/BaumannMTZ22},
where words with fewer than $k$ letters are also included,
but our version fits our purpose slightly better.
Formally, $\alphdownck{L}{\Theta} :=
\{u \in \Gamma* \mid \exists v \in L \colon u \subword v \wedge |u|_\Theta = |v|_\Theta = k\}$.
Here, $|w|_\Theta$ denotes the number of occurrences of letters from the alphabet $\Theta$
in the word $w$.

The next downward closure variant allows us to also lose information about the exact order
of spawns in each segment of a process.
Formally, we define the \emph{partially permuting downward closure}:
Let $\parikhalphdownck{L}{\Theta}$ denote the variant of $\alphdownck{L}{\Theta}$,
where every maximal infix in $(\Gamma\setminus\Theta)^*$
only needs to be preserved up to its Parikh-image~$\parikh(-)$.
Formally $u = u_0 \theta_1 u_1 \cdots \theta_k u_k \in \parikhalphdownck{L}{\Theta}$
if and only if $u_0,\ldots,u_k \in (\Gamma\setminus\Theta)^*$,
$\theta_1,\ldots,\theta_k \in \Theta$, and
there is $v = v_0 \theta_1 v_1 \cdots \theta_k v_k \in \alphdownck{L}{\Theta}$
such that $\parikh(u_i) = \parikh(v_i)$ for every $0 \leq i \leq k$.

\myparagraph{Deciding context-bounded safety}
We need two more auxiliary lemmas before we can prove \cref{thm:CBSafety2EXPSPACE}.
The first allows us to compute simple closure properties for \mcfg{}s in polynomial time,
provided they have constant dimension and rank.
The second computes an NFA for the partially permuting downward closure. In particular, with the relaxation of preservation up to the Parikh-image, we can compute downward closures of {\mcfg}s in $\EXPSPACE$ instead of $\TWOEXPSPACE$.

\begin{lemma} \label{lem:mcfg-closure-poly}
  Given an \mcfg\ $\grammar$ of constant dimension and rank, and an NFA $\cA$,
  we can compute the following in polynomial time:
  \begin{enumerate}[label={(\alph*)}]
    \item an \mcfg\ for the intersection $\langof{\grammar} \cap \langof{\cA}$, and \label{item:mcfg-nfa-intersect}
    \item an \mcfg\ for the language $\pref(\langof{\grammar})$
    containing all prefixes of words in $\langof{\grammar}$. \label{item:mcfg-prefix-lang}
  \end{enumerate}
\end{lemma}

We describe the full proofs in \refApp{\Cref{app:mcfg-closure}
(see \Cref{cor:mfcg-nfa-intersection,cor:mcfg-prefix-lang})}.
To give an idea, the proof of \Cref{lem:mcfg-closure-poly}\ref{item:mcfg-nfa-intersect}
is a relatively simple adaptation of the
triple construction for intersecting a CFG and an NFA.

For \Cref{lem:mcfg-closure-poly}\ref{item:mcfg-prefix-lang} the idea is not difficult either,
but requires some additional tools from language theory.
In fact, we show something more general,
namely that given an \mcfg{} $\cG$ and a finite-state transducer $\mathcal{T}$, we can compute the image
of $\langof{\cG}$ under $\mathcal{T}$ in polynomial time.
Here a \emph{(finite-state) transducer} is an NFA that in addition to reading a symbol
may also write a symbol on each transition,
thus basically generating a pair of input word and output word for each accepting run.
It is not hard to see how to construct a transducer that maps any word to each of its prefixes.

To compute the image of $\langof{\cG}$ under $\mathcal{T}$ we construct an intermediary grammar $\cG'$
that arbitrarily shuffles transitions of $\mathcal{T}$ between words generated by $\cG$.
Then we use \Cref{lem:mcfg-closure-poly}\ref{item:mcfg-nfa-intersect} to intersect with an NFA
that checks whether the transitions form a valid run over the input word given by the letters of $\cG$.
Finally, we project away the input letters and replace transitions with their output letters to obtain
the desired \mcfg{}. None of this is hard to do in polynomial time.

As a remark, it was previously known that the MCFL are closed under taking the image under
a transducer (also called a \emph{rational transduction}) \cite{SekiMFK91},
but we could not find a source regarding effective polynomial complexity for our case.
Also note that to break up the image under a transducer $\mathcal{T}$ into several steps,
we use an adaptation of Nivat's theorem
(see e.g. \cite[Chapter III, Theorem~3.2]{berstel2013transductions}).

\begin{restatable}{lemma}{parikhdcl}\label{lem:parikh-dcl}
  Let $\grammar$ be an accelerated \mcfg\ with $\Theta \subseteq \Gamma$,
  where for every tree $t$~of $\grammar$ there is a tree~$t'$ of $\grammar$,
  such that $\downc{L(t)} \subseteq \downc{L(t')}$ and $\height(t') \leq c\sizeof{\grammar}$ for a constant $c$ as in \cref{lem:sigma-grammar-conciseness}. We can compute an NFA for $\parikhalphdownck{\langof{\grammar}}{\Theta}$
  in time $2^{\poly(k\cdot \sizeof{\grammar})}$.
\end{restatable}

Before we argue how to prove \Cref{lem:parikh-dcl}, let us see how it implies
\cref{thm:CBSafety2EXPSPACE}.
The proof idea is rather simple.
We use our downward-closure construction and \mcfg{} closure properties to turn each
process language $L_a$ into one that might drop process names, but preserves up to the $k$
first context switches.
This replaces each \mcfg\ with an NFA that is at most exponentially larger,
meaning it yields a concurrent program over the regular languages.
For these, we can decide $k$-context-bounded safety in $\EXPSPACE$ \cite{AtigBQ11},
and so we obtain a $\TWOEXPSPACE$ algorithm for $k$-context-bounded safety over MCFLs.
It is known that context-bounded safety is preserved under dropping process names,
but preserving context-switches \cite{DBLP:conf/icalp/0001GMTZ23}. The additional Parikh-relaxation does not change this.
The latter is easy to see from just the semantics of concurrent programs.

\begin{proof}[Proof of \Cref{thm:CBSafety2EXPSPACE}]
  The lower bound follows from the special case of $\CP$s over context-free languages,
  where $k$-context-bounded reachability was shown to be $\TWOEXPSPACE$-hard in \cite{DBLP:conf/icalp/BaumannMTZ20} for unary-encoded $k$ and even fixed $k \geq 1$.
  
  For the upper bound, consider a $\CP$ $\ap$ over the MCFLs with a context bound $k \in \Nat$,
  where for each $a \in \Sigma$, $L_a$ is given as an $\mcfg$ $\grammar_a$
  of constant dimension and rank.
  Let $\Gamma = \Sigma \cup D \cup (D \times D)$, where $D$ are the global states of $\ap$.
  Our goal is to compute for each $L_a \subseteq \Gamma^*$ an NFA for its downward closure,
  that preserves the first up to $k$ context-switches.
  We accomplish this in the following five steps.
  First, we construct an \mcfg\ $\grammar_{a,\pref}$ for $\pref(L_a)$ using \Cref{lem:mcfg-closure-poly}\ref{item:mcfg-prefix-lang}.
  Second, for each $i = 0$ to $k$ we intersect $\grammar_{a,\pref}$ with the regular language
  $aD\big(\Sigma^* (D \times D) \big)^i \Sigma^* (D \cup D \times D)$, which ensures that
  there are exactly $i$ context switches.
  Using \Cref{lem:mcfg-closure-poly}\ref{item:mcfg-nfa-intersect}, we obtain $k+1$ grammars
  $\grammar_{a,\pref,0}$ to $\grammar_{a,\pref,k}$.
  Third, we turn each \mcfg~$\grammar_{a,\pref,i}$ into a downward-closure-faithful accelerated \mcfg~$\grammar'_{a,\pref,i}$,
  as guaranteed by \Cref{lem:sigma-grammar-conciseness}.
  Fourth, for each $\grammar'_{a,\pref,i}$ we compute an NFA $\cA_i$ for the downward closure
  $\parikhalphdownck[i+2]{\langof{\grammar_{a,\pref,i}}}{D \cup D \times D}$ using
  \Cref{lem:parikh-dcl}.
  This preserves the first and last letter in $D \cup D \times D$,
  as well as the $i$ context-switches in between.
  Finally, we construct an NFA $\cA_a$ for the union $\bigcup_{i=0}^{k} \langof{\cA_{a,i}}$
  intersected with the regular language
  $aD \big(\Sigma \cup (D \times D)\big)^* (D \cup D \times D)$.
  The last intersection is to ensure that words are of the correct form
  for a process language of a $\CP$.
  
  Almost every step here takes polynomial time; regarding the construction of
  accelerated grammars $\grammar'_{a,\pref,i}$,
	\refApp{\cref{lem:accel-mcfg-size} in \cref{app:conciseness-step}} provides a more detailed time analysis.
  The exception is invoking \Cref{lem:parikh-dcl}, which introduces an exponential blow-up.
  However, no two applications of said lemma are in sequence,
  so a single exponential blow-up suffices.

  Replacing $\grammar_a$ with $\cA_a$ in $\ap$ for each $a \in \Sigma$,
  we obtain a concurrent program $\ap'$ over the regular languages.
  Using the algorithm from \cite{AtigBQ11}, we can decide $k$-context-bounded safety
  in $\EXPSPACE$ in this setting (simple reduction to coverability in vector addition systems).
  Since $\ap'$ may be exponentially larger than $\ap$ and the unary-encoded $k$
  due to the aforementioned blow-up, we get an overall complexity of $\TWOEXPSPACE$.
  
  Let us now argue correctness of the construction.
  It is known that $k$-context-bounded reachability for $\CP$ is preserved under taking the
  downward closure that considers only prefixes with up to $k$ context switches,
  and is not allowed to lose those context switches \cite{AtigBQ11,DBLP:conf/icalp/0001GMTZ23}.
  Thus we only need to argue that $k$-context-bounded reachability is still preserved
  when we arbitrarily swap the infixes over process names $\Sigma$ for Parikh-equivalent infixes.
  This, however, directly follows from the semantics of $\CP$:
  any time a $\Sigma^*$-infix from a language $L_a$ is considered for the step-relation,
	we first take its Parikh image (see \refApp{\Cref{app:programs-formal-def}}),
  meaning the actual order of the letters within such an infix never matters.
\end{proof}

\myparagraph{Compression through Parikh images} It remains to prove
\cref{lem:parikh-dcl}.
The idea is to construct an NFA that constructs a derivation tree of the
accelerated \mcfg~$\grammar$ on the fly.  To this end, the automaton picks
productions nondeterministically, and explores the derivation tree via
depth-first search, storing only the current branch.  During this, variables of
nonterminals are updated as they become assigned with terminal words.  The key
trick to obtaining the exponential (rather than doubly-exponential) size bound
is that, instead of storing these terminal words directly, it suffices to
\emph{store their Parikh images}.  Note that for a word $w$ over an alphabet
$\Gamma$, the Parikh image of $w$ can be stored in $O(|\Gamma|\cdot \log |w|)$
bits, by using binary encodings, whereas storing $w$ completely would require
at least $|w|$ bits. Through this compression, the automaton only needs to
maintain polynomially many (instead of exponentially many) bits in its state.
This is used when storing infixes consisting of letters in $\Sigma$ or
subalphabets of $\Sigma$.

Upon completely exploring the derivation tree, the NFA outputs any word that matches
the letters and Parikh images stored for the starting nonterminal and then accepts.
Due to our requirement of effectively $(c \sizeof{\grammar})$-bounded tree height,
we only need to consider branches of this polynomial length,
which also keeps the Parikh images small enough that this NFA only needs exponentially many states.
The complete proof can be found in \refApp{\cref{app:parikh-dcl}}.

\section{Conclusion}\label{sec:conclusion}
Our results establish MCFG as a conceptually simple model that captures a wide
range of decidable underapproximations of MPDA and other models. This calls for a broader investigation into algorithmic properties of MCFG. The MCFL reachability problem and the membership problem have been studied in~\cite{MCFGPOPL25,DBLP:journals/ci/KajiNSK94,kaji1992a,kaji1992b} (see \cpageref{preliminaries:complexity-membership}). Another interesting direction is \emph{refinement verification}: Can we decide whether the language of a given MCFG is included in a Dyck language? Inclusion in Dyck languages can be used to verify implementations of reference counting or programs that generate code~\cite{DBLP:journals/pacmpl/BaumannGMTZ23}. In the case of MPDA with bounded context-switching (i.e.~an example of a bounded-\streewidth\ underapproximation), this problem was recently shown to be $\coNP$-complete (if the context bound is part of the input)~\cite{DBLP:journals/pacmpl/BaumannGMTZ23}. By our results, an algorithm for MCFG would apply to \emph{all} bounded-\streewidth\ underapproximations.

Our downward closure construction yields various algorithms for
verifying complex systems whose components are specified by
MSO-definable bounded-\streewidth\ languages: one example is our application in
\cref{sec:applications}. As mentioned in the introduction, another example is
parameterized verification of shared memory systems with non-atomic reads and
writes~\cite{DBLP:conf/fsttcs/Hague11,DBLP:conf/concur/TorreMW15}. Yet another
example concerns certain networks of infinite-state systems communicating via
lossy channels~\cite{DBLP:conf/concur/AtigBT08}. Specifically, our results
yield elementary upper bounds for safety verification where so far, only
non-primitive recursive upper bounds were known. This calls for further
investigation into the exact complexity of these problems.

We leave open whether there is a doubly-exponential downward closure construction for MCFG when the dimension is part of the input. In that setting, our construction is five-fold exponential when applying Jecker's recent Ramsey bounds for $n$-point transformation monoids~\cite{DBLP:conf/stacs/Jecker21}. 

\begin{acks}
	We are grateful to the anonymous reviewers for their helpful suggestions on the presentation.

\raisebox{-8pt}[0pt][0pt]{\includegraphics[height=0.8cm]{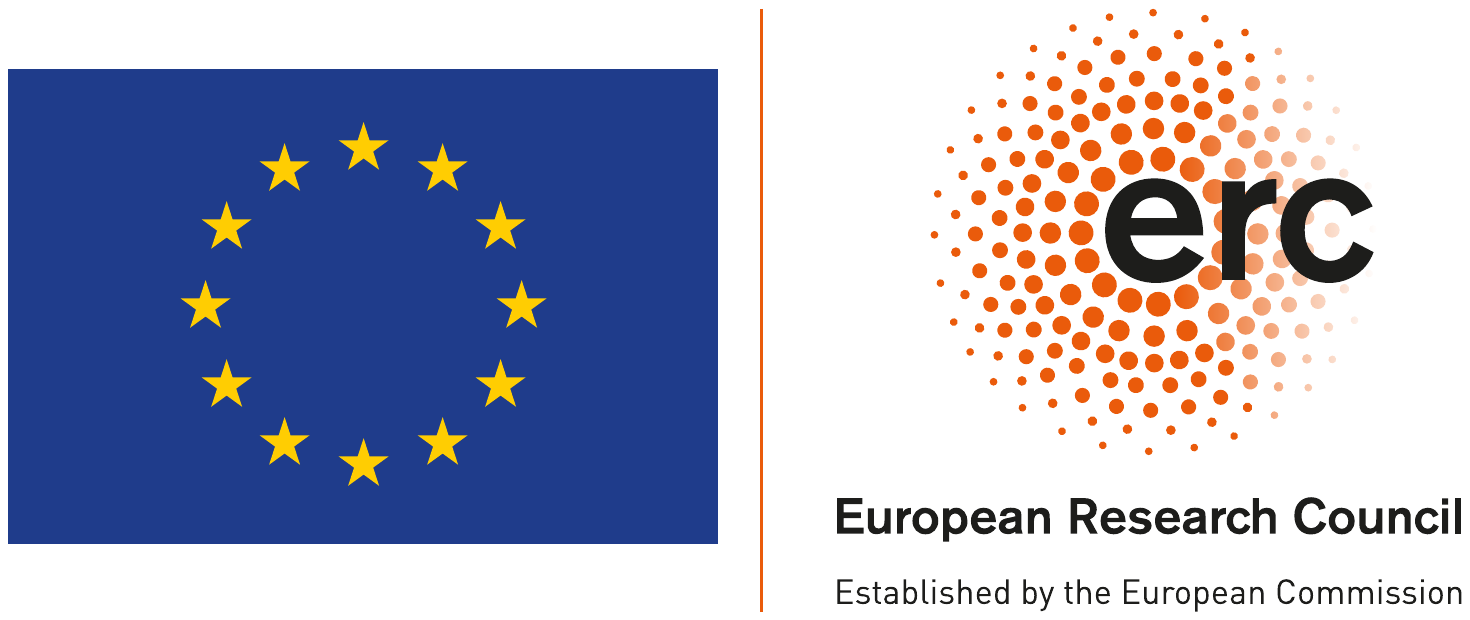}}
Funded by the European Union (\grantsponsor{501100000781}{ERC}{http://dx.doi.org/10.13039/501100000781}, FINABIS, \grantnum[https://doi.org/10.3030/101077902]{501100000781}{101077902}).
Views and opinions expressed are however those of the authors only and do not necessarily reflect those of the European Union
or the European Research Council Executive Agency.
Neither the European Union nor the granting authority can be held responsible for them.
\end{acks}

\label{beforebibliography}
\newoutputstream{pages}
\openoutputfile{main.pages.ctr}{pages}
\addtostream{pages}{\getpagerefnumber{beforebibliography}}
\closeoutputstream{pages}
\bibliographystyle{ACM-Reference-Format}
\bibliography{treewidth-downward.bib}


\begin{thebibliography}{82}

%%% the Web address.

\ifx \showCODEN    \undefined \def \showCODEN     #1{\unskip}     \fi
\ifx \showDOI      \undefined \def \showDOI       #1{#1}\fi
\ifx \showISBNx    \undefined \def \showISBNx     #1{\unskip}     \fi
\ifx \showISBNxiii \undefined \def \showISBNxiii  #1{\unskip}     \fi
\ifx \showISSN     \undefined \def \showISSN      #1{\unskip}     \fi
\ifx \showLCCN     \undefined \def \showLCCN      #1{\unskip}     \fi
\ifx \shownote     \undefined \def \shownote      #1{#1}          \fi
\ifx \showarticletitle \undefined \def \showarticletitle #1{#1}   \fi
\ifx \showURL      \undefined \def \showURL       {\relax}        \fi
% invisible to TeX
\providecommand\bibfield[2]{#2}
\providecommand\bibinfo[2]{#2}
\providecommand\natexlab[1]{#1}
\providecommand\showeprint[2][]{arXiv:#2}

\bibitem[\protect\citeauthoryear{Aho}{Aho}{1968}]%
        {aho1968indexed}
\bibfield{author}{\bibinfo{person}{Alfred~V Aho}.}
  \bibinfo{year}{1968}\natexlab{}.
\newblock \showarticletitle{Indexed grammars—an extension of context-free
  grammars}.
\newblock \bibinfo{journal}{\emph{Journal of the ACM (JACM)}}
  \bibinfo{volume}{15}, \bibinfo{number}{4} (\bibinfo{year}{1968}),
  \bibinfo{pages}{647--671}.
\newblock
\urldef\tempurl%
\url{https://doi.org/10.1145/321479.321488}
\showDOI{\tempurl}


\bibitem[\protect\citeauthoryear{Aiswarya, Gastin, and {Narayan
  Kumar}}{Aiswarya et~al\mbox{.}}{2014}]%
        {AiswaryaGK14}
\bibfield{author}{\bibinfo{person}{C. Aiswarya}, \bibinfo{person}{Paul Gastin},
  {and} \bibinfo{person}{K. {Narayan Kumar}}.} \bibinfo{year}{2014}\natexlab{}.
\newblock \showarticletitle{Verifying Communicating Multi-pushdown Systems via
  Split-Width}. In \bibinfo{booktitle}{\emph{Automated Technology for
  Verification and Analysis - 12th International Symposium, {ATVA} 2014,
  Sydney, NSW, Australia, November 3-7, 2014, Proceedings}}
  \emph{(\bibinfo{series}{Lecture Notes in Computer Science},
  Vol.~\bibinfo{volume}{8837})}, \bibfield{editor}{\bibinfo{person}{Franck
  Cassez} {and} \bibinfo{person}{Jean{-}Fran{\c{c}}ois Raskin}} (Eds.).
  \bibinfo{publisher}{Springer}, \bibinfo{pages}{1--17}.
\newblock
\urldef\tempurl%
\url{https://doi.org/10.1007/978-3-319-11936-6\_1}
\showDOI{\tempurl}


\bibitem[\protect\citeauthoryear{Akshay, Gastin, and Krishna}{Akshay
  et~al\mbox{.}}{2016}]%
        {AkshayGK16}
\bibfield{author}{\bibinfo{person}{S. Akshay}, \bibinfo{person}{Paul Gastin},
  {and} \bibinfo{person}{Shankara~Narayanan Krishna}.}
  \bibinfo{year}{2016}\natexlab{}.
\newblock \showarticletitle{Analyzing Timed Systems Using Tree Automata}. In
  \bibinfo{booktitle}{\emph{27th International Conference on Concurrency
  Theory, {CONCUR} 2016, August 23-26, 2016, Qu{\'{e}}bec City, Canada}}
  \emph{(\bibinfo{series}{LIPIcs}, Vol.~\bibinfo{volume}{59})},
  \bibfield{editor}{\bibinfo{person}{Jos{\'{e}}e Desharnais} {and}
  \bibinfo{person}{Radha Jagadeesan}} (Eds.). \bibinfo{publisher}{Schloss
  Dagstuhl - Leibniz-Zentrum f{\"{u}}r Informatik},
  \bibinfo{pages}{27:1--27:14}.
\newblock
\urldef\tempurl%
\url{https://doi.org/10.4230/LIPICS.CONCUR.2016.27}
\showDOI{\tempurl}


\bibitem[\protect\citeauthoryear{Akshay, Gastin, and Krishna}{Akshay
  et~al\mbox{.}}{2018}]%
        {AkshayGK18}
\bibfield{author}{\bibinfo{person}{S. Akshay}, \bibinfo{person}{Paul Gastin},
  {and} \bibinfo{person}{Shankara~Narayanan Krishna}.}
  \bibinfo{year}{2018}\natexlab{}.
\newblock \showarticletitle{Analyzing Timed Systems Using Tree Automata}.
\newblock \bibinfo{journal}{\emph{Log. Methods Comput. Sci.}}
  \bibinfo{volume}{14}, \bibinfo{number}{2} (\bibinfo{year}{2018}).
\newblock
\urldef\tempurl%
\url{https://doi.org/10.23638/LMCS-14(2:8)2018}
\showDOI{\tempurl}


\bibitem[\protect\citeauthoryear{Akshay, Gastin, Krishna, and Sarkar}{Akshay
  et~al\mbox{.}}{2017}]%
        {AkshayGKS17}
\bibfield{author}{\bibinfo{person}{S. Akshay}, \bibinfo{person}{Paul Gastin},
  \bibinfo{person}{Shankara~Narayanan Krishna}, {and} \bibinfo{person}{Ilias
  Sarkar}.} \bibinfo{year}{2017}\natexlab{}.
\newblock \showarticletitle{Towards an Efficient Tree Automata Based Technique
  for Timed Systems}. In \bibinfo{booktitle}{\emph{28th International
  Conference on Concurrency Theory, {CONCUR} 2017, September 5-8, 2017, Berlin,
  Germany}} \emph{(\bibinfo{series}{LIPIcs}, Vol.~\bibinfo{volume}{85})},
  \bibfield{editor}{\bibinfo{person}{Roland Meyer} {and} \bibinfo{person}{Uwe
  Nestmann}} (Eds.). \bibinfo{publisher}{Schloss Dagstuhl - Leibniz-Zentrum
  f{\"{u}}r Informatik}, \bibinfo{pages}{39:1--39:15}.
\newblock
\urldef\tempurl%
\url{https://doi.org/10.4230/LIPICS.CONCUR.2017.39}
\showDOI{\tempurl}


\bibitem[\protect\citeauthoryear{Alur and Madhusudan}{Alur and
  Madhusudan}{2006}]%
        {AlurM06}
\bibfield{author}{\bibinfo{person}{Rajeev Alur} {and} \bibinfo{person}{P.
  Madhusudan}.} \bibinfo{year}{2006}\natexlab{}.
\newblock \showarticletitle{Adding Nesting Structure to Words}. In
  \bibinfo{booktitle}{\emph{Developments in Language Theory, 10th International
  Conference, {DLT} 2006, Santa Barbara, CA, USA, June 26-29, 2006,
  Proceedings}} \emph{(\bibinfo{series}{Lecture Notes in Computer Science},
  Vol.~\bibinfo{volume}{4036})}, \bibfield{editor}{\bibinfo{person}{Oscar~H.
  Ibarra} {and} \bibinfo{person}{Zhe Dang}} (Eds.).
  \bibinfo{publisher}{Springer}, \bibinfo{pages}{1--13}.
\newblock
\urldef\tempurl%
\url{https://doi.org/10.1007/11779148\_1}
\showDOI{\tempurl}


\bibitem[\protect\citeauthoryear{Anand and Zetzsche}{Anand and
  Zetzsche}{2023}]%
        {DBLP:conf/concur/AnandZ23}
\bibfield{author}{\bibinfo{person}{Ashwani Anand} {and} \bibinfo{person}{Georg
  Zetzsche}.} \bibinfo{year}{2023}\natexlab{}.
\newblock \showarticletitle{Priority Downward Closures}. In
  \bibinfo{booktitle}{\emph{34th International Conference on Concurrency
  Theory, {CONCUR} 2023, September 18-23, 2023, Antwerp, Belgium}}
  \emph{(\bibinfo{series}{LIPIcs}, Vol.~\bibinfo{volume}{279})},
  \bibfield{editor}{\bibinfo{person}{Guillermo~A. P{\'{e}}rez} {and}
  \bibinfo{person}{Jean{-}Fran{\c{c}}ois Raskin}} (Eds.).
  \bibinfo{publisher}{Schloss Dagstuhl - Leibniz-Zentrum f{\"{u}}r Informatik},
  \bibinfo{pages}{39:1--39:18}.
\newblock
\urldef\tempurl%
\url{https://doi.org/10.4230/LIPICS.CONCUR.2023.39}
\showDOI{\tempurl}


\bibitem[\protect\citeauthoryear{Arzt, Rasthofer, Fritz, Bodden, Bartel, Klein,
  Traon, Octeau, and McDaniel}{Arzt et~al\mbox{.}}{2014}]%
        {DBLP:conf/pldi/ArztRFBBKTOM14}
\bibfield{author}{\bibinfo{person}{Steven Arzt}, \bibinfo{person}{Siegfried
  Rasthofer}, \bibinfo{person}{Christian Fritz}, \bibinfo{person}{Eric Bodden},
  \bibinfo{person}{Alexandre Bartel}, \bibinfo{person}{Jacques Klein},
  \bibinfo{person}{Yves~Le Traon}, \bibinfo{person}{Damien Octeau}, {and}
  \bibinfo{person}{Patrick~D. McDaniel}.} \bibinfo{year}{2014}\natexlab{}.
\newblock \showarticletitle{FlowDroid: precise context, flow, field,
  object-sensitive and lifecycle-aware taint analysis for Android apps}. In
  \bibinfo{booktitle}{\emph{{ACM} {SIGPLAN} Conference on Programming Language
  Design and Implementation, {PLDI} '14, Edinburgh, United Kingdom - June 09 -
  11, 2014}}, \bibfield{editor}{\bibinfo{person}{Michael F.~P. O'Boyle} {and}
  \bibinfo{person}{Keshav Pingali}} (Eds.). \bibinfo{publisher}{{ACM}},
  \bibinfo{pages}{259--269}.
\newblock
\urldef\tempurl%
\url{https://doi.org/10.1145/2594291.2594299}
\showDOI{\tempurl}


\bibitem[\protect\citeauthoryear{Atig, Bouajjani, Kumar, and Saivasan}{Atig
  et~al\mbox{.}}{2012}]%
        {DBLP:conf/atva/AtigBKS12}
\bibfield{author}{\bibinfo{person}{Mohamed~Faouzi Atig}, \bibinfo{person}{Ahmed
  Bouajjani}, \bibinfo{person}{K.~Narayan Kumar}, {and}
  \bibinfo{person}{Prakash Saivasan}.} \bibinfo{year}{2012}\natexlab{}.
\newblock \showarticletitle{Linear-Time Model-Checking for Multithreaded
  Programs under Scope-Bounding}. In \bibinfo{booktitle}{\emph{Automated
  Technology for Verification and Analysis - 10th International Symposium,
  {ATVA} 2012, Thiruvananthapuram, India, October 3-6, 2012. Proceedings}}
  \emph{(\bibinfo{series}{Lecture Notes in Computer Science},
  Vol.~\bibinfo{volume}{7561})}, \bibfield{editor}{\bibinfo{person}{Supratik
  Chakraborty} {and} \bibinfo{person}{Madhavan Mukund}} (Eds.).
  \bibinfo{publisher}{Springer}, \bibinfo{pages}{152--166}.
\newblock
\urldef\tempurl%
\url{https://doi.org/10.1007/978-3-642-33386-6\_13}
\showDOI{\tempurl}


\bibitem[\protect\citeauthoryear{Atig, Bouajjani, and Qadeer}{Atig
  et~al\mbox{.}}{2011}]%
        {AtigBQ11}
\bibfield{author}{\bibinfo{person}{Mohamed~Faouzi Atig}, \bibinfo{person}{Ahmed
  Bouajjani}, {and} \bibinfo{person}{Shaz Qadeer}.}
  \bibinfo{year}{2011}\natexlab{}.
\newblock \showarticletitle{Context-Bounded Analysis For Concurrent Programs
  With Dynamic Creation of Threads}.
\newblock \bibinfo{journal}{\emph{Log. Methods Comput. Sci.}}
  \bibinfo{volume}{7}, \bibinfo{number}{4} (\bibinfo{year}{2011}).
\newblock
\urldef\tempurl%
\url{https://doi.org/10.2168/LMCS-7(4:4)2011}
\showDOI{\tempurl}


\bibitem[\protect\citeauthoryear{Atig, Bouajjani, and Touili}{Atig
  et~al\mbox{.}}{2008}]%
        {DBLP:conf/concur/AtigBT08}
\bibfield{author}{\bibinfo{person}{Mohamed~Faouzi Atig}, \bibinfo{person}{Ahmed
  Bouajjani}, {and} \bibinfo{person}{Tayssir Touili}.}
  \bibinfo{year}{2008}\natexlab{}.
\newblock \showarticletitle{On the Reachability Analysis of Acyclic Networks of
  Pushdown Systems}. In \bibinfo{booktitle}{\emph{{CONCUR} 2008 - Concurrency
  Theory, 19th International Conference, {CONCUR} 2008, Toronto, Canada, August
  19-22, 2008. Proceedings}} \emph{(\bibinfo{series}{Lecture Notes in Computer
  Science}, Vol.~\bibinfo{volume}{5201})},
  \bibfield{editor}{\bibinfo{person}{Franck van Breugel} {and}
  \bibinfo{person}{Marsha Chechik}} (Eds.). \bibinfo{publisher}{Springer},
  \bibinfo{pages}{356--371}.
\newblock
\urldef\tempurl%
\url{https://doi.org/10.1007/978-3-540-85361-9\_29}
\showDOI{\tempurl}


\bibitem[\protect\citeauthoryear{Barozzini, Clemente, Colcombet, and
  Parys}{Barozzini et~al\mbox{.}}{2020}]%
        {DBLP:conf/icalp/BarozziniCCP20}
\bibfield{author}{\bibinfo{person}{David Barozzini}, \bibinfo{person}{Lorenzo
  Clemente}, \bibinfo{person}{Thomas Colcombet}, {and} \bibinfo{person}{Pawel
  Parys}.} \bibinfo{year}{2020}\natexlab{}.
\newblock \showarticletitle{Cost Automata, Safe Schemes, and Downward
  Closures}. In \bibinfo{booktitle}{\emph{47th International Colloquium on
  Automata, Languages, and Programming, {ICALP} 2020, July 8-11, 2020,
  Saarbr{\"{u}}cken, Germany (Virtual Conference)}}
  \emph{(\bibinfo{series}{LIPIcs}, Vol.~\bibinfo{volume}{168})},
  \bibfield{editor}{\bibinfo{person}{Artur Czumaj}, \bibinfo{person}{Anuj
  Dawar}, {and} \bibinfo{person}{Emanuela Merelli}} (Eds.).
  \bibinfo{publisher}{Schloss Dagstuhl - Leibniz-Zentrum f{\"{u}}r Informatik},
  \bibinfo{pages}{109:1--109:18}.
\newblock
\urldef\tempurl%
\url{https://doi.org/10.4230/LIPICS.ICALP.2020.109}
\showDOI{\tempurl}


\bibitem[\protect\citeauthoryear{Baumann, Ganardi, Majumdar, Thinniyam, and
  Zetzsche}{Baumann et~al\mbox{.}}{2023a}]%
        {DBLP:conf/icalp/0001GMTZ23}
\bibfield{author}{\bibinfo{person}{Pascal Baumann}, \bibinfo{person}{Moses
  Ganardi}, \bibinfo{person}{Rupak Majumdar}, \bibinfo{person}{Ramanathan~S.
  Thinniyam}, {and} \bibinfo{person}{Georg Zetzsche}.}
  \bibinfo{year}{2023}\natexlab{a}.
\newblock \showarticletitle{Context-Bounded Analysis of Concurrent Programs
  (Invited Talk)}. In \bibinfo{booktitle}{\emph{50th International Colloquium
  on Automata, Languages, and Programming, {ICALP} 2023, July 10-14, 2023,
  Paderborn, Germany}} \emph{(\bibinfo{series}{LIPIcs},
  Vol.~\bibinfo{volume}{261})}, \bibfield{editor}{\bibinfo{person}{Kousha
  Etessami}, \bibinfo{person}{Uriel Feige}, {and} \bibinfo{person}{Gabriele
  Puppis}} (Eds.). \bibinfo{publisher}{Schloss Dagstuhl - Leibniz-Zentrum
  f{\"{u}}r Informatik}, \bibinfo{pages}{3:1--3:16}.
\newblock
\urldef\tempurl%
\url{https://doi.org/10.4230/LIPICS.ICALP.2023.3}
\showDOI{\tempurl}


\bibitem[\protect\citeauthoryear{Baumann, Ganardi, Majumdar, Thinniyam, and
  Zetzsche}{Baumann et~al\mbox{.}}{2023b}]%
        {DBLP:journals/pacmpl/BaumannGMTZ23}
\bibfield{author}{\bibinfo{person}{Pascal Baumann}, \bibinfo{person}{Moses
  Ganardi}, \bibinfo{person}{Rupak Majumdar}, \bibinfo{person}{Ramanathan~S.
  Thinniyam}, {and} \bibinfo{person}{Georg Zetzsche}.}
  \bibinfo{year}{2023}\natexlab{b}.
\newblock \showarticletitle{Context-Bounded Verification of Context-Free
  Specifications}.
\newblock \bibinfo{journal}{\emph{Proc. {ACM} Program. Lang.}}
  \bibinfo{volume}{7}, \bibinfo{number}{{POPL}} (\bibinfo{year}{2023}),
  \bibinfo{pages}{2141--2170}.
\newblock
\urldef\tempurl%
\url{https://doi.org/10.1145/3571266}
\showDOI{\tempurl}


\bibitem[\protect\citeauthoryear{Baumann, Majumdar, Thinniyam, and
  Zetzsche}{Baumann et~al\mbox{.}}{2020}]%
        {DBLP:conf/icalp/BaumannMTZ20}
\bibfield{author}{\bibinfo{person}{Pascal Baumann}, \bibinfo{person}{Rupak
  Majumdar}, \bibinfo{person}{Ramanathan~S. Thinniyam}, {and}
  \bibinfo{person}{Georg Zetzsche}.} \bibinfo{year}{2020}\natexlab{}.
\newblock \showarticletitle{The Complexity of Bounded Context Switching with
  Dynamic Thread Creation}. In \bibinfo{booktitle}{\emph{47th International
  Colloquium on Automata, Languages, and Programming, {ICALP} 2020, July 8-11,
  2020, Saarbr{\"{u}}cken, Germany (Virtual Conference)}}
  \emph{(\bibinfo{series}{LIPIcs}, Vol.~\bibinfo{volume}{168})},
  \bibfield{editor}{\bibinfo{person}{Artur Czumaj}, \bibinfo{person}{Anuj
  Dawar}, {and} \bibinfo{person}{Emanuela Merelli}} (Eds.).
  \bibinfo{publisher}{Schloss Dagstuhl - Leibniz-Zentrum f{\"{u}}r Informatik},
  \bibinfo{pages}{111:1--111:16}.
\newblock
\urldef\tempurl%
\url{https://doi.org/10.4230/LIPICS.ICALP.2020.111}
\showDOI{\tempurl}


\bibitem[\protect\citeauthoryear{Baumann, Majumdar, Thinniyam, and
  Zetzsche}{Baumann et~al\mbox{.}}{2022}]%
        {DBLP:journals/pacmpl/BaumannMTZ22}
\bibfield{author}{\bibinfo{person}{Pascal Baumann}, \bibinfo{person}{Rupak
  Majumdar}, \bibinfo{person}{Ramanathan~S. Thinniyam}, {and}
  \bibinfo{person}{Georg Zetzsche}.} \bibinfo{year}{2022}\natexlab{}.
\newblock \showarticletitle{Context-bounded verification of thread pools}.
\newblock \bibinfo{journal}{\emph{Proc. {ACM} Program. Lang.}}
  \bibinfo{volume}{6}, \bibinfo{number}{{POPL}} (\bibinfo{year}{2022}),
  \bibinfo{pages}{1--28}.
\newblock
\urldef\tempurl%
\url{https://doi.org/10.1145/3498678}
\showDOI{\tempurl}


\bibitem[\protect\citeauthoryear{Berstel}{Berstel}{1979}]%
        {berstel2013transductions}
\bibfield{author}{\bibinfo{person}{Jean Berstel}.}
  \bibinfo{year}{1979}\natexlab{}.
\newblock \bibinfo{booktitle}{\emph{Transductions and context-free languages}}.
\newblock \bibinfo{publisher}{Springer-Verlag}.
\newblock
\urldef\tempurl%
\url{https://doi.org/10.1007/978-3-663-09367-1}
\showDOI{\tempurl}


\bibitem[\protect\citeauthoryear{Breveglieri, Cherubini, Citrini, and
  Crespi{-}Reghizzi}{Breveglieri et~al\mbox{.}}{1996}]%
        {DBLP:journals/ijfcs/BreveglieriCCC96}
\bibfield{author}{\bibinfo{person}{Luca Breveglieri},
  \bibinfo{person}{Alessandra Cherubini}, \bibinfo{person}{Claudio Citrini},
  {and} \bibinfo{person}{Stefano Crespi{-}Reghizzi}.}
  \bibinfo{year}{1996}\natexlab{}.
\newblock \showarticletitle{Multi-Push-Down Languages and Grammars}.
\newblock \bibinfo{journal}{\emph{Int. J. Found. Comput. Sci.}}
  \bibinfo{volume}{7}, \bibinfo{number}{3} (\bibinfo{year}{1996}),
  \bibinfo{pages}{253--292}.
\newblock
\urldef\tempurl%
\url{https://doi.org/10.1142/S0129054196000191}
\showDOI{\tempurl}


\bibitem[\protect\citeauthoryear{Chatterjee, Goharshady, Ibsen{-}Jensen, and
  Pavlogiannis}{Chatterjee et~al\mbox{.}}{2020}]%
        {DBLP:conf/esop/ChatterjeeGIP20}
\bibfield{author}{\bibinfo{person}{Krishnendu Chatterjee},
  \bibinfo{person}{Amir~Kafshdar Goharshady}, \bibinfo{person}{Rasmus
  Ibsen{-}Jensen}, {and} \bibinfo{person}{Andreas Pavlogiannis}.}
  \bibinfo{year}{2020}\natexlab{}.
\newblock \showarticletitle{Optimal and Perfectly Parallel Algorithms for
  On-demand Data-Flow Analysis}. In \bibinfo{booktitle}{\emph{Programming
  Languages and Systems - 29th European Symposium on Programming, {ESOP} 2020,
  Held as Part of the European Joint Conferences on Theory and Practice of
  Software, {ETAPS} 2020, Dublin, Ireland, April 25-30, 2020, Proceedings}}
  \emph{(\bibinfo{series}{Lecture Notes in Computer Science},
  Vol.~\bibinfo{volume}{12075})}, \bibfield{editor}{\bibinfo{person}{Peter
  M{\"{u}}ller}} (Ed.). \bibinfo{publisher}{Springer},
  \bibinfo{pages}{112--140}.
\newblock
\urldef\tempurl%
\url{https://doi.org/10.1007/978-3-030-44914-8\_5}
\showDOI{\tempurl}


\bibitem[\protect\citeauthoryear{Chatterjee, Goharshady, and
  Pavlogiannis}{Chatterjee et~al\mbox{.}}{2017}]%
        {DBLP:conf/atva/ChatterjeeGP17}
\bibfield{author}{\bibinfo{person}{Krishnendu Chatterjee},
  \bibinfo{person}{Amir~Kafshdar Goharshady}, {and} \bibinfo{person}{Andreas
  Pavlogiannis}.} \bibinfo{year}{2017}\natexlab{}.
\newblock \showarticletitle{JTDec: {A} Tool for Tree Decompositions in Soot}.
  In \bibinfo{booktitle}{\emph{Automated Technology for Verification and
  Analysis - 15th International Symposium, {ATVA} 2017, Pune, India, October
  3-6, 2017, Proceedings}} \emph{(\bibinfo{series}{Lecture Notes in Computer
  Science}, Vol.~\bibinfo{volume}{10482})},
  \bibfield{editor}{\bibinfo{person}{Deepak D'Souza} {and}
  \bibinfo{person}{K.~Narayan Kumar}} (Eds.). \bibinfo{publisher}{Springer},
  \bibinfo{pages}{59--66}.
\newblock
\urldef\tempurl%
\url{https://doi.org/10.1007/978-3-319-68167-2\_4}
\showDOI{\tempurl}


\bibitem[\protect\citeauthoryear{Chatterjee, Ibsen{-}Jensen, Goharshady, and
  Pavlogiannis}{Chatterjee et~al\mbox{.}}{2018}]%
        {DBLP:journals/toplas/ChatterjeeIGP18}
\bibfield{author}{\bibinfo{person}{Krishnendu Chatterjee},
  \bibinfo{person}{Rasmus Ibsen{-}Jensen}, \bibinfo{person}{Amir~Kafshdar
  Goharshady}, {and} \bibinfo{person}{Andreas Pavlogiannis}.}
  \bibinfo{year}{2018}\natexlab{}.
\newblock \showarticletitle{Algorithms for Algebraic Path Properties in
  Concurrent Systems of Constant Treewidth Components}.
\newblock \bibinfo{journal}{\emph{{ACM} Trans. Program. Lang. Syst.}}
  \bibinfo{volume}{40}, \bibinfo{number}{3} (\bibinfo{year}{2018}),
  \bibinfo{pages}{9:1--9:43}.
\newblock
\urldef\tempurl%
\url{https://doi.org/10.1145/3210257}
\showDOI{\tempurl}


\bibitem[\protect\citeauthoryear{Chatterjee, Ibsen{-}Jensen, and
  Pavlogiannis}{Chatterjee et~al\mbox{.}}{2016}]%
        {DBLP:conf/esa/ChatterjeeIP16}
\bibfield{author}{\bibinfo{person}{Krishnendu Chatterjee},
  \bibinfo{person}{Rasmus Ibsen{-}Jensen}, {and} \bibinfo{person}{Andreas
  Pavlogiannis}.} \bibinfo{year}{2016}\natexlab{}.
\newblock \showarticletitle{Optimal Reachability and a Space-Time Tradeoff for
  Distance Queries in Constant-Treewidth Graphs}. In
  \bibinfo{booktitle}{\emph{24th Annual European Symposium on Algorithms, {ESA}
  2016, August 22-24, 2016, Aarhus, Denmark}} \emph{(\bibinfo{series}{LIPIcs},
  Vol.~\bibinfo{volume}{57})}, \bibfield{editor}{\bibinfo{person}{Piotr
  Sankowski} {and} \bibinfo{person}{Christos~D. Zaroliagis}} (Eds.).
  \bibinfo{publisher}{Schloss Dagstuhl - Leibniz-Zentrum f{\"{u}}r Informatik},
  \bibinfo{pages}{28:1--28:17}.
\newblock
\urldef\tempurl%
\url{https://doi.org/10.4230/LIPICS.ESA.2016.28}
\showDOI{\tempurl}


\bibitem[\protect\citeauthoryear{Chatterjee, Ibsen{-}Jensen, Pavlogiannis, and
  Goyal}{Chatterjee et~al\mbox{.}}{2015}]%
        {DBLP:conf/popl/ChatterjeeIPG15}
\bibfield{author}{\bibinfo{person}{Krishnendu Chatterjee},
  \bibinfo{person}{Rasmus Ibsen{-}Jensen}, \bibinfo{person}{Andreas
  Pavlogiannis}, {and} \bibinfo{person}{Prateesh Goyal}.}
  \bibinfo{year}{2015}\natexlab{}.
\newblock \showarticletitle{Faster Algorithms for Algebraic Path Properties in
  Recursive State Machines with Constant Treewidth}. In
  \bibinfo{booktitle}{\emph{Proceedings of the 42nd Annual {ACM}
  {SIGPLAN-SIGACT} Symposium on Principles of Programming Languages, {POPL}
  2015, Mumbai, India, January 15-17, 2015}},
  \bibfield{editor}{\bibinfo{person}{Sriram~K. Rajamani} {and}
  \bibinfo{person}{David Walker}} (Eds.). \bibinfo{publisher}{{ACM}},
  \bibinfo{pages}{97--109}.
\newblock
\urldef\tempurl%
\url{https://doi.org/10.1145/2676726.2676979}
\showDOI{\tempurl}


\bibitem[\protect\citeauthoryear{Cheng and Hwu}{Cheng and Hwu}{2000}]%
        {DBLP:conf/pldi/ChengH00}
\bibfield{author}{\bibinfo{person}{Ben{-}Chung Cheng} {and}
  \bibinfo{person}{Wen{-}mei~W. Hwu}.} \bibinfo{year}{2000}\natexlab{}.
\newblock \showarticletitle{Modular interprocedural pointer analysis using
  access paths: design, implementation, and evaluation}. In
  \bibinfo{booktitle}{\emph{Proceedings of the 2000 {ACM} {SIGPLAN} Conference
  on Programming Language Design and Implementation (PLDI), Vancouver, Britith
  Columbia, Canada, June 18-21, 2000}},
  \bibfield{editor}{\bibinfo{person}{Monica~S. Lam}} (Ed.).
  \bibinfo{publisher}{{ACM}}, \bibinfo{pages}{57--69}.
\newblock
\urldef\tempurl%
\url{https://doi.org/10.1145/349299.349311}
\showDOI{\tempurl}


\bibitem[\protect\citeauthoryear{Clemente, Parys, Salvati, and
  Walukiewicz}{Clemente et~al\mbox{.}}{2016}]%
        {DBLP:conf/lics/ClementePSW16}
\bibfield{author}{\bibinfo{person}{Lorenzo Clemente}, \bibinfo{person}{Pawel
  Parys}, \bibinfo{person}{Sylvain Salvati}, {and} \bibinfo{person}{Igor
  Walukiewicz}.} \bibinfo{year}{2016}\natexlab{}.
\newblock \showarticletitle{The Diagonal Problem for Higher-Order Recursion
  Schemes is Decidable}. In \bibinfo{booktitle}{\emph{Proceedings of the 31st
  Annual {ACM/IEEE} Symposium on Logic in Computer Science, {LICS} '16, New
  York, NY, USA, July 5-8, 2016}}, \bibfield{editor}{\bibinfo{person}{Martin
  Grohe}, \bibinfo{person}{Eric Koskinen}, {and} \bibinfo{person}{Natarajan
  Shankar}} (Eds.). \bibinfo{publisher}{{ACM}}, \bibinfo{pages}{96--105}.
\newblock
\urldef\tempurl%
\url{https://doi.org/10.1145/2933575.2934527}
\showDOI{\tempurl}


\bibitem[\protect\citeauthoryear{Conrado, Kjelstr{\o}m, van~de Pol, and
  Pavlogiannis}{Conrado et~al\mbox{.}}{2025}]%
        {MCFGPOPL25}
\bibfield{author}{\bibinfo{person}{Giovanna~Kobus Conrado},
  \bibinfo{person}{Adam~Husted Kjelstr{\o}m}, \bibinfo{person}{Jaco van~de
  Pol}, {and} \bibinfo{person}{Andreas Pavlogiannis}.}
  \bibinfo{year}{2025}\natexlab{}.
\newblock \showarticletitle{Program Analysis via Multiple Context Free Language
  Reachability}.
\newblock \bibinfo{journal}{\emph{Proc. {ACM} Program. Lang.}}
  \bibinfo{volume}{9}, \bibinfo{number}{{POPL}} (\bibinfo{year}{2025}),
  \bibinfo{pages}{509--538}.
\newblock
\urldef\tempurl%
\url{https://doi.org/10.1145/3704854}
\showDOI{\tempurl}


\bibitem[\protect\citeauthoryear{Conrado and Pavlogiannis}{Conrado and
  Pavlogiannis}{2025}]%
        {DBLP:journals/sttt/ConradoP25}
\bibfield{author}{\bibinfo{person}{Giovanna~Kobus Conrado} {and}
  \bibinfo{person}{Andreas Pavlogiannis}.} \bibinfo{year}{2025}\natexlab{}.
\newblock \showarticletitle{CFL-based methods for approximating interleaved
  Dyck reachability}.
\newblock \bibinfo{journal}{\emph{Int. J. Softw. Tools Technol. Transf.}}
  \bibinfo{volume}{27}, \bibinfo{number}{2} (\bibinfo{year}{2025}),
  \bibinfo{pages}{255--266}.
\newblock
\urldef\tempurl%
\url{https://doi.org/10.1007/S10009-025-00787-0}
\showDOI{\tempurl}


\bibitem[\protect\citeauthoryear{Courcelle}{Courcelle}{1989}]%
        {Courcelle89}
\bibfield{author}{\bibinfo{person}{Bruno Courcelle}.}
  \bibinfo{year}{1989}\natexlab{}.
\newblock \showarticletitle{The Monadic Second-Order Logic of Graphs, {II:}
  Infinite Graphs of Bounded Width}.
\newblock \bibinfo{journal}{\emph{Math. Syst. Theory}} \bibinfo{volume}{21},
  \bibinfo{number}{4} (\bibinfo{year}{1989}), \bibinfo{pages}{187--221}.
\newblock
\urldef\tempurl%
\url{https://doi.org/10.1007/BF02088013}
\showDOI{\tempurl}


\bibitem[\protect\citeauthoryear{Courcelle}{Courcelle}{1991}]%
        {DBLP:journals/eatcs/Courcelle91}
\bibfield{author}{\bibinfo{person}{Bruno Courcelle}.}
  \bibinfo{year}{1991}\natexlab{}.
\newblock \showarticletitle{On Constructing Obstruction Sets of Words}.
\newblock \bibinfo{journal}{\emph{Bull. {EATCS}}}  \bibinfo{volume}{44}
  (\bibinfo{year}{1991}), \bibinfo{pages}{178--186}.
\newblock


\bibitem[\protect\citeauthoryear{Courcelle}{Courcelle}{2010}]%
        {Cou10}
\bibfield{author}{\bibinfo{person}{Bruno Courcelle}.}
  \bibinfo{year}{2010}\natexlab{}.
\newblock \showarticletitle{{Special tree-width and the verification of monadic
  second-order graph properties}}. In \bibinfo{booktitle}{\emph{IARCS Annual
  Conference on Foundations of Software Technology and Theoretical Computer
  Science (FSTTCS 2010)}} \emph{(\bibinfo{series}{Leibniz International
  Proceedings in Informatics (LIPIcs)}, Vol.~\bibinfo{volume}{8})},
  \bibfield{editor}{\bibinfo{person}{Kamal Lodaya} {and} \bibinfo{person}{Meena
  Mahajan}} (Eds.). \bibinfo{publisher}{Schloss Dagstuhl -- Leibniz-Zentrum
  f{\"u}r Informatik}, \bibinfo{address}{Dagstuhl, Germany},
  \bibinfo{pages}{13--29}.
\newblock
\showISBNx{978-3-939897-23-1}
\showISSN{1868-8969}
\urldef\tempurl%
\url{https://doi.org/10.4230/LIPIcs.FSTTCS.2010.13}
\showDOI{\tempurl}


\bibitem[\protect\citeauthoryear{Cyriac}{Cyriac}{2014}]%
        {Cyriac14}
\bibfield{author}{\bibinfo{person}{Aiswarya Cyriac}.}
  \bibinfo{year}{2014}\natexlab{}.
\newblock \emph{\bibinfo{title}{Verification of communicating recursive
  programs via split-width. (V{\'{e}}rification de programmes r{\'{e}}cursifs
  et communicants via split-width)}}.
\newblock \bibinfo{thesistype}{Ph.D. Dissertation}.
  \bibinfo{school}{{\'{E}}cole normale sup{\'{e}}rieure de Cachan, France}.
\newblock
\urldef\tempurl%
\url{https://tel.archives-ouvertes.fr/tel-01015561}
\showURL{%
\tempurl}


\bibitem[\protect\citeauthoryear{Cyriac, Gastin, and {Narayan Kumar}}{Cyriac
  et~al\mbox{.}}{2012}]%
        {CyriacGK12}
\bibfield{author}{\bibinfo{person}{Aiswarya Cyriac}, \bibinfo{person}{Paul
  Gastin}, {and} \bibinfo{person}{K. {Narayan Kumar}}.}
  \bibinfo{year}{2012}\natexlab{}.
\newblock \showarticletitle{{MSO} Decidability of Multi-Pushdown Systems via
  Split-Width}. In \bibinfo{booktitle}{\emph{{CONCUR} 2012 - Concurrency Theory
  - 23rd International Conference, {CONCUR} 2012, Newcastle upon Tyne, UK,
  September 4-7, 2012. Proceedings}} \emph{(\bibinfo{series}{Lecture Notes in
  Computer Science}, Vol.~\bibinfo{volume}{7454})},
  \bibfield{editor}{\bibinfo{person}{Maciej Koutny} {and} \bibinfo{person}{Irek
  Ulidowski}} (Eds.). \bibinfo{publisher}{Springer}, \bibinfo{pages}{547--561}.
\newblock
\urldef\tempurl%
\url{https://doi.org/10.1007/978-3-642-32940-1\_38}
\showDOI{\tempurl}


\bibitem[\protect\citeauthoryear{Ding and Zhang}{Ding and Zhang}{2023}]%
        {DBLP:conf/sas/DingZ23}
\bibfield{author}{\bibinfo{person}{Shuo Ding} {and} \bibinfo{person}{Qirun
  Zhang}.} \bibinfo{year}{2023}\natexlab{}.
\newblock \showarticletitle{Mutual Refinements of Context-Free Language
  Reachability}. In \bibinfo{booktitle}{\emph{Static Analysis - 30th
  International Symposium, {SAS} 2023, Cascais, Portugal, October 22-24, 2023,
  Proceedings}} \emph{(\bibinfo{series}{Lecture Notes in Computer Science},
  Vol.~\bibinfo{volume}{14284})}, \bibfield{editor}{\bibinfo{person}{Manuel~V.
  Hermenegildo} {and} \bibinfo{person}{Jos{\'{e}}~F. Morales}} (Eds.).
  \bibinfo{publisher}{Springer}, \bibinfo{pages}{231--258}.
\newblock
\urldef\tempurl%
\url{https://doi.org/10.1007/978-3-031-44245-2\_12}
\showDOI{\tempurl}


\bibitem[\protect\citeauthoryear{Ehrenfeucht and Rozenberg}{Ehrenfeucht and
  Rozenberg}{1977}]%
        {ehrenfeucht1977some}
\bibfield{author}{\bibinfo{person}{Andrzej Ehrenfeucht} {and}
  \bibinfo{person}{Grzegorz Rozenberg}.} \bibinfo{year}{1977}\natexlab{}.
\newblock \showarticletitle{On some context free languages that are not
  deterministic ETOL languages}.
\newblock \bibinfo{journal}{\emph{RAIRO. Informatique th{\'e}orique}}
  \bibinfo{volume}{11}, \bibinfo{number}{4} (\bibinfo{year}{1977}),
  \bibinfo{pages}{273--291}.
\newblock
\urldef\tempurl%
\url{https://doi.org/10.1051/ITA/1977110402731}
\showDOI{\tempurl}


\bibitem[\protect\citeauthoryear{Engelfriet and Heyker}{Engelfriet and
  Heyker}{1991}]%
        {DBLP:journals/jcss/EngelfrietH91}
\bibfield{author}{\bibinfo{person}{Joost Engelfriet} {and}
  \bibinfo{person}{Linda Heyker}.} \bibinfo{year}{1991}\natexlab{}.
\newblock \showarticletitle{The String Generating Power of Context-Free
  Hypergraph Grammars}.
\newblock \bibinfo{journal}{\emph{J. Comput. Syst. Sci.}} \bibinfo{volume}{43},
  \bibinfo{number}{2} (\bibinfo{year}{1991}), \bibinfo{pages}{328--360}.
\newblock
\urldef\tempurl%
\url{https://doi.org/10.1016/0022-0000(91)90018-Z}
\showDOI{\tempurl}


\bibitem[\protect\citeauthoryear{Engelfriet, Rozenberg, and Slutzki}{Engelfriet
  et~al\mbox{.}}{1980}]%
        {DBLP:journals/jcss/EngelfrietRS80}
\bibfield{author}{\bibinfo{person}{Joost Engelfriet}, \bibinfo{person}{Grzegorz
  Rozenberg}, {and} \bibinfo{person}{Giora Slutzki}.}
  \bibinfo{year}{1980}\natexlab{}.
\newblock \showarticletitle{Tree Transducers, {L} Systems, and Two-Way
  Machines}.
\newblock \bibinfo{journal}{\emph{J. Comput. Syst. Sci.}} \bibinfo{volume}{20},
  \bibinfo{number}{2} (\bibinfo{year}{1980}), \bibinfo{pages}{150--202}.
\newblock
\urldef\tempurl%
\url{https://doi.org/10.1016/0022-0000(80)90058-6}
\showDOI{\tempurl}


\bibitem[\protect\citeauthoryear{Engelfriet and Skyum}{Engelfriet and
  Skyum}{1976}]%
        {DBLP:journals/ipl/EngelfrietS76}
\bibfield{author}{\bibinfo{person}{Joost Engelfriet} {and}
  \bibinfo{person}{Sven Skyum}.} \bibinfo{year}{1976}\natexlab{}.
\newblock \showarticletitle{Copying Theorems}.
\newblock \bibinfo{journal}{\emph{Inf. Process. Lett.}} \bibinfo{volume}{4},
  \bibinfo{number}{6} (\bibinfo{year}{1976}), \bibinfo{pages}{157--161}.
\newblock
\urldef\tempurl%
\url{https://doi.org/10.1016/0020-0190(76)90086-7}
\showDOI{\tempurl}


\bibitem[\protect\citeauthoryear{Ganardi, Majumdar, Pavlogiannis,
  Sch{\"{u}}tze, and Zetzsche}{Ganardi et~al\mbox{.}}{2022}]%
        {DBLP:conf/icalp/GanardiMPSZ22}
\bibfield{author}{\bibinfo{person}{Moses Ganardi}, \bibinfo{person}{Rupak
  Majumdar}, \bibinfo{person}{Andreas Pavlogiannis}, \bibinfo{person}{Lia
  Sch{\"{u}}tze}, {and} \bibinfo{person}{Georg Zetzsche}.}
  \bibinfo{year}{2022}\natexlab{}.
\newblock \showarticletitle{Reachability in Bidirected Pushdown {VASS}}. In
  \bibinfo{booktitle}{\emph{49th International Colloquium on Automata,
  Languages, and Programming, {ICALP} 2022, July 4-8, 2022, Paris, France}}
  \emph{(\bibinfo{series}{LIPIcs}, Vol.~\bibinfo{volume}{229})},
  \bibfield{editor}{\bibinfo{person}{Mikolaj Bojanczyk},
  \bibinfo{person}{Emanuela Merelli}, {and} \bibinfo{person}{David~P.
  Woodruff}} (Eds.). \bibinfo{publisher}{Schloss Dagstuhl - Leibniz-Zentrum
  f{\"{u}}r Informatik}, \bibinfo{pages}{124:1--124:20}.
\newblock
\urldef\tempurl%
\url{https://doi.org/10.4230/LIPICS.ICALP.2022.124}
\showDOI{\tempurl}


\bibitem[\protect\citeauthoryear{Ganty and Majumdar}{Ganty and
  Majumdar}{2012}]%
        {DBLP:journals/toplas/GantyM12}
\bibfield{author}{\bibinfo{person}{Pierre Ganty} {and} \bibinfo{person}{Rupak
  Majumdar}.} \bibinfo{year}{2012}\natexlab{}.
\newblock \showarticletitle{Algorithmic verification of asynchronous programs}.
\newblock \bibinfo{journal}{\emph{{ACM} Trans. Program. Lang. Syst.}}
  \bibinfo{volume}{34}, \bibinfo{number}{1} (\bibinfo{year}{2012}),
  \bibinfo{pages}{6:1--6:48}.
\newblock
\urldef\tempurl%
\url{https://doi.org/10.1145/2160910.2160915}
\showDOI{\tempurl}


\bibitem[\protect\citeauthoryear{Goharshady and Zaher}{Goharshady and
  Zaher}{2023}]%
        {DBLP:conf/vmcai/GoharshadyZ23}
\bibfield{author}{\bibinfo{person}{Amir~Kafshdar Goharshady} {and}
  \bibinfo{person}{Ahmed~Khaled Zaher}.} \bibinfo{year}{2023}\natexlab{}.
\newblock \showarticletitle{Efficient Interprocedural Data-Flow Analysis Using
  Treedepth and Treewidth}. In \bibinfo{booktitle}{\emph{Verification, Model
  Checking, and Abstract Interpretation - 24th International Conference,
  {VMCAI} 2023, Boston, MA, USA, January 16-17, 2023, Proceedings}}
  \emph{(\bibinfo{series}{Lecture Notes in Computer Science},
  Vol.~\bibinfo{volume}{13881})}, \bibfield{editor}{\bibinfo{person}{Cezara
  Dragoi}, \bibinfo{person}{Michael Emmi}, {and} \bibinfo{person}{Jingbo Wang}}
  (Eds.). \bibinfo{publisher}{Springer}, \bibinfo{pages}{177--202}.
\newblock
\urldef\tempurl%
\url{https://doi.org/10.1007/978-3-031-24950-1\_9}
\showDOI{\tempurl}


\bibitem[\protect\citeauthoryear{Habermehl, Meyer, and Wimmel}{Habermehl
  et~al\mbox{.}}{2010}]%
        {DBLP:conf/icalp/HabermehlMW10}
\bibfield{author}{\bibinfo{person}{Peter Habermehl}, \bibinfo{person}{Roland
  Meyer}, {and} \bibinfo{person}{Harro Wimmel}.}
  \bibinfo{year}{2010}\natexlab{}.
\newblock \showarticletitle{The Downward-Closure of Petri Net Languages}. In
  \bibinfo{booktitle}{\emph{Automata, Languages and Programming, 37th
  International Colloquium, {ICALP} 2010, Bordeaux, France, July 6-10, 2010,
  Proceedings, Part {II}}} \emph{(\bibinfo{series}{Lecture Notes in Computer
  Science}, Vol.~\bibinfo{volume}{6199})},
  \bibfield{editor}{\bibinfo{person}{Samson Abramsky}, \bibinfo{person}{Cyril
  Gavoille}, \bibinfo{person}{Claude Kirchner},
  \bibinfo{person}{Friedhelm~Meyer auf~der Heide}, {and}
  \bibinfo{person}{Paul~G. Spirakis}} (Eds.). \bibinfo{publisher}{Springer},
  \bibinfo{pages}{466--477}.
\newblock
\urldef\tempurl%
\url{https://doi.org/10.1007/978-3-642-14162-1\_39}
\showDOI{\tempurl}


\bibitem[\protect\citeauthoryear{Hague}{Hague}{2011}]%
        {DBLP:conf/fsttcs/Hague11}
\bibfield{author}{\bibinfo{person}{Matthew Hague}.}
  \bibinfo{year}{2011}\natexlab{}.
\newblock \showarticletitle{Parameterised Pushdown Systems with Non-Atomic
  Writes}. In \bibinfo{booktitle}{\emph{{IARCS} Annual Conference on
  Foundations of Software Technology and Theoretical Computer Science, {FSTTCS}
  2011, December 12-14, 2011, Mumbai, India}} \emph{(\bibinfo{series}{LIPIcs},
  Vol.~\bibinfo{volume}{13})}, \bibfield{editor}{\bibinfo{person}{Supratik
  Chakraborty} {and} \bibinfo{person}{Amit Kumar}} (Eds.).
  \bibinfo{publisher}{Schloss Dagstuhl - Leibniz-Zentrum f{\"{u}}r Informatik},
  \bibinfo{pages}{457--468}.
\newblock
\urldef\tempurl%
\url{https://doi.org/10.4230/LIPICS.FSTTCS.2011.457}
\showDOI{\tempurl}


\bibitem[\protect\citeauthoryear{Harkema}{Harkema}{2001}]%
        {DBLP:conf/lacl/Harkema01}
\bibfield{author}{\bibinfo{person}{Henk Harkema}.}
  \bibinfo{year}{2001}\natexlab{}.
\newblock \showarticletitle{A Characterization of Minimalist Languages}. In
  \bibinfo{booktitle}{\emph{Logical Aspects of Computational Linguistics, 4th
  International Conference, {LACL} 2001, Le Croisic, France, June 27-29, 2001,
  Proceedings}} \emph{(\bibinfo{series}{Lecture Notes in Computer Science},
  Vol.~\bibinfo{volume}{2099})}, \bibfield{editor}{\bibinfo{person}{Philippe
  de~Groote}, \bibinfo{person}{Glyn Morrill}, {and} \bibinfo{person}{Christian
  Retor{\'{e}}}} (Eds.). \bibinfo{publisher}{Springer},
  \bibinfo{pages}{193--211}.
\newblock
\urldef\tempurl%
\url{https://doi.org/10.1007/3-540-48199-0\_12}
\showDOI{\tempurl}


\bibitem[\protect\citeauthoryear{Higman}{Higman}{1952}]%
        {Higman1952OrderingBD}
\bibfield{author}{\bibinfo{person}{Graham Higman}.}
  \bibinfo{year}{1952}\natexlab{}.
\newblock \showarticletitle{Ordering by Divisibility in Abstract Algebras}.
\newblock \bibinfo{journal}{\emph{Proceedings of The London Mathematical
  Society}} (\bibinfo{year}{1952}), \bibinfo{pages}{326--336}.
\newblock
\urldef\tempurl%
\url{https://doi.org/10.1112/plms/s3-2.1.326}
\showDOI{\tempurl}


\bibitem[\protect\citeauthoryear{Huang, Dong, Milanova, and Dolby}{Huang
  et~al\mbox{.}}{2015}]%
        {DBLP:conf/issta/HuangDMD15}
\bibfield{author}{\bibinfo{person}{Wei Huang}, \bibinfo{person}{Yao Dong},
  \bibinfo{person}{Ana~L. Milanova}, {and} \bibinfo{person}{Julian Dolby}.}
  \bibinfo{year}{2015}\natexlab{}.
\newblock \showarticletitle{Scalable and precise taint analysis for Android}.
  In \bibinfo{booktitle}{\emph{Proceedings of the 2015 International Symposium
  on Software Testing and Analysis, {ISSTA} 2015, Baltimore, MD, USA, July
  12-17, 2015}}, \bibfield{editor}{\bibinfo{person}{Michal Young} {and}
  \bibinfo{person}{Tao Xie}} (Eds.). \bibinfo{publisher}{{ACM}},
  \bibinfo{pages}{106--117}.
\newblock
\urldef\tempurl%
\url{https://doi.org/10.1145/2771783.2771803}
\showDOI{\tempurl}


\bibitem[\protect\citeauthoryear{Inverso, Tomasco, Fischer, Torre, and
  Parlato}{Inverso et~al\mbox{.}}{2022}]%
        {InversoTFTP22}
\bibfield{author}{\bibinfo{person}{Omar Inverso}, \bibinfo{person}{Ermenegildo
  Tomasco}, \bibinfo{person}{Bernd Fischer}, \bibinfo{person}{Salvatore~La
  Torre}, {and} \bibinfo{person}{Gennaro Parlato}.}
  \bibinfo{year}{2022}\natexlab{}.
\newblock \showarticletitle{Bounded Verification of Multi-threaded Programs via
  Lazy Sequentialization}.
\newblock \bibinfo{journal}{\emph{{ACM} Trans. Program. Lang. Syst.}}
  \bibinfo{volume}{44}, \bibinfo{number}{1} (\bibinfo{year}{2022}),
  \bibinfo{pages}{1:1--1:50}.
\newblock
\urldef\tempurl%
\url{https://doi.org/10.1145/3478536}
\showDOI{\tempurl}


\bibitem[\protect\citeauthoryear{Jecker}{Jecker}{2021}]%
        {DBLP:conf/stacs/Jecker21}
\bibfield{author}{\bibinfo{person}{Isma{\"{e}}l Jecker}.}
  \bibinfo{year}{2021}\natexlab{}.
\newblock \showarticletitle{A Ramsey Theorem for Finite Monoids}. In
  \bibinfo{booktitle}{\emph{38th International Symposium on Theoretical Aspects
  of Computer Science, {STACS} 2021, March 16-19, 2021, Saarbr{\"{u}}cken,
  Germany (Virtual Conference)}} \emph{(\bibinfo{series}{LIPIcs},
  Vol.~\bibinfo{volume}{187})}, \bibfield{editor}{\bibinfo{person}{Markus
  Bl{\"{a}}ser} {and} \bibinfo{person}{Benjamin Monmege}} (Eds.).
  \bibinfo{publisher}{Schloss Dagstuhl - Leibniz-Zentrum f{\"{u}}r Informatik},
  \bibinfo{pages}{44:1--44:13}.
\newblock
\urldef\tempurl%
\url{https://doi.org/10.4230/LIPICS.STACS.2021.44}
\showDOI{\tempurl}


\bibitem[\protect\citeauthoryear{Joshi}{Joshi}{1985}]%
        {Joshi1985}
\bibfield{author}{\bibinfo{person}{Aravind~Krishna Joshi}.}
  \bibinfo{year}{1985}\natexlab{}.
\newblock \showarticletitle{Tree adjoining grammars: How much
  context-sensitivity is required to provide reasonable structural
  discriptions}.
\newblock In \bibinfo{booktitle}{\emph{Natural Language Parsind: Psychological,
  Computational, and Theoretical Perspectives}},
  \bibfield{editor}{\bibinfo{person}{D.~Dowty}, \bibinfo{person}{L.~Karttunen},
  {and} \bibinfo{person}{A.~Zwicky}} (Eds.). \bibinfo{publisher}{Cambridge
  University Press}, \bibinfo{pages}{206--250}.
\newblock
\urldef\tempurl%
\url{https://doi.org/10.1017/CBO9780511597855.007}
\showDOI{\tempurl}


\bibitem[\protect\citeauthoryear{Kaji, Nakanishi, Seki, and Kasami}{Kaji
  et~al\mbox{.}}{1992a}]%
        {kaji1992b}
\bibfield{author}{\bibinfo{person}{Yuichi Kaji}, \bibinfo{person}{Ryuichi
  Nakanishi}, \bibinfo{person}{Hiroyuki Seki}, {and} \bibinfo{person}{Tadao
  Kasami}.} \bibinfo{year}{1992}\natexlab{a}.
\newblock \showarticletitle{The universal recognition problems for parallel
  multiple context-free grammars and for their subclasses}.
\newblock \bibinfo{journal}{\emph{IEICE TRANSACTIONS on Information and
  Systems}} \bibinfo{volume}{75}, \bibinfo{number}{4} (\bibinfo{year}{1992}),
  \bibinfo{pages}{499--508}.
\newblock


\bibitem[\protect\citeauthoryear{Kaji, Nakanishi, Seki, and Kasami}{Kaji
  et~al\mbox{.}}{1994}]%
        {DBLP:journals/ci/KajiNSK94}
\bibfield{author}{\bibinfo{person}{Yuichi Kaji}, \bibinfo{person}{Ryuchi
  Nakanishi}, \bibinfo{person}{Hiroyuki Seki}, {and} \bibinfo{person}{Tadao
  Kasami}.} \bibinfo{year}{1994}\natexlab{}.
\newblock \showarticletitle{The Computational Complexity of the Universal
  Recognition Problem for Parallel Multiple Context-Free Grammars}.
\newblock \bibinfo{journal}{\emph{Comput. Intell.}}  \bibinfo{volume}{10}
  (\bibinfo{year}{1994}), \bibinfo{pages}{440--452}.
\newblock
\urldef\tempurl%
\url{https://doi.org/10.1111/J.1467-8640.1994.TB00008.X}
\showDOI{\tempurl}


\bibitem[\protect\citeauthoryear{Kaji, Nakanisi, Seki, and Kasami}{Kaji
  et~al\mbox{.}}{1992b}]%
        {kaji1992a}
\bibfield{author}{\bibinfo{person}{Yuichi Kaji}, \bibinfo{person}{Ryuichi
  Nakanisi}, \bibinfo{person}{Hiroyuki Seki}, {and} \bibinfo{person}{Tadao
  Kasami}.} \bibinfo{year}{1992}\natexlab{b}.
\newblock \showarticletitle{The universal recognition problems for multiple
  context-free grammars and for linear context-free rewriting systems}.
\newblock \bibinfo{journal}{\emph{IEICE Transactions on Information and
  Systems}} \bibinfo{volume}{75}, \bibinfo{number}{1} (\bibinfo{year}{1992}),
  \bibinfo{pages}{78--88}.
\newblock


\bibitem[\protect\citeauthoryear{Kanazawa and Salvati}{Kanazawa and
  Salvati}{2010}]%
        {DBLP:conf/lata/KanazawaS10}
\bibfield{author}{\bibinfo{person}{Makoto Kanazawa} {and}
  \bibinfo{person}{Sylvain Salvati}.} \bibinfo{year}{2010}\natexlab{}.
\newblock \showarticletitle{The Copying Power of Well-Nested Multiple
  Context-Free Grammars}. In \bibinfo{booktitle}{\emph{Language and Automata
  Theory and Applications, 4th International Conference, {LATA} 2010, Trier,
  Germany, May 24-28, 2010. Proceedings}} \emph{(\bibinfo{series}{Lecture Notes
  in Computer Science}, Vol.~\bibinfo{volume}{6031})},
  \bibfield{editor}{\bibinfo{person}{Adrian{-}Horia Dediu},
  \bibinfo{person}{Henning Fernau}, {and} \bibinfo{person}{Carlos
  Mart{\'{\i}}n{-}Vide}} (Eds.). \bibinfo{publisher}{Springer},
  \bibinfo{pages}{344--355}.
\newblock
\urldef\tempurl%
\url{https://doi.org/10.1007/978-3-642-13089-2\_29}
\showDOI{\tempurl}


\bibitem[\protect\citeauthoryear{Kjelstr{\o}m and Pavlogiannis}{Kjelstr{\o}m
  and Pavlogiannis}{2022}]%
        {DBLP:journals/pacmpl/KjelstromP22}
\bibfield{author}{\bibinfo{person}{Adam~Husted Kjelstr{\o}m} {and}
  \bibinfo{person}{Andreas Pavlogiannis}.} \bibinfo{year}{2022}\natexlab{}.
\newblock \showarticletitle{The decidability and complexity of interleaved
  bidirected Dyck reachability}.
\newblock \bibinfo{journal}{\emph{Proc. {ACM} Program. Lang.}}
  \bibinfo{volume}{6}, \bibinfo{number}{{POPL}} (\bibinfo{year}{2022}),
  \bibinfo{pages}{1--26}.
\newblock
\urldef\tempurl%
\url{https://doi.org/10.1145/3498673}
\showDOI{\tempurl}


\bibitem[\protect\citeauthoryear{Krishna, Lal, Pavlogiannis, and Tuppe}{Krishna
  et~al\mbox{.}}{2024}]%
        {DBLP:journals/pacmpl/KrishnaLPT24}
\bibfield{author}{\bibinfo{person}{Shankaranarayanan Krishna},
  \bibinfo{person}{Aniket Lal}, \bibinfo{person}{Andreas Pavlogiannis}, {and}
  \bibinfo{person}{Omkar Tuppe}.} \bibinfo{year}{2024}\natexlab{}.
\newblock \showarticletitle{On-the-Fly Static Analysis via Dynamic Bidirected
  Dyck Reachability}.
\newblock \bibinfo{journal}{\emph{Proc. {ACM} Program. Lang.}}
  \bibinfo{volume}{8}, \bibinfo{number}{{POPL}} (\bibinfo{year}{2024}),
  \bibinfo{pages}{1239--1268}.
\newblock
\urldef\tempurl%
\url{https://doi.org/10.1145/3632884}
\showDOI{\tempurl}


\bibitem[\protect\citeauthoryear{{La Torre}, Madhusudan, and Parlato}{{La
  Torre} et~al\mbox{.}}{2007}]%
        {TorreMP07}
\bibfield{author}{\bibinfo{person}{Salvatore {La Torre}},
  \bibinfo{person}{Parthasarathy Madhusudan}, {and} \bibinfo{person}{Gennaro
  Parlato}.} \bibinfo{year}{2007}\natexlab{}.
\newblock \showarticletitle{A Robust Class of Context-Sensitive Languages}. In
  \bibinfo{booktitle}{\emph{22nd {IEEE} Symposium on Logic in Computer Science
  {(LICS} 2007), 10-12 July 2007, Wroclaw, Poland, Proceedings}}.
  \bibinfo{publisher}{{IEEE} Computer Society}, \bibinfo{pages}{161--170}.
\newblock
\urldef\tempurl%
\url{https://doi.org/10.1109/LICS.2007.9}
\showDOI{\tempurl}


\bibitem[\protect\citeauthoryear{{La Torre}, Muscholl, and Walukiewicz}{{La
  Torre} et~al\mbox{.}}{2015}]%
        {DBLP:conf/concur/TorreMW15}
\bibfield{author}{\bibinfo{person}{Salvatore {La Torre}}, \bibinfo{person}{Anca
  Muscholl}, {and} \bibinfo{person}{Igor Walukiewicz}.}
  \bibinfo{year}{2015}\natexlab{}.
\newblock \showarticletitle{Safety of Parametrized Asynchronous Shared-Memory
  Systems is Almost Always Decidable}. In \bibinfo{booktitle}{\emph{26th
  International Conference on Concurrency Theory, {CONCUR} 2015, Madrid, Spain,
  September 1.4, 2015}} \emph{(\bibinfo{series}{LIPIcs},
  Vol.~\bibinfo{volume}{42})}, \bibfield{editor}{\bibinfo{person}{Luca Aceto}
  {and} \bibinfo{person}{David de~Frutos{-}Escrig}} (Eds.).
  \bibinfo{publisher}{Schloss Dagstuhl - Leibniz-Zentrum f{\"{u}}r Informatik},
  \bibinfo{pages}{72--84}.
\newblock
\urldef\tempurl%
\url{https://doi.org/10.4230/LIPICS.CONCUR.2015.72}
\showDOI{\tempurl}


\bibitem[\protect\citeauthoryear{La~Torre, Napoli, and Parlato}{La~Torre
  et~al\mbox{.}}{2014}]%
        {LaTorreSNM14}
\bibfield{author}{\bibinfo{person}{Salvatore La~Torre},
  \bibinfo{person}{Margherita Napoli}, {and} \bibinfo{person}{Gennaro
  Parlato}.} \bibinfo{year}{2014}\natexlab{}.
\newblock \showarticletitle{A Unifying Approach for Multistack Pushdown
  Automata}. In \bibinfo{booktitle}{\emph{Proceedings of MFCS 2014}},
  \bibfield{editor}{\bibinfo{person}{Erzs{\'e}bet Csuhaj-Varj{\'u}},
  \bibinfo{person}{Martin Dietzfelbinger}, {and} \bibinfo{person}{Zolt{\'a}n
  {\'E}sik}} (Eds.). \bibinfo{publisher}{Springer Berlin Heidelberg},
  \bibinfo{address}{Berlin, Heidelberg}, \bibinfo{pages}{377--389}.
\newblock
\urldef\tempurl%
\url{https://doi.org/10.1007/978-3-662-44522-8_32}
\showDOI{\tempurl}


\bibitem[\protect\citeauthoryear{{La Torre}, Napoli, and Parlato}{{La Torre}
  et~al\mbox{.}}{2020}]%
        {TorreNP20}
\bibfield{author}{\bibinfo{person}{Salvatore {La Torre}},
  \bibinfo{person}{Margherita Napoli}, {and} \bibinfo{person}{Gennaro
  Parlato}.} \bibinfo{year}{2020}\natexlab{}.
\newblock \showarticletitle{Reachability of scope-bounded multistack pushdown
  systems}.
\newblock \bibinfo{journal}{\emph{Inf. Comput.}}  \bibinfo{volume}{275}
  (\bibinfo{year}{2020}), \bibinfo{pages}{104588}.
\newblock
\urldef\tempurl%
\url{https://doi.org/10.1016/J.IC.2020.104588}
\showDOI{\tempurl}


\bibitem[\protect\citeauthoryear{Lautemann, Schwentick, and
  Th{\'{e}}rien}{Lautemann et~al\mbox{.}}{1994}]%
        {LautemannST94}
\bibfield{author}{\bibinfo{person}{Clemens Lautemann}, \bibinfo{person}{Thomas
  Schwentick}, {and} \bibinfo{person}{Denis Th{\'{e}}rien}.}
  \bibinfo{year}{1994}\natexlab{}.
\newblock \showarticletitle{Logics For Context-Free Languages}. In
  \bibinfo{booktitle}{\emph{Computer Science Logic, 8th International Workshop,
  {CSL} '94, Kazimierz, Poland, September 25-30, 1994, Selected Papers}}
  \emph{(\bibinfo{series}{Lecture Notes in Computer Science},
  Vol.~\bibinfo{volume}{933})}, \bibfield{editor}{\bibinfo{person}{Leszek
  Pacholski} {and} \bibinfo{person}{Jerzy Tiuryn}} (Eds.).
  \bibinfo{publisher}{Springer}, \bibinfo{pages}{205--216}.
\newblock
\urldef\tempurl%
\url{https://doi.org/10.1007/BFB0022257}
\showDOI{\tempurl}


\bibitem[\protect\citeauthoryear{Li, Satya, and Zhang}{Li
  et~al\mbox{.}}{2022}]%
        {DBLP:journals/pacmpl/LiSZ22}
\bibfield{author}{\bibinfo{person}{Yuanbo Li}, \bibinfo{person}{Kris Satya},
  {and} \bibinfo{person}{Qirun Zhang}.} \bibinfo{year}{2022}\natexlab{}.
\newblock \showarticletitle{Efficient algorithms for dynamic bidirected
  Dyck-reachability}.
\newblock \bibinfo{journal}{\emph{Proc. {ACM} Program. Lang.}}
  \bibinfo{volume}{6}, \bibinfo{number}{{POPL}} (\bibinfo{year}{2022}),
  \bibinfo{pages}{1--29}.
\newblock
\urldef\tempurl%
\url{https://doi.org/10.1145/3498724}
\showDOI{\tempurl}


\bibitem[\protect\citeauthoryear{Li, Zhang, and Reps}{Li et~al\mbox{.}}{2021}]%
        {DBLP:journals/pacmpl/LiZR21}
\bibfield{author}{\bibinfo{person}{Yuanbo Li}, \bibinfo{person}{Qirun Zhang},
  {and} \bibinfo{person}{Thomas~W. Reps}.} \bibinfo{year}{2021}\natexlab{}.
\newblock \showarticletitle{On the complexity of bidirected interleaved
  Dyck-reachability}.
\newblock \bibinfo{journal}{\emph{Proc. {ACM} Program. Lang.}}
  \bibinfo{volume}{5}, \bibinfo{number}{{POPL}} (\bibinfo{year}{2021}),
  \bibinfo{pages}{1--28}.
\newblock
\urldef\tempurl%
\url{https://doi.org/10.1145/3434340}
\showDOI{\tempurl}


\bibitem[\protect\citeauthoryear{Li, Zhang, and Reps}{Li et~al\mbox{.}}{2023}]%
        {DBLP:journals/pacmpl/LiZR23}
\bibfield{author}{\bibinfo{person}{Yuanbo Li}, \bibinfo{person}{Qirun Zhang},
  {and} \bibinfo{person}{Thomas~W. Reps}.} \bibinfo{year}{2023}\natexlab{}.
\newblock \showarticletitle{Single-Source-Single-Target Interleaved-Dyck
  Reachability via Integer Linear Programming}.
\newblock \bibinfo{journal}{\emph{Proc. {ACM} Program. Lang.}}
  \bibinfo{volume}{7}, \bibinfo{number}{{POPL}} (\bibinfo{year}{2023}),
  \bibinfo{pages}{1003--1026}.
\newblock
\urldef\tempurl%
\url{https://doi.org/10.1145/3571228}
\showDOI{\tempurl}


\bibitem[\protect\citeauthoryear{Madhusudan and Parlato}{Madhusudan and
  Parlato}{2011}]%
        {TreeAuxMadhu}
\bibfield{author}{\bibinfo{person}{P. Madhusudan} {and}
  \bibinfo{person}{Gennaro Parlato}.} \bibinfo{year}{2011}\natexlab{}.
\newblock \showarticletitle{The tree width of auxiliary storage}.
\newblock \bibinfo{journal}{\emph{SIGPLAN Not.}} \bibinfo{volume}{46},
  \bibinfo{number}{1} (\bibinfo{year}{2011}), \bibinfo{pages}{283–294}.
\newblock
\showISSN{0362-1340}
\urldef\tempurl%
\url{https://doi.org/10.1145/1925844.1926419}
\showDOI{\tempurl}


\bibitem[\protect\citeauthoryear{Majumdar, Thinniyam, and Zetzsche}{Majumdar
  et~al\mbox{.}}{2022}]%
        {DBLP:journals/lmcs/MajumdarTZ22}
\bibfield{author}{\bibinfo{person}{Rupak Majumdar},
  \bibinfo{person}{Ramanathan~S. Thinniyam}, {and} \bibinfo{person}{Georg
  Zetzsche}.} \bibinfo{year}{2022}\natexlab{}.
\newblock \showarticletitle{General Decidability Results for Asynchronous
  Shared-Memory Programs: Higher-Order and Beyond}.
\newblock \bibinfo{journal}{\emph{Log. Methods Comput. Sci.}}
  \bibinfo{volume}{18}, \bibinfo{number}{4} (\bibinfo{year}{2022}).
\newblock
\urldef\tempurl%
\url{https://doi.org/10.46298/LMCS-18(4:2)2022}
\showDOI{\tempurl}


\bibitem[\protect\citeauthoryear{Michaelis}{Michaelis}{2001}]%
        {DBLP:conf/lacl/Michaelis01}
\bibfield{author}{\bibinfo{person}{Jens Michaelis}.}
  \bibinfo{year}{2001}\natexlab{}.
\newblock \showarticletitle{Transforming Linear Context-Free Rewriting Systems
  into Minimalist Grammars}. In \bibinfo{booktitle}{\emph{Logical Aspects of
  Computational Linguistics, 4th International Conference, {LACL} 2001, Le
  Croisic, France, June 27-29, 2001, Proceedings}}
  \emph{(\bibinfo{series}{Lecture Notes in Computer Science},
  Vol.~\bibinfo{volume}{2099})}, \bibfield{editor}{\bibinfo{person}{Philippe
  de~Groote}, \bibinfo{person}{Glyn Morrill}, {and} \bibinfo{person}{Christian
  Retor{\'{e}}}} (Eds.). \bibinfo{publisher}{Springer},
  \bibinfo{pages}{228--244}.
\newblock
\urldef\tempurl%
\url{https://doi.org/10.1007/3-540-48199-0\_14}
\showDOI{\tempurl}


\bibitem[\protect\citeauthoryear{Musuvathi and Qadeer}{Musuvathi and
  Qadeer}{2007}]%
        {MusuvathiQadeer}
\bibfield{author}{\bibinfo{person}{Madanlal Musuvathi} {and}
  \bibinfo{person}{Shaz Qadeer}.} \bibinfo{year}{2007}\natexlab{}.
\newblock \showarticletitle{Iterative context bounding for systematic testing
  of multithreaded programs}. In \bibinfo{booktitle}{\emph{Proceedings of the
  {ACM} {SIGPLAN} 2007 Conference on Programming Language Design and
  Implementation, {PLDI} 2007, San Diego, CA, USA, June 10-13, 2007}}.
  \bibinfo{publisher}{{ACM}}, \bibinfo{pages}{446--455}.
\newblock
\urldef\tempurl%
\url{https://doi.org/10.1145/1250734.1250785}
\showDOI{\tempurl}


\bibitem[\protect\citeauthoryear{Parchmann, Duske, and Specht}{Parchmann
  et~al\mbox{.}}{1980}]%
        {parchmann1980deterministic}
\bibfield{author}{\bibinfo{person}{Rainer Parchmann},
  \bibinfo{person}{J{\"u}rgen Duske}, {and} \bibinfo{person}{Johann Specht}.}
  \bibinfo{year}{1980}\natexlab{}.
\newblock \showarticletitle{On deterministic indexed languages}.
\newblock \bibinfo{journal}{\emph{Inf. Control.}} \bibinfo{volume}{45},
  \bibinfo{number}{1} (\bibinfo{year}{1980}), \bibinfo{pages}{48--67}.
\newblock
\urldef\tempurl%
\url{https://doi.org/10.1016/S0019-9958(80)90867-0}
\showDOI{\tempurl}


\bibitem[\protect\citeauthoryear{Pavlogiannis}{Pavlogiannis}{2022}]%
        {DBLP:journals/siglog/Pavlogiannis22}
\bibfield{author}{\bibinfo{person}{Andreas Pavlogiannis}.}
  \bibinfo{year}{2022}\natexlab{}.
\newblock \showarticletitle{CFL/Dyck Reachability: An Algorithmic Perspective}.
\newblock \bibinfo{journal}{\emph{{ACM} {SIGLOG} News}} \bibinfo{volume}{9},
  \bibinfo{number}{4} (\bibinfo{year}{2022}), \bibinfo{pages}{5--25}.
\newblock
\urldef\tempurl%
\url{https://doi.org/10.1145/3583660.3583664}
\showDOI{\tempurl}


\bibitem[\protect\citeauthoryear{Qadeer and Rehof}{Qadeer and Rehof}{2005}]%
        {QadeerR05}
\bibfield{author}{\bibinfo{person}{Shaz Qadeer} {and} \bibinfo{person}{Jakob
  Rehof}.} \bibinfo{year}{2005}\natexlab{}.
\newblock \showarticletitle{Context-Bounded Model Checking of Concurrent
  Software}. In \bibinfo{booktitle}{\emph{Tools and Algorithms for the
  Construction and Analysis of Systems, 11th International Conference, {TACAS}
  2005, Held as Part of the Joint European Conferences on Theory and Practice
  of Software, {ETAPS} 2005, Edinburgh, UK, April 4-8, 2005, Proceedings}}
  \emph{(\bibinfo{series}{Lecture Notes in Computer Science},
  Vol.~\bibinfo{volume}{3440})}, \bibfield{editor}{\bibinfo{person}{Nicolas
  Halbwachs} {and} \bibinfo{person}{Lenore~D. Zuck}} (Eds.).
  \bibinfo{publisher}{Springer}, \bibinfo{pages}{93--107}.
\newblock
\urldef\tempurl%
\url{https://doi.org/10.1007/978-3-540-31980-1\_7}
\showDOI{\tempurl}


\bibitem[\protect\citeauthoryear{Rabin}{Rabin}{1969}]%
        {RabinTree}
\bibfield{author}{\bibinfo{person}{Michael~O. Rabin}.}
  \bibinfo{year}{1969}\natexlab{}.
\newblock \showarticletitle{Decidability of Second-Order Theories and Automata
  on Infinite Trees}.
\newblock \bibinfo{journal}{\emph{Trans. Amer. Math. Soc.}}
  \bibinfo{volume}{141} (\bibinfo{year}{1969}), \bibinfo{pages}{1--35}.
\newblock
\showISSN{00029947}
\urldef\tempurl%
\url{https://doi.org/10.1090/S0002-9947-1969-0246760-1}
\showDOI{\tempurl}


\bibitem[\protect\citeauthoryear{Ramalingam}{Ramalingam}{2000}]%
        {ramalingam2000context}
\bibfield{author}{\bibinfo{person}{Ganesan Ramalingam}.}
  \bibinfo{year}{2000}\natexlab{}.
\newblock \showarticletitle{Context-sensitive synchronization-sensitive
  analysis is undecidable}.
\newblock \bibinfo{journal}{\emph{ACM Transactions on Programming languages and
  Systems (TOPLAS)}} \bibinfo{volume}{22}, \bibinfo{number}{2}
  (\bibinfo{year}{2000}), \bibinfo{pages}{416--430}.
\newblock
\urldef\tempurl%
\url{https://doi.org/10.1145/349214.349241}
\showDOI{\tempurl}


\bibitem[\protect\citeauthoryear{Ramsey}{Ramsey}{1930}]%
        {Ramsey:1930:ramsey-theorem}
\bibfield{author}{\bibinfo{person}{F.~P. Ramsey}.}
  \bibinfo{year}{1930}\natexlab{}.
\newblock \showarticletitle{On a Problem of Formal Logic}.
\newblock \bibinfo{journal}{\emph{Proceedings of the London Mathematical
  Society}} \bibinfo{volume}{s2-30}, \bibinfo{number}{1}
  (\bibinfo{year}{1930}), \bibinfo{pages}{264--286}.
\newblock
\urldef\tempurl%
\url{https://doi.org/10.1112/plms/s2-30.1.264}
\showDOI{\tempurl}


\bibitem[\protect\citeauthoryear{Reps}{Reps}{2000}]%
        {DBLP:journals/toplas/Reps00}
\bibfield{author}{\bibinfo{person}{Thomas~W. Reps}.}
  \bibinfo{year}{2000}\natexlab{}.
\newblock \showarticletitle{Undecidability of context-sensitive
  data-independence analysis}.
\newblock \bibinfo{journal}{\emph{{ACM} Trans. Program. Lang. Syst.}}
  \bibinfo{volume}{22}, \bibinfo{number}{1} (\bibinfo{year}{2000}),
  \bibinfo{pages}{162--186}.
\newblock
\urldef\tempurl%
\url{https://doi.org/10.1145/345099.345137}
\showDOI{\tempurl}


\bibitem[\protect\citeauthoryear{Seki, Matsumura, Fujii, and Kasami}{Seki
  et~al\mbox{.}}{1991}]%
        {SekiMFK91}
\bibfield{author}{\bibinfo{person}{Hiroyuki Seki}, \bibinfo{person}{Takashi
  Matsumura}, \bibinfo{person}{Mamoru Fujii}, {and} \bibinfo{person}{Tadao
  Kasami}.} \bibinfo{year}{1991}\natexlab{}.
\newblock \showarticletitle{On Multiple Context-Free Grammars}.
\newblock \bibinfo{journal}{\emph{Theor. Comput. Sci.}} \bibinfo{volume}{88},
  \bibinfo{number}{2} (\bibinfo{year}{1991}), \bibinfo{pages}{191--229}.
\newblock
\urldef\tempurl%
\url{https://doi.org/10.1016/0304-3975(91)90374-B}
\showDOI{\tempurl}


\bibitem[\protect\citeauthoryear{Sp{\"{a}}th, Ali, and Bodden}{Sp{\"{a}}th
  et~al\mbox{.}}{2019}]%
        {DBLP:journals/pacmpl/SpathAB19}
\bibfield{author}{\bibinfo{person}{Johannes Sp{\"{a}}th},
  \bibinfo{person}{Karim Ali}, {and} \bibinfo{person}{Eric Bodden}.}
  \bibinfo{year}{2019}\natexlab{}.
\newblock \showarticletitle{Context-, flow-, and field-sensitive data-flow
  analysis using synchronized Pushdown systems}.
\newblock \bibinfo{journal}{\emph{Proc. {ACM} Program. Lang.}}
  \bibinfo{volume}{3}, \bibinfo{number}{{POPL}} (\bibinfo{year}{2019}),
  \bibinfo{pages}{48:1--48:29}.
\newblock
\urldef\tempurl%
\url{https://doi.org/10.1145/3290361}
\showDOI{\tempurl}


\bibitem[\protect\citeauthoryear{Sridharan and Bod{\'{\i}}k}{Sridharan and
  Bod{\'{\i}}k}{2006}]%
        {DBLP:conf/pldi/SridharanB06}
\bibfield{author}{\bibinfo{person}{Manu Sridharan} {and}
  \bibinfo{person}{Rastislav Bod{\'{\i}}k}.} \bibinfo{year}{2006}\natexlab{}.
\newblock \showarticletitle{Refinement-based context-sensitive points-to
  analysis for Java}. In \bibinfo{booktitle}{\emph{Proceedings of the {ACM}
  {SIGPLAN} 2006 Conference on Programming Language Design and Implementation,
  Ottawa, Ontario, Canada, June 11-14, 2006}},
  \bibfield{editor}{\bibinfo{person}{Michael~I. Schwartzbach} {and}
  \bibinfo{person}{Thomas Ball}} (Eds.). \bibinfo{publisher}{{ACM}},
  \bibinfo{pages}{387--400}.
\newblock
\urldef\tempurl%
\url{https://doi.org/10.1145/1133981.1134027}
\showDOI{\tempurl}


\bibitem[\protect\citeauthoryear{Thorup}{Thorup}{1998}]%
        {DBLP:journals/iandc/Thorup98}
\bibfield{author}{\bibinfo{person}{Mikkel Thorup}.}
  \bibinfo{year}{1998}\natexlab{}.
\newblock \showarticletitle{All Structured Programs have Small Tree-Width and
  Good Register Allocation}.
\newblock \bibinfo{journal}{\emph{Inf. Comput.}} \bibinfo{volume}{142},
  \bibinfo{number}{2} (\bibinfo{year}{1998}), \bibinfo{pages}{159--181}.
\newblock
\urldef\tempurl%
\url{https://doi.org/10.1006/INCO.1997.2697}
\showDOI{\tempurl}


\bibitem[\protect\citeauthoryear{Vijay{-}Shanker, Weir, and
  Joshi}{Vijay{-}Shanker et~al\mbox{.}}{1987}]%
        {DBLP:conf/acl/Vijay-ShankerWJ87}
\bibfield{author}{\bibinfo{person}{K. Vijay{-}Shanker},
  \bibinfo{person}{David~J. Weir}, {and} \bibinfo{person}{Aravind~K. Joshi}.}
  \bibinfo{year}{1987}\natexlab{}.
\newblock \showarticletitle{Characterizing Structural Descriptions produced by
  Various Grammatical Formalisms}. In \bibinfo{booktitle}{\emph{25th Annual
  Meeting of the Association for Computational Linguistics, Stanford
  University, Stanford, California, USA, July 6-9, 1987}},
  \bibfield{editor}{\bibinfo{person}{Candy~L. Sidner}} (Ed.).
  \bibinfo{publisher}{{ACL}}, \bibinfo{pages}{104--111}.
\newblock
\urldef\tempurl%
\url{https://doi.org/10.3115/981175.981190}
\showDOI{\tempurl}


\bibitem[\protect\citeauthoryear{Weir}{Weir}{1992}]%
        {DBLP:conf/acl/Weir92}
\bibfield{author}{\bibinfo{person}{David~J. Weir}.}
  \bibinfo{year}{1992}\natexlab{}.
\newblock \showarticletitle{Linear Context-Free Rewriting Systems and
  Deterministic Tree-Walking Transducers}. In \bibinfo{booktitle}{\emph{30th
  Annual Meeting of the Association for Computational Linguistics, 28 June - 2
  July 1992, University of Delaware, Newark, Deleware, USA, Proceedings}},
  \bibfield{editor}{\bibinfo{person}{Henry~S. Thompson}} (Ed.).
  \bibinfo{publisher}{{ACL}}, \bibinfo{pages}{136--143}.
\newblock
\urldef\tempurl%
\url{https://doi.org/10.3115/981967.981985}
\showDOI{\tempurl}


\bibitem[\protect\citeauthoryear{Zetzsche}{Zetzsche}{2015}]%
        {DBLP:conf/icalp/Zetzsche15}
\bibfield{author}{\bibinfo{person}{Georg Zetzsche}.}
  \bibinfo{year}{2015}\natexlab{}.
\newblock \showarticletitle{An Approach to Computing Downward Closures}. In
  \bibinfo{booktitle}{\emph{Automata, Languages, and Programming - 42nd
  International Colloquium, {ICALP} 2015, Kyoto, Japan, July 6-10, 2015,
  Proceedings, Part {II}}} \emph{(\bibinfo{series}{Lecture Notes in Computer
  Science}, Vol.~\bibinfo{volume}{9135})},
  \bibfield{editor}{\bibinfo{person}{Magn{\'{u}}s~M. Halld{\'{o}}rsson},
  \bibinfo{person}{Kazuo Iwama}, \bibinfo{person}{Naoki Kobayashi}, {and}
  \bibinfo{person}{Bettina Speckmann}} (Eds.). \bibinfo{publisher}{Springer},
  \bibinfo{pages}{440--451}.
\newblock
\urldef\tempurl%
\url{https://doi.org/10.1007/978-3-662-47666-6\_35}
\showDOI{\tempurl}


\bibitem[\protect\citeauthoryear{Zetzsche}{Zetzsche}{2016}]%
        {DBLP:conf/icalp/Zetzsche16}
\bibfield{author}{\bibinfo{person}{Georg Zetzsche}.}
  \bibinfo{year}{2016}\natexlab{}.
\newblock \showarticletitle{The Complexity of Downward Closure Comparisons}. In
  \bibinfo{booktitle}{\emph{43rd International Colloquium on Automata,
  Languages, and Programming, {ICALP} 2016, July 11-15, 2016, Rome, Italy}}
  \emph{(\bibinfo{series}{LIPIcs}, Vol.~\bibinfo{volume}{55})},
  \bibfield{editor}{\bibinfo{person}{Ioannis Chatzigiannakis},
  \bibinfo{person}{Michael Mitzenmacher}, \bibinfo{person}{Yuval Rabani}, {and}
  \bibinfo{person}{Davide Sangiorgi}} (Eds.). \bibinfo{publisher}{Schloss
  Dagstuhl - Leibniz-Zentrum f{\"{u}}r Informatik},
  \bibinfo{pages}{123:1--123:14}.
\newblock
\urldef\tempurl%
\url{https://doi.org/10.4230/LIPICS.ICALP.2016.123}
\showDOI{\tempurl}


\bibitem[\protect\citeauthoryear{Zhang and Su}{Zhang and Su}{2017}]%
        {DBLP:conf/popl/ZhangS17}
\bibfield{author}{\bibinfo{person}{Qirun Zhang} {and} \bibinfo{person}{Zhendong
  Su}.} \bibinfo{year}{2017}\natexlab{}.
\newblock \showarticletitle{Context-sensitive data-dependence analysis via
  linear conjunctive language reachability}. In
  \bibinfo{booktitle}{\emph{Proceedings of the 44th {ACM} {SIGPLAN} Symposium
  on Principles of Programming Languages, {POPL} 2017, Paris, France, January
  18-20, 2017}}, \bibfield{editor}{\bibinfo{person}{Giuseppe Castagna} {and}
  \bibinfo{person}{Andrew~D. Gordon}} (Eds.). \bibinfo{publisher}{{ACM}},
  \bibinfo{pages}{344--358}.
\newblock
\urldef\tempurl%
\url{https://doi.org/10.1145/3009837.3009848}
\showDOI{\tempurl}


\end{thebibliography}
\label{endbibliography}
\newoutputstream{pageendbibliography}
\openoutputfile{main.pageendbibliography.ctr}{pageendbibliography}
\addtostream{pageendbibliography}{\getpagerefnumber{endbibliography}}
\closeoutputstream{pageendbibliography}

\onlyFull{%

\newpage
\label{startappendix}
\newoutputstream{pagestartappendix}
\openoutputfile{main.pagestartappendix.ctr}{pagestartappendix}
\addtostream{pagestartappendix}{\getpagerefnumber{startappendix}}
\closeoutputstream{pagestartappendix}

\appendix

\crefalias{section}{appsec}
\crefalias{subsection}{appsec}

\section{Special treewidth versus Treewidth}\label[appsec]{app:stw}
For readers familiar with the bag decomposition based definition of treewidth: Tree decompositions for treewidth require the bags containing a vertex to be a connected subtree. In the case of special treewidth, the bags containing a vertex need to form a connected path. For those familiar with treewidth algebra: The join operation can have vertices with the same color in both operands, and the operation ``fuses'' the respective vertices in the case of treewidth. In the case of special treewidth, join is allowed only if the active colors of the operands are disjoint, and hence it avoids the need for ``fusion''.

For bounded degree graphs with degree bound $d$, the special treewidth is at most $20dt$ where $t$ is the bound on treewidth \cite{Cou10}. The behaviors of multi-pushdown systems, queue systems, message sequence charts etc. have degree at most $3$ \cite{Cyriac14}. The behaviors of timed multi-pushdown systems were already analysed for bounded special treewidth \cite{AkshayGK18}. 

\section{MPDA rungraphs and examples}\label[appsec]{app:MPDS-examples}

We define the behavior of MPDA by means of  \emph{rungraphs}. Rungraphs consists of a sequence of vertices (also called \emph{events}) labelled by transitions, linearly ordered to capture the execution order. In addition there is a binary matching relation that matches push events to the corresponding pop events. Indeed, the labelling by transitions must be locally compatible, and the matching relation must represent the LIFO policy of stacks.

Formally, it is given by the structure $\rungraph = \tuple{\vertices, \Succ,\Match,\vlabel}$  where   $\vertices = \{v_1, \cdots, v_n\}$ is a finite set of vertices, $\Succ$ is the successor relation of a total order on $\vertices$,  $\Match \subseteq \vertices \times \vertices$ is a binary relation  and the  labelling $\vlabel: \vertices \mapsto \Sigma$ labels vertices by letters. The rungraph $\rungraph$ is accepted by the MPDA $\MPDS$ if there is a map $\tmap: \vertices \mapsto \trans$  that maps  vertices  to transitions satisfying the following conditions. We first set some notations and shorthands that use in these conditions: 
The transitive closure of $\Succ$ is denoted by $<$.  For a transition $\tau = \tuple{q_1,a,\op,q_2}$, we will use $\tgtstate(\tau) =q_2$ and $\srcstate(\tau) = q_1$ to refer to the target and the source state of the transition. Similarly, we will use $\letter(\tau) = a$ and $\oper(\tau) = \op$.
\begin{itemize}
  \item The map $\tmap$ must be consistent with the labelling $\lambda$: $\letter(\tmap(v)) = \lambda(v)$ for all $v \in \vertices$.
	\item   The stack access policy must be respected by $\Match$: If $(u,v) \in \Match$ then 
	\matchitem $u <v$, 
	\matchitem there is an $i \in \stackset$ and some $a \in \stackalp$ such that \matchsubitem $\oper(\tmap(u)) = \mpdspush{i}{a}$ and $\oper(\tmap(v)) = \mpdspop{i}{a}$
	\matchsubitem there is no $(u',v') \in \Match$ such that  for some $b \in \stackalp$, $\oper(\tmap(u')) = \mpdspush{i}{b}$, $\oper(\tmap(v')) = \mpdspop{i}{b}$, and $u < u' < v < v'$ or $ u' < u < v' < v$, 
	\matchitem \label{condition:M3} any event labelled by a push (resp. pop) transition must be the source (resp. target) of a $\Match$ edge. 
	\item The map $\tmap$ must be consistent between consecutive events: if $(u,v) \in \Succ$ then $\tgtstate(\tmap(u)) = \srcstate(\tmap(v))$.
	\item  For the first event $u$ (i.e. there is no $v$ such that $(u,v) \in \Succ$), $\srcstate(\tmap(u)) = \initstate$, and
	\item  For the last event $v$, $\tgtstate(\tmap(v)) = \finstates$.
\end{itemize}

\begin{example}\label{ex:mpds-wstar}
	Consider the non context-free language $L_\text{wstar}=\{(\#w)^n\mid w \in \{a,b\}^\ast, n \ge 1\}$. We give a 2-stack MPDA $\MPDS_\textsf{wstar}$  in Figure~\ref{fig:mpds-wwstar} for  $L_\text{wstar}$. It uses $\{\#, a, b\}$ as the stack alphabet for both the stacks, and $\# $ is used as a stack-bottom indicator. In state $q_1$ it reads the word $w$ for the first time and stores it in stack 1 in reverse.  In state $q_2$ it copies stack 1 to stack 2 on $\epsilon$-moves to reverse it again. In state $q_7$ it reads $w$ again from stack 2, and copies it to stack 1 in reverse. From $q_1$ or $q_6$ it can choose to stop producing more copies of $w$ and just empty stack 1 in state $q_{10}$ and accept in $q_{11}$.

	A rungraph is depicted in Figure~\ref{fig:rungraph}. The map $\tmap$ is explicitly given by the node labeling. The $\Match$ relation is given by the curved edges, where the upward (blue) curves correspond to stack~1 and the downward (brown) curves correspond to stack~2. This rungraph generates the word $\#abaa\#abaa$. The pattern repeats for every additional $\#abaa$. 
\end{example}

\begin{figure}
	\scalebox{0.9}{
		\begin{tikzpicture}[node distance=2.5cm, initial text=,>=latex, thick]
			\node[ draw, rounded corners,  initial ] (q0) {$q_0$};
			\node[ draw, rounded corners, right=of q0 ] (q1) {$q_1$};
			\node[ draw, rounded corners, right=of q1 ] (q2) {$q_2$};
			\node[ draw, rounded corners, above=1.5cm of q2 ] (q3) {$q_3$};
			\node[ draw, rounded corners, below=1.5cm of q2 ] (q4) {$q_4$};
			\node[ draw, rounded corners, below right =0.3cm and 2.5cm of q3] (q5) {$q_5$};
			\node[ draw, rounded corners, above right= 0.3cm and 2.5cm of q4 ] (q6) {$q_6$};
			\node[ draw, rounded corners, above right= 0.3cm and 2.5cm of q5 ] (q8) {$q_8$};						
			\node[ draw, rounded corners, below right=0.3cm and 2.5 cm of q6 ] (q9) {$q_9$};			
			\node[ draw, rounded corners, above=1.5cm of q9 ] (q7) {$q_7$};
			\node[  draw, rounded corners, below=1.5cm of q0 ] (q10) {$q_{10}$};
			\node[  accepting, draw, rounded corners, below=1.5cm of q1 ] (q11) {$q_{11}$};
			
			\draw[->] 
			(q0) edge node[above] {$\#, \mpdspush{1}{\#}$}  node[below, pos=0.5] {\tname{1}}(q1)
			(q1) edge [loop above, min distance =1.5cm, out=120, in = 60] node[above] {$a, \mpdspush{1}{a}$} node[below] {\tname{2}} (q1)
			(q1) edge [loop below, min distance =1.5cm, out = -60, in = -120] node[below] {$b, \mpdspush{1}{b}$}  node[above] {\tname{3}} (q1)		
			(q1) edge node[above] {$ \mpdspush{2}{\#}$}  node[below] {\tname{4}} (q2)	
			(q2) edge[bend right] node[right, pos = 0.7, xshift =0mm] {$ \mpdspop{1}{a}$}   node[right, pos = 0.9] {\tname{5}}(q3)		
			(q3) edge[bend right] node[left] {$ \mpdspush{2}{a}$}  node[left, pos = 0.1] {\tname{6}} (q2)						
			(q2) edge[bend left] node[right, pos =0.7, xshift = 0mm] {$ \mpdspop{1}{b}$}   node[right, pos = 0.9] {\tname{7}} (q4)		
			(q4) edge[bend left] node[left] {$ \mpdspush{2}{b}$}   node[left, pos = 0.1] {\tname{8}} (q2)		
			(q2) edge node[below, sloped, pos =0.7] {$ \mpdspop{1}{\#}$}  node[above] {\tname{9}} (q5)	
			(q5) edge node[below, sloped, pos=0.4] {$ \#, \mpdspush{1}{\#}$}  node[above] {\tname{10}} (q7)
			(q7) edge node[above, sloped, pos=0.7] {$ \mpdspop{2}{\#}$}  node[below] {\tname{15}}(q6)	
			(q6) edge node[above, sloped, pos=0.4] {$ \mpdspush{2}{\#}$}  node[below] {\tname{16}} (q2)				
			(q7) edge[bend left] node[left, pos = 0.7,  xshift=0mm] {$ \mpdspop{2}{a}$}   node[left, pos = 0.9] {\tname{11}} (q8)		
			(q8) edge[bend left] node[right] {$a, \mpdspush{1}{a}$}   node[right, pos = 0.1] {\tname{12}} (q7)						
			(q7) edge[bend right] node[left, pos = 0.7, xshift=0mm] {$ \mpdspop{2}{b}$}  node[left, pos = 0.9] {\tname{13}}(q9)		
			(q9) edge[bend right] node[right] {$b, \mpdspush{1}{b}$}  node[right, pos=0.1] {\tname{14}} (q7)		
			(q1) edge node[above, sloped] {$ \noop$}  node[below, sloped] {\tname{17}} (q10)	
			(q6) edge[bend left, ] node[above] {$ \noop$}   node[above, pos = 0.6] {\tname{18}} (q10)	
			(q10) edge [loop above] node[above] {$\mpdspop{1}{a}$}  node[left, pos = 0.2] {\tname{19}} (q10)
			(q10) edge [loop below] node[below] {$\mpdspop{1}{b}$} node[left, pos =0.8 ] {\tname{20}}  (q10)
			(q10) edge  node[below, pos = 0.7] {$\mpdspop{1}{\#}$} node[above, pos =0.5 ] {\tname{21}}  (q11)
			;
		\end{tikzpicture}
	}
	\caption{A 2-stack MPDA $\MPDS_\textsf{wstar}$ accepting the language  $L_\text{wstar}=\{(\#w)^n\mid w \in \{a,b\}^\ast, n \ge 1\}$. We omit $\epsilon$ from a transition label for readability. }\label{fig:mpds-wwstar}
\end{figure}
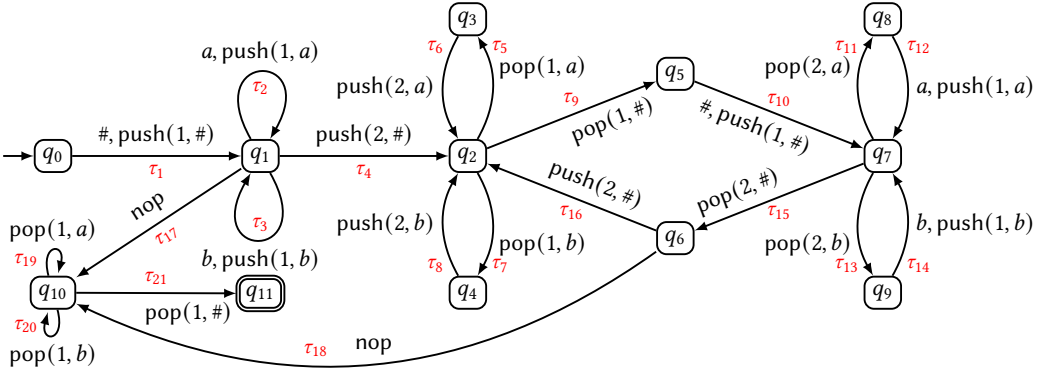

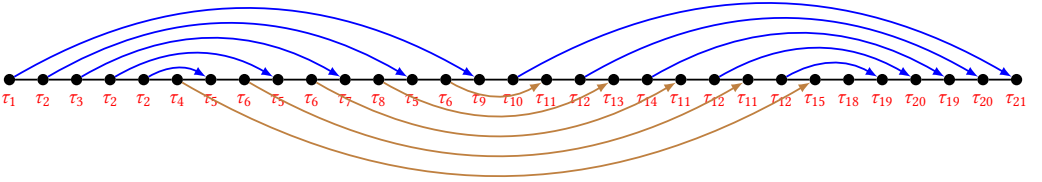
\begin{figure}
	\scalebox{0.85}{
		\begin{tikzpicture}[node distance=3.5mm, initial text=,>=latex, thick,	runnode/.style ={fill,circle,inner sep=0pt, outer sep = 0pt, minimum size=5pt}]
			
			\node[runnode, label = below:\tname{1}] (n1) {};

			\foreach \x / \y[remember=\x as \lastx (initially 1)] in {2/2, 3/3, 4/2, 5/2, 6/4, 7/5, 8/6, 9/5, 10/6, 11/7, 12/8, 13/5, 14/6, 15/9, 16/10, 17/11, 18/12, 19/13, 20/14 , 21/11, 22/12, 23/11, 24/12, 25/15, 26/18, 27/19, 28/20, 29/19, 30/20, 31/21}
			{\node[runnode, right = of n\lastx, label = below:\tname{\y}] (n\x) {};}
			
			\draw[-] 	\foreach \x [remember=\x as \lastx (initially 1)] in {2,...,31}
			{(n\lastx) edge (n\x)} ; 
			
			\draw[->, bend left, color=blue] 
			\foreach \x / \y  in {5/7, 4/9, 3/11, 2/13, 1/15}
			{(n\x) edge (n\y)} ; 
			\draw[->, bend right, color=brown] 
			\foreach \x / \y  in {14/17, 12/19, 10/21, 8/23, 6/25}
			{(n\x) edge (n\y)} ; 
			\draw[->, bend left, color=blue] 
			\foreach \x / \y  in {24/27, 22/28, 20/29, 18/30, 16/31}
			{(n\x) edge (n\y)} ; 
			
		\end{tikzpicture}
	}
	\caption{A rungraph  }\label{fig:rungraph}
\end{figure}

\section{Multiple context-free grammars and languages}
\label[appsec]{appendix:mcfg}

To make the definition precise, we need some notation.

\myparagraph{Notation}
A \emph{ranked set} $\hat A$ is a set $A$ together with a function $\arity{}: A \to \Nat$ denoting the rank/arity of each element in $A$. We may write a ranked set as a pair $\hat A = (A, \arity)$. If $A= \{a_1, \dots, a_k\}$ is a finite set, we may also write  the ranked set as $\hat A = \{a_1^{\arity(a_1)}, \dots, a_k^{\arity(a_k)}\}$, where each element is tagged with its arity. 
Given a set $B$ disjoint from $A$, the set of \emph{atomic terms} of $\hat A$ over $B$, denoted $\atomicTermsof{\hat A}{B}$, is the set $\{a(b_1, \dots, b_k ) \mid a \in A, \arity(a) = k, b_i \in B, \text{ for all } i \in \set{k}\}$. A subset $S$ of $\atomicTermsof{\hat A}{B}$ is called \emph{repetition-free} if 
\begin{enumerate}
	\item for all terms $a(b_1, \dots, b_k) \in S$, $b_i \neq b_j$ if $i \neq j$, and
	\item for every pair of distinct terms $a(b_1, \dots, b_k)$  and $a'(b'_1, \dots, b'_{k'} )$ in $S$, $b_i \neq b'_j$ for all $i \in \set{k}, j\in \set{k'}$.
\end{enumerate}
That is, every argument is unique in a repetition-free set. We denote that $S$ is a repetition-free subset of $\atomicTermsof{\hat A}{B}$ by $S \rfsubset \atomicTermsof{\hat A}{B}$.

\myparagraph{Multiple Context-Free Grammars}
A  \emph{multiple context free grammar} (\mcfg) is a tuple of the form $\grammar= (\Terminals, \hat\Nonterminals, \Variables, \Productions, \StartNonterminal)$ where:  
\begin{itemize}
	\item $\Terminals$ is a finite  alphabet. 
	\item $\hat\Nonterminals = (\Nonterminals, \arity)$ is a finite ranked set of nonterminals with $\arity(A) \ge 1$ for all $A \in \Nonterminals$. $\Nonterminals$ is disjoint from $\Terminals$. %
	\item $\nt{\StartNonterminal}{1} \in \hat\Nonterminals$ is the start nonterminal. 
	\item $\Variables$ is a finite set of variables disjoint from $\Nonterminals$ and $\Terminals$.
	\item $\Productions$ is a finite set of \emph{productions} $\pi$ of the form $\mcfgrule{A(s_1, \dots s_k)}{ S}$ where 
	\begin{itemize}
		\item $A \in \Nonterminals$ is a nonterminal with $\arity(A) = k$. $A$ is called the \emph{head} of $\pi$, denoted $\head(\prodrule)$. 
		\item $S$ is a repetition-free subset of $ \atomicTermsof{\hat \Nonterminals}{\Variables}$.  That is, $S  \rfsubset \atomicTermsof{\hat{\Nonterminals}}{\Variables}$. $S$ is called the \emph{body} of $\pi$, denoted $\body(\prodrule)$.   
		\item $s_j \in (\Terminals \cup \Variables)^\ast$ for $1 \le j \le  k$.
		\item  Let $s = s_1s_2 \dots s_k$, be the concatenation of the arguments of $A$.  Then, for every term $t = B(y_1, y_2, \dots, y_n) \in S$, each argument $y_i$ of $t$ should appear at most once in $s$. That is, the productions are copyless, but may permute the arguments (and may concatenate additional symbols from $\Sigma$) and may also forgot some arguments. %
	\end{itemize}
\end{itemize}
If $\pi = A(s_1, \ldots, s_n) \leftarrow S$ is a production rule and $S = \setof{B_1(x_{1,1}, \ldots, x_{1,n_1}), \ldots}$ we define $\tgt_\pi(x_{\ell,i}) = j$ if $x_{\ell,i}$ occurs in $s_j$. We define $\src_{\pi,\ell}(s_j) = \setof{i \mid \tgt_\pi(x_{\ell,i}) = j}$. Where the production is clear from context, we omit it.

\smallskip
A nonterminal $A$  of an \mcfg\ $\grammar$ with $\arity(A) = k$ generates a $k$-tuple of words $(u_1, u_2, \dots u_k) \in (\Sigma^*)^k$ if there is a parse tree/derivation tree for it (to be defined below). We write $A \rightarrow^h (u_1, u_2, \dots u_k) $  to say that there is \textit{an $A$-rooted}  parse tree $t$ of $A$  of height $h$ with yield  $(u_1, u_2, \dots u_k)$ (denoted $\yield(t) = (u_1, u_2, \dots u_k)$).    We define this inductively: 
\begin{itemize}
	\item $A \rightarrow^0 (u_1, u_2, \dots u_k)$ if $\mcfgrule{A(u_1, u_2, \dots u_k)}{ \emptyset}$ is a production rule. $A (u_1, u_2, \dots u_k)$ is an $A$-rooted derivation tree with yield $ (u_1, u_2, \dots u_k)$.
	\item $A \rightarrow^{h} (u_1, u_2, \dots u_k)$ if there is production rule $\prodrule =  \mcfgrule{A(s_1, \dots s_k)}{ S}$ and  a homomorphism $\eta: \Variables \to \Sigma^\ast$ 
	such that
	\begin{enumerate}
		\item For all terms $ B(y_1,  \dots, y_n) \in S$, $B \rightarrow^{\ell}(w_1,\dots w_{n})$ for some $\ell < h$ where $w_j = \eta(y_j)$.
		\item $u_i =  \bar \eta(s_i)$ where $\bar \eta$ extends $\eta$ to $\Sigma \cup \Variables$ in the obvious way ($\bar \eta (a) = a$ for all $a \in \Sigma$).
	\end{enumerate}
	Let $S = \{B_i(x_{i,1}, \ldots, x_{i,m}) \mid i \in \set{|S|}\}$, and  suppose $t_i$  is a $B_i$-rooted tree. Then the tree constructed with  root node $(A, \pi)$ and children $t_i$ (written $(A, \pi)[t_1, \ldots, t_{|S|}]$) is an ($A$-rooted) derivation tree.
\end{itemize}

We address subtrees using paths $p \in \Nat^*$. If $p = \varepsilon$, then the subtree of a tree $t$ at $p$, written $\subtree{t}{p}$, is~$t$. If $p = i \cdot p'$, then $\subtree{t}{p}$ is $\subtree{t_i}{p'}$, where $t_i$ is the $i$-th child of $t$.

\medskip

We write $A \rightarrow^\ast (u_1, u_2, \dots u_k) $  if $A \rightarrow^h (u_1, u_2, \dots u_k) $ for some $h \in \Nat$. A nonterminal $A$  with $\arity(A)=k$ generates a $k$-ary relation, denoted $\Relof{A}{\grammar} \subseteq (\Sigma^*)^k$. That is, 
$\Relof{A}{\grammar} = \{(u_1, u_2, \dots u_k)  \mid A \rightarrow^\ast (u_1, u_2, \dots u_k)\}$.
When the \mcfg \ $\grammar$ is clear from the context, we  simply write $\Relof{A}{}$ instead of $\Relof{A}{\grammar}$.
Note that the start nonterminal  $\StartNonterminal$ generates a unary relation, i.e.\ a subset of $\Sigma^*$, which we also call the language of the \mcfg\  $\initgrammar$. That is, $\langof{\initgrammar} = \Relof{\StartNonterminal}{}$.  A language $L$ is called a\textit{ multiple context-free language} (MCFL) if there is an \mcfg\ $\initgrammar$, such that $L = \langof{\initgrammar}$. 
We say that an \mcfg\  $\grammar= (\Terminals, \hat\Nonterminals, \Variables , \Productions, \StartNonterminal)$ is a \kmcfg\ if for all $A \in \Nonterminals$, $\arity(A) \le k$. 

\section{MCFL examples}\label{app:MCFL-examples}
\begin{example}
	Consider the \twoicfg\ over the nonterminals $\{	\nt{B}{1}, 	\nt{A}{2}\}$ with the following productions. 
	\begin{align*}
		\mcfgrule{B( xby )&}{ \{A(x, y)\}}\tag{$\prodrule_1$}\\
		\mcfgrule{A(a, xy)&}{  \{A(x, y)\}}\tag{$\prodrule_2$}\\
		\mcfgrule{A( xy, a)&}{\{A(x, y)\}}\tag{$\prodrule_3$}\\
		\mcfgrule{A( a , a)&}{  \emptyset} \tag{$\prodrule_4$}
	\end{align*}
	Here, $\Relof{A}{} =(a^+,a) \cup (a, a^+)$, and $\Relof{B}{} = a^+ba +aba^+$. We have  $\downc{\Relof{B}{}} = a^\ast(b+\epsilon)(a+\epsilon) + (a+\epsilon)(b+\epsilon)a^\ast$.  We have $\downc{\Relof{A}{}} = (a^\ast, a+\epsilon) + (a+\epsilon, a^\ast)$.
	\qed
\end{example}

\begin{example}
	Consider the  $\twoicfg$ $\initgrammar_{ww} = (\{a,b\}, \hat\Nonterminals, \Variables, \Productions, \StartNonterminal)$ where $\hat\Nonterminals = \{ \nt{A}{2}, 	\nt{\StartNonterminal}{1} \} $, $\Variables = \{x_1, \dots, x_4\}$ and  $\Productions$ listed below 
	\begin{align}
		\StartNonterminal( x_1x_2 ) &\leftarrow \{A(x_1,x_2)\} \tag{$\prodrule_1$}\\
		A( x_1x_3, x_2x_4) &\leftarrow \{ A(x_1,x_2), A(x_3,x_4)\}  \tag{$\prodrule_2$}\\
		A(a, a) & \leftarrow \emptyset \tag{$\prodrule_3$}\\
		A(b, b) & \leftarrow \emptyset \tag{$\prodrule_4$}
	\end{align}

	Here, $\Relof{A}{} =\{(w,w) \mid w\in \Sigma^+\}$, and $\langof{\initgrammar_{ww}} =\Relof{\StartNonterminal}{} = \{ww\mid w\in \Sigma^+\} $. We have  $\downc{\Relof{\StartNonterminal}{}} = (a+b)^\ast$, and $\downc{\Relof{A}{}} = ((a+b)^\ast, (a+b)^\ast)$.
	\qed
\end{example}

\section{Tree automata to MCFG}\label[appsec]{app:ta2mcfg}
\subsection{Summaries and their computation}

Given a \cpclog{} $\cpcgraph = \tuple{\vertices, \sedge ,(\edges_\gamma)_{\gamma \in\Elabels},\vlabel, \colorlab}$, we are  interested in the summary wrt $\sedge$  and hence it suffices to consider the graph $\cpcgraph_{\sedge} = \tuple{\vertices, \sedge, \colorlab}$.

\textbf{Piece Summary} The graph $\cpcgraph_{\sedge}$ is a collection of pieces. 
We abstract a piece just by the pair of colors of the extreme vertices. That is, if $u_1  u_2  \ldots u_n$ is a piece (i.e., $(u_i, u_{i+1}) \in \sedge$), then its abstraction is $(\colorlab(u_1), \colorlab(u_n))$. While strict \cpclog{}s technically require all extreme vertices to be colored, removing the coloring, i.e.\ via $\colorless{\cdot}$, could still lead to an \LOG{}, even if the first and last vertex in the $\sedge$-order are not colored. In particular, it is possible that $\colorlab(u_1)$ is undefined, which means that no predecessor can be added to $u_1$. In this case, we let $\colorlab(u_1)$ be $\vdash$. Clearly, we only need to consider at most one such vertex. Similarly, if $\colorlab(u_n)$ is undefined we let it be $\dashv$.
By $\compabstr(\cpcgraph)$ we denote the set of pairs of start and end colors of the   pieces of $\cpcgraph_{\sedge} = \tuple{\vertices, \sedge, \colorlab}$. Note that, $\compabstr(\cpcgraph) \subseteq {(\colors \cup \{\vdash\}) \times (\colors \cup \{\dashv\}) }.$ Note that $(c,c)$ is a possible abstraction in $\compabstr(\cpcgraph)$ which means the  corresponding piece has only one vertex, and that is colored $c$. Also, if $(c_1, c_2)$ and $(c_3, c_4)$ are both in $\compabstr(\cpcgraph)$, then $\{c_1, c_2\} \cap \{c_3, c_4\} = \emptyset$.

For any fixed set of colors $\colors$, %
the set of all possible piece summaries is $\Smrycomp{\colors} = 2^{(\colors \cup \{\vdash\}) \times (\colors \cup \{\dashv\})}$.
Next we show that the we can  design a tree automaton %
$\TA_\summarycomp$ over expressions from  $\expr(C, \Alphabet, \Elabels)$, such that for all $\psi \in \expr(C, \Alphabet, \Elabels)$, $\TArun{\TA_\summarycomp}(\psi) = \compabstr(\cpcgraph_\psi)$.

	\subsection{Tree automata computing summaries}

	The tree automaton   $\TA_\summarycomp = \tuple{\Smrycomp{\colors}, \deltacomp,\emptyset, \{(\vdash, \dashv)\}}$ computes the piece summaries.  We define  $\deltacomp =\bigcup_{\op \in \UOP} \deltacomp_\op \cup \deltacomp_\join$ as follows. Here, $\summ, \summ_1, \summ_2  \in \Smrycomp{\colors}$.
	\begin{itemize}
		\item    $\deltacomp_{\node{c}{a}}(\summ) = \summ \cup \{(c,c)\} $
		\item  $\deltacomp_ {\addsedge{c_1}{c_2}}(\summ) = \begin{cases} 
			\text{undefined } & \text{ if } (c_2, c_1) \in \summ \\
			\summ' \cup \{(x,y)\} & \text{ if } \summ = \summ' \cup \{(x, c_1), (c_2, y)\} \text{ for } x, y \in \colors \cup \{\vdash, \dashv\}\\
			\text{undefined } & \text{otherwise} 
		\end{cases}
		$
		\item$\deltacomp_{\addedge{\gamma}{c_1}{c_2}}(\summ) = \summ$
		\item $\deltacomp_{\fcolor{c}}(X,\summ) = \begin{cases}  		\summ &  \text{ if } \neg \exists y ((c, y) \in \summ \vee (y, c ) \in \summ)\\
			
			\summ' \cup \{(\vdash,y)\}& \text{ if } \summ = \summ' \cup \{ (c, y)\}, y \neq c, \neg \exists x \, (\vdash, x) \in \summ \\
			\summ' \cup \{(x,\dashv)\} & \text{ if } \summ = \summ' \cup \{ (x, c)\}, x \neq c, \neg \exists y \, ( y, \dashv) \in \summ \\
			\summ' \cup \{(\vdash,\dashv)\} & \text{ if } \summ = \summ' \cup \{ (c, c)\}, 
			\neg \exists y \, ((\vdash, y) \in \summ \vee (y, \dashv ) \in \summ)\\
			\text{undefined } & \text{otherwise} 
		\end{cases}
		$
		\item $\deltacomp_\join(\summ_1, \summ_2) = \summ_1 \cup \summ_2$
	\end{itemize}
	Note that the transition $\deltacomp_ {\addsedge{c_1}{c_2}}$ forbids cycles and multiple $\sedge$-successors/predecessors for a vertex. The transition $\deltacomp_{\fcolor{c}}$ ensures that there is at most one $\vdash$ (or $\dashv$) labeled vertex in the graph. The accepting state guarantees that there is a single piece in the end.

\newcommand{\permutation}{g}
\textbf{Size bound} We argue that number of states of $\TA_\summarycomp$ is at most the number of permutations  of the set $\colors \cup \{\vdash, \dashv\}$. We will show an injective map $h$ from $\Smrycomp{\colors}$ to permutations $S_n$ of the set $\colors \cup \{\vdash, \dashv, \bot\}$, where $n = |\colors| + 3$. Let $h(\summ) = \permutation \in S_n$, where $\permutation$ is defined as follows.  Fix an ordering $\prec$ on $\colors \cup \{\vdash, \dashv, \bot\}$. If  $(c_1, c_2) \in \summ$, then $\permutation(c_1) = c_2$ and $\permutation(c_2) = c_1$.  For every color that is not present in $\summ$, let $\permutation$ form a big cycle (cyclically closing the $\prec$) with all those colors. Note that $\bot$ does not appear in $\summ$ and hence it always identifies the missing colors from $\summ$. In conclusion, the number of states is at most $|S_n| = n! \leq n^n = 2^{\mathcal{O}(n\log n)}= 2^{\mathcal{O}(|C|\log |C|) } $.

\subsection{Tree Automata to MCFG}

Given a tree automaton $\TA = \tuple{Q, \delta^\TA, q_\bot, Q_\text{accept}}$ over $\expr(C, \Alphabet, \Elabels)$, we 
construct the tree automaton $\TB$, the Cartesian product of $\TA$ and $\TA_\summarycomp$. That is, the states of $\TB$ are of the form $\tuple{q,\summ}$, where $q \in Q$ and $\summ \in \Smrycomp{\colors}$.
Since also $\langof{\TA_\summarycomp} \subseteq \expr(C, \Alphabet, \Elabels)$, we have 
\begin{claim}\label{claim:LA-LB}
  \begin{align*}
  	\langof{\TB} &= \langof{\TA} \cap \{\psi \mid \cpcgraph_\psi \text{ is a  (colored) \LOG}\}\text{, which implies}
  	\\ \pglangof{\TB} &=  \glangof{\TB} = \glangof{\TA}\text{, which implies}\\ 
  	\wlangof{\TB} &= \wlangof{\TA}\text{.}
  \end{align*}
\end{claim}
Given the automaton $\TB$ over $\expr(C, \Alphabet, \Elabels)$ with transition relation $\delta^\TB$, we will construct an \mcfg{} $\grammar_\TB= (\Terminals$, $\hat\Nonterminals_\TB$, $\Variables_\TB$, $\Productions_\TB$, $\StartNonterminal)$ as follows. Let  $\prec$ 
be a total order on $\colors \cup \{\vdash, \dashv\}$. For a summary  $\summ \subseteq (\colors \cup \{\vdash, \dashv\})^2$, we define $\tup(\summ) = (x_{c_{1}, c'_{1}}, x_{c_{2}, c'_{2}} \dots x_{c_{m}, c'_{m}})$ be the unique tuple such that $\summ = \{(c_{1}, c'_{1}), (c_{2}, c'_{2}), \dots ,(c_{m}, c'_{m})\}$ and $c_{j} \prec c_{{j+1}}$ for all $j: 1 \le j \le m-1$. That is, we order the pairs in the set $\summ$ according to the first member of each pair wrt.\ $\prec$, to obtain a tuple of variables indexed by the pairs. The length of the tuple $\tup(\summ)$ is indeed $|\summ|$. For a tuple of variables $\tau$, $\tau[x \mapsfrom y]$ is the unique tuple obtained by replacing every occurrence of $x$ with $y$.
\begin{itemize}
	\item $\hat\Nonterminals_\TB = (Q \times \Smrycomp{\colors}, \arity)$ where $\arity\big(\tuple{q,\summ}\big)= |\summ|$. That is, the arity of a nonterminal is indeed the number of pieces. 
	\item $\Variables_\TB = \{x_{c_1, c_2} \mid c_1, c_2 \in \colors \cup \{\vdash, \dashv\}\}$
	\item $\Productions_\TB = \{\prodrule_t  \mid t \in \delta^\TB\}$. For $t \in \delta^\TB, \prodrule_t$ is defined as follows: 
	\begin{itemize}
		\item if $t \equiv \delta^\TB_{\node{c}{a}}(\tuple{q, \summ}) \mapsto \tuple{q', \summ'}$ then 
		$$\prodrule_t \equiv  \tuple{q', \summ'}(\tup(\summ')[x_{c,c} \mapsfrom a]  ) \leftarrow  \{\tuple{q, \summ}(\tup(\summ))\}$$
		\item if $t \equiv \delta^\TB_{\addsedge{c_1}{c_2}}(\tuple{q, \summ}) \mapsto \tuple{q', \summ'}$ then
		$$\prodrule_t \equiv  \tuple{q', \summ'}(\tup(\summ')[x_{c_3,c_4} \mapsfrom x_{c_3,c_1}x_{c_2,c_4}]  ) \leftarrow  \{\tuple{q, \summ}(\tup(\summ))\}$$
		where $(c_3, c_1) , (c_2, c_4) \in \summ$ . 
		\item if $t \equiv \delta^\TB_{\addedge{\gamma}{c_1}{c_2}}(\tuple{q, \summ}) \mapsto \tuple{q', \summ}$ then 
		$$\prodrule_t \equiv  \tuple{q', \summ}(\tup(\summ )) \leftarrow  \{\tuple{q, \summ}(\tup(\summ))\}$$
		\item if $t \equiv \delta^\TB_{\fcolor{c}}(\tuple{q, \summ}) \mapsto \tuple{q', \summ} $ then 
		$$\prodrule_t \equiv  \tuple{q', \summ}(\tup(\summ )) \leftarrow  \{\tuple{q, \summ}(\tup(\summ))\}$$
		\item if $t \equiv \delta^\TB_{\fcolor{c}}(\tuple{q, \summ}) \mapsto \tuple{q', \summ'}$ with $\summ \neq \summ'$ then 
		$$\prodrule_t \equiv  \tuple{q', \summ'}(\tup(\summ')[x_{c_1,c_2} \mapsfrom x_{c_3,c_4} ]  ) \leftarrow \{\tuple{q, \summ}(\tup(\summ))\}$$
		where $(c_1, c_2)  \in \summ' \setminus \summ$, and $(c_3, c_4)  \in \summ \setminus \summ'$. 
		\item if $t \equiv \delta^\TB_{\join}(\tuple{q_1, \summ_1}, \tuple{q_2, \summ_2}) \mapsto \tuple{q, \summ}$ then 
		$$\prodrule_t \equiv  \tuple{q, \summ}(\tup(\summ )) \leftarrow  \{\tuple{q_1, \summ_1}(\tup(\summ_1)),  \tuple{q_2, \summ_2}(\tup(\summ_2))\}$$
	\end{itemize}
\end{itemize}

Finally, we add the 
  production rule $\prodrule_q$  for each $q \in Q_{\text{accept}}$.
$$\prodrule_q \equiv {A_\text{start}}(x_{\vdash, \dashv}) \leftarrow \tuple{q, \{(\vdash,\dashv)\}}(x_{\vdash,\dashv})$$
\begin{claim}
	$ \wlangof{\TB} = \langof{\initgrammar_\TB}$
\end{claim}
Hence by Claim~\ref{claim:LA-LB} we get $ \wlangof{\TA} = \langof{\initgrammar_\TB}$. Since $\TB$ is the Cartesian product of $\TA$ and $\TA_\summarycomp$ the  size upper bounds claimed in Theorem~\ref{thm:TA2MCFL} follow.  This completes the proof of Theorem~\ref{thm:TA2MCFL}.

\section{MCFL to MSO}\label[appsec]{app:mso}
\subsection{Nested word of a parsetree}\label[appsec]{app:NestedWord}

\newcommand{\tlen}{\Terminals\text{-}length}
The \LOG\ of a parsetree is defined inductively. Corresponding to a subtree $t$ we will have a piecewise-\LOG\ $\lograph_t$. Recall that a piecewise-\LOG\ is defined like a \cpclog, but without the vertex coloring $\colorlab$. %
The end-points of the pieces in $\lograph_t$  will be connected by a matching edge $\matching$.
More precisely if the root of $t$ is a nonterminal $\nt{A}{n}$, then there will nodes $u_{t,1}^b, u_{t,1}^e,  u_{t,2}^b, u_{t,2}^e, \dots u_{t,n}^b, u_{t,n}^e$ such that for each $i$ 
\begin{enumerate}
	\item $u_{t,i}^b < u_{t,i}^e $ (recall, $<$ is the transitive closure of $\Succ$)
	\item there is no $u$ such that $ \sedge (u,u_{t,i}^b)$ and
	\item there is no $u$ such that $ \Succ( u_{t,i}^e , u)$.
\end{enumerate}
Further, if the the production applied at the root of $t$ is of the form \ref{eqn:production-rule-form} (see p.~\labelcpageref{eqn:production-rule-form}),  then the root of $t$ has $m$ children, each rooting subtrees $t_1, \dots t_m$ respectively.  Then the piecewise-\LOG\ $\lograph_t$ will be $\tuple{\vertices_t, \sedge_t, \matching_t, \vlabel_t}$  where 
\begin{DIFnomarkup}
\begin{itemize}
	\item $\vertices_t = \bigcup_{i \in \set{m}} \vertices_{t_i} \cup \{v_{t,i,j} \mid i \in \set{n}, \text{$j$th letter of $s_i$ is from $\Terminals$ } \}$\\$ \cup \{u_{t,1}^b, u_{t,1}^e,  u_{t,2}^b, u_{t,2}^e, \dots u_{t,n}^b, u_{t,n}^e\}$.%
	\item  $\sedge_t = \bigcup_{i \in \set{m}} \sedge_{t_i}$\\
	$ \cup \{(u_{t,i}^b, v_{t,i,1}) \mid \text{ if $1$st letter of $s_i$ is from $\Terminals$}\} \cup \{(u_{t,i}^b, u_{t_j,k}^b) \mid \text{ if $1$st letter of $s_i$ is $x_{j,k}$}\} $\\
	$\cup \{( v_{t,i,|s_i|}, u_{t,i}^e) \mid \text{ if the last letter of $s_i$ is from $\Terminals$}\} \cup \{(u_{t_j,k}^e, u_{t,i}^e ) \mid \text{ if the last letter of $s_i$ is $x_{j,k}$}\} $ \\
	$\cup \{(v_{t,i,j}, v_{t,i,j+1}) \mid \text{ if $j$th letter and ($j+1$)th letters of $s_i$ are from $\Terminals$}\}$\\
	$\cup \{(v_{t,i,j}, u_{t_j,k}^b) \mid \text{ if $j$th letter  of $s_i$ is from $\Terminals$ and ($j+1$)th letter of $s_i$ is $x_{j,k}$}\}$\\
	$\cup \{( (u_{t_j,k}^e, v_{t,i,j+1}) \mid \text{ if  $j$th letter of $s_i$ is $x_{j,k}$ and $(j+1)$th letter  of $s_i$ is from $\Terminals$}\}$\\
		$\cup \{( (u_{t_j,k}^e, u_{t_{j'},k'}^b) \mid \text{ if  $j$th letter of $s_i$ is $x_{j,k}$ and $(j+1)$th letter  of $s_i$ is $x_{j',k'}$}\}$
		\item  $\matching_t = \bigcup_{i \in \set{m}} \matching_{t_i} \cup \{(u_{t,i}^b, u_{t,i}^e) \mid i \in \set{n}\} \cup \{(u_{t,i}^e, u_{t,i+1}^b) \mid i \in \set{n-1}\}$
		\item $\vlabel_t = \bigcup_{i \in \set{m}} \vlabel_{t_i} \cup \{v_{t,i,j} \mapsto a \mid \text{$j$th letter of $s_i$ is a}\}$
\end{itemize}
\end{DIFnomarkup}
Here, vertices unlabelled by $\vlabel_t$ are treated as having label $\epsilon$.

\subsection{Well-nesting in MSO}\label[appsec]{app:MSOwellnesting}
We need some shorthands. 
\begin{align*}
\mpath(x_1, \dots x_n) &\equiv x_1 \matching x_2 \matching \dots x_n\\
\maxmpath(x_1, \dots, x_n) &\equiv \mpath(x_1, \dots, x_n) \wedge \forall x  (\neg \mpath(x, x_1, \dots, x_n) \wedge \neg \mpath(x_1, \dots, x_n, x ))\\
\eps(x) &\equiv \bigwedge_{a \in \Alphabet}\neg a(x) \\
\end{align*}
The shorthand $\mpath(x_1, \dots x_n)$ (resp. $\maxmpath(x_1, \dots x_n)$) states that its arguments form a $\matching$-connected path (resp.  maximal path). Note that $\mpath(x,y)$ just means $x \matching y$. The shorthand $\eps(x)$ states absence of any label from $\Terminals$, which functionally means $\epsilon$.
Let us state the formulas:
\begin{align*}
\phi_1 &\equiv \neg \exists x_1, x_2 \dots x_{2d+1} \, \, \mpath(x_1, \dots x_{2d+1} )\\
\phi_2&\equiv \bigwedge_{i\in \set{d}} \neg \exists x_1, x_2 \dots x_{2i+1} \, \, \maxmpath(x_1, \dots x_{2i+1} )\\
\phi_3 &\equiv  \forall x_1, x_2, x_3 \, (x_1 \matching x_2 \wedge x_1 \matching x_3 \implies x_2 = x_3) \wedge (x_1 \matching x_3 \wedge x_2 \matching x_3 \implies x_1 = x_2)\\
\phi_4 &\equiv \forall x\, \eps(x) \Longleftrightarrow \exists y \, (x \matching y \vee y \matching x)\\
\phi_5 &\equiv \bigwedge_{k\in \set{d}} \forall x_1, x_2, \dots x_{2k} \, \, \maxmpath(x_1, x_2, \dots x_{2k})  \implies \left( \bigwedge_{i\in \set{k}} x_{2i-1}  \le x_{2i} \right. \\
& \qquad \qquad \qquad \qquad \qquad \qquad \qquad \qquad \qquad \qquad \left. \bigwedge_{i,j\in \set{k}, i \neq j} x_{2i}  < x_{2j-1} \vee x_{2j} < x_{2i-1} \right)\\
\phi_6 &\equiv \bigwedge_{k, k'\in \set{d}} \forall x_1, x_2, \dots x_{2k} \, \forall x'_1, x'_2, \dots x'_{2k'}\\
&\qquad \qquad \qquad \qquad  \left(  \maxmpath(x_1, x_2, \dots x_{2k}) \wedge \maxmpath(x'_1, x'_2, \dots x'_{2k'}) \phantom{ \bigvee_{i'\in \set{k'}}}\right. \\
&\qquad \qquad \qquad \qquad  \left. \wedge \bigvee_{i\in \set{k}} \bigvee_{i'\in \set{k'}} (x_{2i-1} < x'_{2i'} < x_{2i}) \vee (x_{2i-1} < x'_{2i'-1} < x_{2i}) \right)\\
& \qquad \qquad \qquad \qquad \qquad \qquad  \implies \bigwedge_{i'\in \set{k'}}\bigvee_{i\in \set{k}}(x_{2i-1} < x'_{2i'-1} < x'_{2i'} < x_{2i})\\
\end{align*}
Finally
$$ \phiwellnesting \equiv \bigwedge_{i\in \set{6}} \phi_i  $$

\subsection{The formula for valid parsetree}\label[appsec]{app:MSOparsetree}

Let us give the concrete formula for this part now. The formulas will use second-order free variables $Y_A$ for each $A \in \Nonterminals$. We first define the shorthands. 
\begin{align*}
\Ampath(A, x_1, \dots x_n) &\equiv \mpath(x_1, \dots x_n) \wedge \bigwedge_{i \in \set{n}} Y_A(x_i)\\
\maxAmpath(A, x_1, \dots, x_n) &\equiv \maxmpath(x_1, \dots x_n) \wedge \Ampath(A, x_1, \dots x_n)
\end{align*}
For a production rule $\pi$ of the form \ref{eqn:production-rule-form} (see p.~\labelcpageref{eqn:production-rule-form}) we define the following shorthand. Note that $n$ is the arity of the nonterminal in the LHS of $\pi$. 
\begin{align*}
\prods(\pi, y^b_1, y^e_1, y^b_2, y^e_2, \dots, y^b_n, y^e_n) &\equiv \exists x_{1,1}^b, x_{1,1}^e, \dots x_{1,n_1}^b,x_{1,n_1}^e, \dots, x_{m,1}^b, x_{m,1}^e, \dots x_{m,n_m}^b,x_{m,n_m}^e \, \\
& \bigwedge_{i \in \set{m}} \maxAmpath(B_i, x^b_{i,1}, x^e_{i,1}, x^b_{i,2}, x^e_{i,2}, \dots, x^b_{i,n_i}, x^e_{i,n_i} ) \\
& \bigwedge_{i\in \set{n}} \nextscan(y^b_i, y^e_i, s_i)\\
\end{align*}
The macro $\scan(x,y,w)$ is supposed to scan the sequence $w$ along the top level path from $x$ to $y$. The macro $\nextscan(x,y,w)$ will do the same, but starting from the next position of $x$. We define these macros for each possible $s_i$ in the grammar and their suffixes.  
\begin{align*}
	\nextscan(x,y, w) & \equiv \exists z\, \Succ(x,z) \wedge \scan(z,y, w)&\\
	\scan(x,y,\epsilon)&\equiv x = y&\\
	\scan(x,y, aw) & \equiv a(x) \wedge \exists z\, \Succ(x,z) \wedge \scan(z,y, w)& \text{ for } a \in \Terminals\\
	\scan(x,y, \nu w) & \equiv x = \nu^b \wedge \exists x' ( x\matching x' \wedge \exists z ( \Succ(x',z) \wedge \scan(z,y, w)))& \text{ for } \nu \in \Variables\\
\end{align*}

Let us define the formulas now. 
\begin{align*}
	\psi_1 &\equiv \forall x\,  \eps(x) \implies \bigvee_{A \in \Nonterminals} \left( Y_A (x) \wedge \bigwedge_{B \in \Nonterminals\setminus\{A\}} \neg Y_B(x) \right)\\
	\psi_2 &\equiv \bigwedge_{\nt{A}{k} \in \hat\Nonterminals} \Bigg( \neg \exists x_1, \dots, x_{2k+1} \Ampath(x_1, \dots, x_{2k+1} )\\
  & \qquad \qquad \wedge \bigwedge_{m \in \set{2k-1}} \neg \exists x_1, \dots, x_{m} \maxAmpath(A, x_1, \dots, x_{m} ) \Bigg)\\
	\psi_3 &\equiv \bigwedge_{\nt{A}{k} \in \hat\Nonterminals}  \exists y^b_1, y^e_1, y^b_2, y^e_2, \dots, y^b_k, y^e_k,  \maxAmpath(A, y^b_1, y^e_1, y^b_2, y^e_2, \dots, y^b_k, y^e_k ) \implies \\
	& \qquad \qquad \qquad \qquad \qquad \qquad  \bigvee_{\pi \in \Productions, \text{LHS of } \pi \text{ is } A} \prods(\pi, y^b_1, y^e_1, y^b_2, y^e_2, \dots, y^b_k, y^e_k) \\
	\psi_4 &\equiv \exists x_\text{min} \exists x_\text{max} \forall x \,\, x_\text{min} \le x \le x_\text{max} \wedge \maxAmpath(\StartNonterminal, x_\text{min},  x_\text{max})\\
\end{align*}

Finally we have
\[ \phi_\initgrammar = \exists (Y_A)_{A \in \Nonterminals}  \bigwedge_{i \in \set{4}} \psi_i \, . \]

\subsection{Correctness Arguments}\label[appsec]{appendix:MSOcorreectness}
\subsubsection{Correctness of the constructed formula}
\begin{claim}
$\wlangof{\phi} = \langof{\initgrammar}$
\end{claim}
\begin{proof}
	\myparagraph{$\subseteq$} Suppose $w \in \wlangof{\phi}$. By definition, there is an \LOG\ $\lograph \in \glangof{\phi}$ such that $w = \word(\lograph)$. We want to argue that there is a parse tree $t$ such that 
	\begin{enumerate}
		\item $\lograph$ is the nested word of $t$, and
		\item  $\yield(t) = w$.  
	\end{enumerate} 
	Since $\lograph$ satisfies $\phi$, we look at the satisfying assignment of $Y_A$ to the $\matching$ paths, and extract a parsetree as follows. The root will be labelled by the nonterminal of the vertices $v_\text{min}$ and $  v_\text{max}$, which must be $\StartNonterminal$ thanks to $\psi_4$. Following $\psi_3$ there is a production rule that can be applied to every node built in the top down fashion this way. Thanks to $\psi_3$ this procedure gives a valid parse tree $t$, which satisfies the above two conditions. 
	
	\myparagraph{$\supseteq$} Suppose $w \in  \langof{\initgrammar}$. By definition there is a parsetree $t$ with $w = \yield(t)$. Following the construction in \cref{app:NestedWord} we get a nested word $\lograph_t$ corresponding to $t$. Notice that it satisfies all the conditions of $\phiwellnesting$ and $\phi_\initgrammar$, and hence $\lograph_t \in \glangof{\phi}$. Further $\word(\lograph_t) = w$.
		\end{proof}

\subsubsection{Special treewidth of the nested words}
\begin{claim}\label{claim:nestedwordstw}
	Nested words of a $d$-\mcfg($r$) $\initgrammar$ have special treewidth at most $6dr$. 
\end{claim}

The idea is that the well-nesting structure of the nested word would give the hint for a special tree decomposition, and more or less has the same shape as the parsetree (though a node in the parsetree is replaced by a path of bags). We define the tree-decomposition inductively. 

Let $p= (v_1, \dots ,v_{2n})$ be a maximal $\matching$ path. The partially-complete nested word covered by such a path is the induced subgraph on $V_p = \{v_1, \dots ,v_{2n}\} \cup \{u \mid v_{2i-1} < u < v_{2i}, i \in \set n \}$, call it $\lograph_p$. Claim~\ref{claim:nestedwordstw} would follow from the following claim: 
\begin{claim}\label{claim:nestedwordstwStrong}
There is a special tree decomposition for $\lograph_p$ of width at most $6dr$, and having only $\{v_1, \dots ,v_{2n}\}$ in the root of this special tree decomposition. 
\end{claim}
We prove Claim~\ref{claim:nestedwordstwStrong} by induction. 

\begin{itemize}
	\item The base case is for a maximal $\matching$ path  $p= (v_1, \dots ,v_{2n})$, $n \leq d$, which does not cover another maximal $\matching$ path. In this case, we have a special tree decomposition of width $2n+1$ for $\lograph_p$.  Indeed, $\lograph_p$  can be generated in a left to right manner, keeping all boundary vertices in the bag going forward, and discarding the inner vertices whose neighbours were all added. This special tree decomposition is indeed a special path decomposition, since the tree is just a path. For instance, consider Figure~\ref{fig:nestedword} and consider a maximal $\matching$ path  $p_1= (6,8,11,14)$.    Then $\lograph_{p_1}$ is the induced subgraph of Figure~\ref{fig:nestedword} on the vertices $\{6,7,8,11,12,13,14\}$. This is can be generated by the special path decomposition $\{6,7,8\}$, $\{6,8,11\}$, $\{6,8,11,12\}$, $\{6,8,11,12, 13\}$, $\{6,8,11,13, 14\}$, $\{6,8,11,14\}$. 
	
	\item  Consider a maximal $\matching$ path  $p= (v_1, \dots ,v_{2n})$, $n \leq d$, which contains other $\matching$ paths within. A  maximal $\matching$ path $p_1$ is said to be a child of a  maximal $\matching$ path $p_2$ if 1) $p_1$ is contained in $\lograph_{p_2}$ and 2) there is no other  maximal $\matching$ path  $p_3$ such that $p_1$ is contained in $p_3$ which is in turn contained in $p_2$. 
	For the inductive case, suppose a maximal $\matching$ path  $p= (v_1, \dots ,v_{2n})$ has children $p_1, \dots ,p_m$.  We assume that we have the special tree decompositions $t_i$ for each of these $\lograph_{p_i}$. We add a common parent  node for all the $t_i$s 	and have  in this all the vertices in the path $p$ and those in the paths $p_i$. This node has width at most $2d \times (r+1)$ where $r$ is the number of children and $d$ is the maximum arity. Now we need to add the other nodes that are not covered by any $\lograph_{p_i}$. We follow again a left-to-right pass similar to the base case, and this would increase the bag-size by 2. Finally $\lograph_p$ is fully covered, we add one more parent which contains only the vertices of $p$ as required.
\end{itemize}

Note that the maximum bag size is at most $2d(r+1) +2$ where $2d$ is the an upper bound on the length of a maximal $\matching$ path, and $r$ is an upper bound on the number of children a path can have. Notice that nested words of $d$-\mcfg($r$) satisfy these bounds. Hence their special tree width is at most $2d(r+1) +2$ which is clearly below $6dr$ if $d,r, \ge 1$.

\section{Downward closures of MCFL}
\label{appendix:downward-closures}

\subsection{Concise Overapproximations}
\label[appsec]{app:conciseness-step}
In the main body of the text, we have described how to construct an accelerated grammar $\grammar_\Sigma$ from a base grammar $\grammar$. Let us briefly show that the size of this construction is small and its construction does not take too much time:

\begin{lemma}\label{lem:accel-mcfg-size}
	Given an \mcfg\ $\cG$ we can construct the accelerated \mcfg\ $\cG_\Sigma$ in time $\sizeof{\cG}^{\poly(k)}$, where $k$ is dimension of $\cG$.
\end{lemma}
\begin{proof}
	Let $n$ be the size of $\cG$.
	We already argued above that adding all our decorations leads to a grammar of size $\cO(n^{4^k})$,
	and this construction does not take longer to perform than its size suggests.
	To obtain $\cG_\Sigma$ we now add $\cO(n k^{2k})$ more productions per Production of $\cG$,
	for each of which we need to search for watershed alphabets.
	As mentioned before, the latter process is a simple search in the connectivity graph of $\cG_\Sigma$,
	which takes polynomial time in $n$.
	
	In total, we stay polynomial in $n$ and doubly exponential in $k$ to construct $\cG_\Sigma$.
\end{proof}

Let us now turn towards shortening individual trees:

Assume we have a tree $t$ of an accelerated \mcfg~$\cG_\Sigma$ and a path $p$ in $t$ such that along $p$, we have four $A$-rooted subtrees $t_0$, $t_1$, $t_2$, and $t_3$ that all embed into each other with the same idempotent embedding $f$ (that is, $t_i \embed_f t_{i-1}$). In particular, we can assume that $t_0 = t$. We construct a tree $\tilde t$ such that $\tilde t$ accepts more words than $t$ but is smaller in size.

As the root production for $\tilde t$, we choose the accelerated rule $\pi_{f,i,g} = A(s'_1, \ldots s'_n) \leftarrow S$ obtained from the root production $\pi$ of $t$, where $i$ is the first letter of $p$ and $g$ is the embedding of $t_1 \embed_g \subtree{t}{i}$.

Recall that $\pi_{f,i,g}$ has $m + 1$ children to the $m$ children $\pi$. For $j \neq i$ and $j \neq m + 1$, we obtain $\subtree{\tilde t}{j}$ from $\subtree{t}{j}$ as by projection to $J$, the set of its entries that do not embed into one of the fixpoints.

The $i$-th child is obtained from $\subtree{t}{i}$ by using the search decoration to $t_1$ using the embedding $g$, that is $\subtree{\tilde t}{i} = (\subtree{t}{i})^{q,g}$ where $q$ is chosen such that $\subtree{t}{(i\cdot q)} = t_1$. Note that in particular this means that the subtree in $\tilde t$ corresponding to $t_1$ in $t$ is empty, which we will leverage to show that $\tilde t$ is smaller than $t$.

Finally, we need to construct the $(m+1)$-st child, which will produce the fixpoint entries for the new root production. Here we chose $(t_2)^I$, where $I$ is the set of fixpoints of $f$.

\begin{lemma}
	\label{lem:alphabet-grammar-rule-enlarges}
	$L(t) \subseteq L(\tilde t)$.
\end{lemma}

\begin{proof}
	By the construction of $\pi_{f,i,g}$, for all $i \not\in \Img(f)$ we have $\yield(\tilde t)[i] = \yield(t)[i]$. 
	
	Let therefore $i \in \Img(f)$. We can write $w \in L(\yield(t)[i])$ as $w = u w_i v$, where $w_i \in$ \linebreak $L(\yield(t_2)[i])$. All entries of $t_3$ are fully embedded in one of the fixpoint entries in $t_2$. Therefore the set of leaf nodes that produce $u$ and $v$ can not be rooted below $t_3$. This means that $\Sigma(u) \subseteq \Sigma_{i,\mathrm{left}}$ and $\Sigma(v) \subseteq \Sigma_{i,\mathrm{right}}$. We therefore have $L(\yield(t))[i]) \subseteq L(\yield(\tilde t)[i])$.
\end{proof}

\begin{lemma}
	$\card{\tilde t}  < \card{t}$.
\end{lemma}

\begin{proof}
	First, observe that $\card{\subtree{\tilde t}{j}} \leq \card{\subtree{t_0}{j}}$ for $j \neq i$. Additionally $\card{\subtree{\tilde t}{(m+1)}} \leq \card{t_2}$. Finally, because $t_1 \embed_g \subtree{t}{i}$, the annotation of $\subtree{\tilde{t}}{i}$ will arrive at $t_1$ with a target embedding of $\id$. Having zero entries, $t_1^{\varepsilon,\id}$ is empty. Therefore $\card{\subtree{\tilde{t}}{i}} \leq \card{\subtree{t_0}{i}} - \card{t_1}$. Observing that $\card{t_2} < \card{t_1}$ we get $\card{\tilde{t}} \leq (\card{t} - \card{t_1}) + \card{t_2} < \card{t_0}$.
\end{proof}

\subsection{Upper Bound for Downward Closures without Fixed Dimension}
If we now assume that the dimension $k$ of the grammar is part of the input, then a naïve analysis yields the following:

\begin{enumerate}
	\item There are $k^k$ possible embedding maps.
	\item A single step of the decoration procedure intoduces $\cO(n^2 \times k^k)$ new nonterminals and $\cO(n^2 \times k^k)$ new productions per nonterminal, increasing the grammar size to $\cO(n^4 \times k^{2k})$.
	\item This step may need to be iterated $k$ times until it stabilizes, yielding a grammar size of $(nk)^{2^{\cO(k)}} = n^{2^{\cO(k)}}$.
	\item Adding accelerating productions adds a further $\poly(\sizeof{\cG'}) \times k^{\cO(k)}$ productions, so the grammar size remains $n^{2^{\cO(k)}}$.
	\item By \cite{DBLP:conf/stacs/Jecker21}, the path length to guarantee a monochromatic 4-clique is bounded above by $(k^{4k})^{\cO(k)} = 2^{poly(k)} \times n^{2^{\cO(k)}}$, which is still $n^{2^{\cO(k)}}$.
	\item Trees of this height have size $N^N$ where $N = n^{2^{\cO(k)}}$ and there are $N^{N^N}$ many of them, yielding a downward closure size of $2^{2^{n^{2^{\cO(k)}}}}$
\end{enumerate}

Therefore a trivial upper bound for the space complexity of downward closures of {\mcfg}s without fixed dimension is $\FIVEEXPSPACE$.

\section{Lower bound for MCFL}\label[appsec]{app:MCFL-lowerbound}
	In this section we prove \cref{claim:LBmcfl}. 	Consider the language 
\[ L_{ww,n}=\big\{ww \;\big|\; |w| = 2^n, w \in \{a,b\}^\ast \big\}. \]
 We will give a $\twoicfg$ that is linear in $n$ for it, and show the downward closure of the above language needs a double exponential sized NFA. 

\begin{example}
	Consider the  $\twoicfg$ $\initgrammar_{ww,n} = (\{a,b\}, \hat\Nonterminals, \Variables, \Productions, \StartNonterminal)$ where $\hat\Nonterminals = \{ \nt{A_i}{2} \mid 0 \le i \le n\} \cup \{ 	\nt{\StartNonterminal}{1} \} $, $\Variables = \{u_i, v_i, x_i, y_i \mid 1 \le i < n\} \cup \{u_n, v_n\}$ and  $\Productions$ contains the following rules, where $i \in \{0, 1, \dots n-1\}$ 
	\begin{align}
		\StartNonterminal( u_n v_n ) &\leftarrow \{A_n(u_n,v_n)\} \tag{$\prodrule_n$}\\
		A_{i+1}( u_ix_i,v_iy_i) &\leftarrow \{A_{i}(u_{i},v_{i}), A_{i}(x_i, y_i)\} \label{rule:rule2} \tag{$\prodrule_i$}\\
		A_0(a, a) & \leftarrow \emptyset \tag{$\prodrule_{a}$}\\
		A_0(b, b) & \leftarrow \emptyset \tag{$\prodrule_{b}$}
	\end{align}

	Here, $\Relof{A_i}{} =\big\{(w,w) \;\big|\; |w| = 2^i, w \in \{a,b\}^\ast\big\}$, and $\langof{\initgrammar_{ww,n}} =\Relof{\StartNonterminal}{} = \big\{ww \;\big|\; |w| = 2^n, w \in \{a,b\}^\ast \big\} $. 
	\qed
\end{example}

A fooling-set argument proves the following claim which in turn proves \cref{claim:LBmcfl}.
\begin{lemma}\label{claim:minNFAsizeww}
	Any NFA for $\downc{L_{ww,n}}$ has at least $2^{2^n}$ states. 
\end{lemma}
\begin{proof}
	Let $W=\big\{w \;\big|\; |w| = 2^n, w \in \{a,b\}^\ast \big\}$ and let $u \in W$. Since $uu \in \downc{L_{ww,n}}$ but $vu \notin \downc{L_{ww,n}}$  for any  $v \in W \setminus\{u\}$,  there is a state $q_u$ reached upon reading $u$ which (i)~leads to accepting state on reading $u$ and (ii) is not reached by any other $v \in W$. Hence the map $W\to Q$, $u\mapsto q_u$, is injective, implying $|Q| \geq |W| = 2^{2^n}$, proving the claim. 
\end{proof}

\section{Lower bound for bounded-\streewidth\  MPDA}\label[appsec]{app:mpds-lowerbound} \label[appsec]{app:mpda-lower-bound}

\newcommand{\rev}{\mathsf{rev}}
In this section, we prove \cref{claim:LBMPDS}. For a word $w\in\Sigma^*$, let $w^\rev$ denote the reverse of $w$.
Consider the language 
\[ L_{ww^\rev,n}=\big\{ww^\rev \;\big|\; |w| = 2^n, w \in \{a,b\}^\ast\big\}. \]
 We will give a 2-stack MPDA of size linear in $n$ for this language. By the same argument as in \cref{app:MCFL-lowerbound}, any NFA for the downward closure of $L_{ww^\rev,n}$ needs at least $2^{2^{n}}$ states.

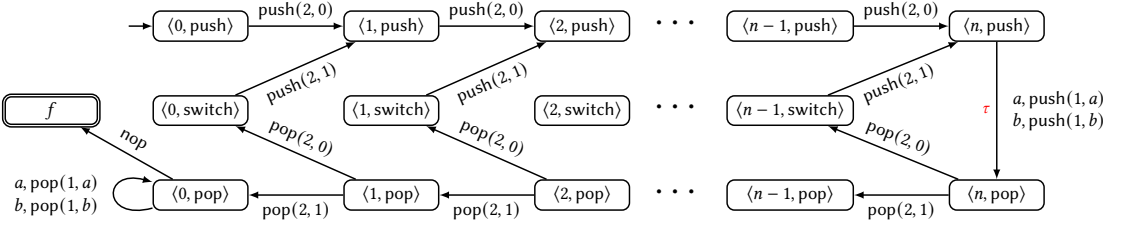
\begin{figure}
	\scalebox{0.7}{
		\begin{tikzpicture}[node distance=1.8cm, initial text=,>=latex, thick, minimum width=1.8cm]
			\node[rectangle, draw, rounded corners, initial ] (t0) {$\statepush{0}$};
			\node[rectangle, draw, rounded corners, right=of t0 ] (t1) {$\statepush{1}$};
			\node[rectangle, draw, rounded corners, right=of t1 ] (t2) {$\statepush{2}$};	
			\node[rectangle, rounded corners, node distance = 0cm, right=of t2,  ] (t3) {\scalebox{2}{$\cdots$}};
			\node[rectangle, draw, rounded corners,node distance = 0cm,  right=of t3 , minimum width=2.4cm ] (t4) {$\statepush{n-1}$};	
			\node[rectangle, draw, rounded corners, right=of t4 ] (t5) {$\statepush{n}$};	
			\node[rectangle, draw, rounded corners, node distance = 1cm,  below=of t0 ] (m0) {$\stateswitch{0}$};
			\node[accepting, rectangle, draw, rounded corners, node distance = 1cm,  left=of m0 ] (f) {$f$};
			\node[rectangle, draw, rounded corners, right=of m0 ] (m1) {$\stateswitch{1}$};
			\node[rectangle, draw, rounded corners, right=of m1 ] (m2) {$\stateswitch{2}$};	
			\node[rectangle, rounded corners, node distance = 0cm, right=of m2, ] (m3) {\scalebox{2}{$\cdots$}};
			\node[rectangle, draw, rounded corners, node distance = 0cm, right=of m3, minimum width=2.4cm ] (m4) {$\stateswitch{n-1}$};
			\node[rectangle, draw, rounded corners, node distance = 1cm,  below=of m0,  ] (b0) {$\statepop{0}$};
			\node[rectangle, draw, rounded corners, right=of b0 ] (b1) {$\statepop{1}$};
			\node[rectangle, draw, rounded corners, right=of b1 ] (b2) {$\statepop{2}$};	
			\node[rectangle, rounded corners, node distance = 0cm, right=of b2, ] (b3) {\scalebox{2}{$\cdots$}};
			\node[rectangle, draw, rounded corners, node distance = 0cm, right=of b3, minimum width=2.4cm  ] (b4) {$\statepop{n-1}$};	
			\node[rectangle, draw, rounded corners, right=of b4 ] (b5) {$\statepop{n}$};	
			\draw[->] 
			(t0) edge node[above] {$\mpdspush{2}{0}$} (t1)
			(t1) edge node[above] {$\mpdspush{2}{0}$} (t2)
			(t4) edge node[above] {$\mpdspush{2}{0}$} (t5)
			(m0) edge[] node[below, sloped] {$\mpdspush{2}{1}$} (t1)
			(m4) edge node[below, sloped] {$\mpdspush{2}{1}$} (t5)
			(m1) edge node[below, sloped] {$\mpdspush{2}{1}$} (t2)
			(b2) edge node[above, sloped] {$\mpdspop{2}{0}$} (m1)
			(b1) edge node[above, sloped] {$\mpdspop{2}{0}$} (m0)
			(b0) edge node[above , sloped] {$\noop$} (f)
			(b5) edge node[below] {$\mpdspop{2}{1}$} (b4)
			(b2) edge node[below] {$\mpdspop{2}{1}$} (b1)
			(b1) edge node[below] {$\mpdspop{2}{1}$} (b0)
			(b5) edge node[above, sloped] {$\mpdspop{2}{0}$} (m4)
			(t5) edge[]  node[ minimum width = 0cm, right]{$\begin{array}{l}
					a, \mpdspush{1}{a}\\
					b, \mpdspush{1}{b}\\
				\end{array}$}  node[left,minimum width = 0cm] {\tname{ }} (b5)
			(b0) edge[loop left, min distance = 10mm, in = 165, out = 195]  node[left] {$\begin{array}{l}
					a, \mpdspop{1}{a}\\
					b, \mpdspop{1}{b}\\
				\end{array}$} (b0)
			;
		\end{tikzpicture}
	}
	\caption{A 2-stack MPDA with $3n+3$ states that accepts  $\{ww^\rev \mid w \in (a+b)^{2^n} \}$. We omit $\epsilon$ from a transition label for readability. Here, all the transitions with operations on stack $2$ are on $\epsilon$. }\label{fig:wwrk}
\end{figure}

	Consider the MPDA given in Figure~\ref{fig:wwrk}. The first time it reaches the state $\statepush{n}$ the stack 2 contains $0^n$. In each subsequent visit to the state $\statepush{n}$, the MPDA will increment the $n$-bit number in stack $2$. For each of these $n$-bit numbers it will take the transition $\tau{ }$ reading a letter and pushing it to stack 1.  Once the $n$-bit number becomes $1^n$,  it takes the transition $\tau{ }$ for the last time and  then goes through $n$ transitions, each labelled $\mpdspop{2}{1}$, and reaches state $\statepop{0}$. Now it would have read a word $w$ of length $2^n$ which is now stored in Stack 1 in reverse. Furthermore, stack 2 is now empty. Finally it reads $w^\rev$ and  empties stack 1, then moves to accepting state $f$.

It remains to argue that all runs of this MPDA have bounded special treewidth.
Since in the rungraph of an MPDA, every vertex has degree at most $3$, and
graphs of degree $d$ and treewidth $t$ have special treewidth at most
$20dt$~\cite{Cou10}, it suffices to show that the treewidth of these MPDA is at
most $4$. Observe that every run of the MPDA can be decomposed into two
\emph{phases}: A phase consists of a sequence of transitions where only one
stack can be popped, but all stacks can be pushed to. Indeed, during the
counting from $0^n$ to $1^n$, the MPDA never pops stack 1. Then, after emptying
stack 2, it switches into a second phase, where it pops the content of stack 1
(so as to compare the second part of the input word to $w^\rev$). It was shown
by Madhusudan and Parlato~\cite{TreeAuxMadhu} that MPDA with at most $k$ phases
have in fact rungraphs with treewidth at most $2^k$. Thus, since our MPDA has
at most $2$ phases in each run, it has treewidth at most $4$.

For \cref{claim:LBMPDS}, it remains to show the following. This follows by the
same argument as in \cref{claim:minNFAsizeww}.
\begin{lemma}\label{lemma:minNFAsizewwr}
	Any NFA for $\downc{L_{ww^\rev,n}}$ has at least $2^{2^n}$ states. 
\end{lemma}

\section{Concurrent Programs and Bounded Context Switching}\label[appsec]{app:programs}

\newcommand{\cT}{\mathcal{T}}

\subsection{Formal Definitions} \label[appsec]{app:programs-formal-def}

\myparagraph{More on Multisets}
Let us first recall some additional definitions related to multisets.
We will need this for the formal definition of concurrent programs.

First we need a way to denote a multiset with just a few elements.
As an example, we write
$\mmap=\multi{a,a,c\,}$ for the multiset
$\mmap\in\multiset{\{a,b,c,d\}}$ such that
$\mmap(a)=2$,
$\mmap(b)=\mmap(d)=0$, and $\mmap(c)=1$.
Next we would like some multiset operations akin to addition and subtraction.
Given two multisets $\mmap,\mmap'\in\multiset{X}$ let
$\mmap\oplus\mmap'\in\multiset{X}$ denote the multiset with
$(\mmap\oplus \mmap')(a)=\mmap(a)+\mmap'(a)$ for all $a\in X$.
If $\mmap(a) \geq \mmap'(a)$ for all $a\in X$, we also define
$\mmap' \ominus \mmap \in\multiset{X}$:
for all $a\in X$, we have $(\mmap\ominus \mmap')(a)=\mmap(a)-\mmap'(a)$.
For $X \subseteq Y$ we identify $\mmap\in\multiset{X}$ with a
multiset in $\multiset{Y}$ where undefined values are mapped to $0$.

\myparagraph{Concurrent Programs}
We recall the definition of concurrent programs over an arbitrary class of languages $\cC$, as given in~\cite{DBLP:conf/icalp/0001GMTZ23}.
To be more precise, $\cC$ is required to be a language class exhibiting the properties of a so-called
``full trio'', but it is known that the MCFLs fall under this category \cite{SekiMFK91}.
As such, in the sequel we will only consider the MCFLs for the class $\cC$.

Let $\cC$ be a full trio. 
A \emph{concurrent program} ($\CP$) over $\cC$ is a 
tuple $\ap=(D,\Sigma,(L_a)_{a\in\Sigma},d_0,\mmap_0)$, where 
$D$ is a finite set of \emph{global states},  $\Sigma$ is an alphabet 
of \emph{process names},
the \emph{process languages}
$(L_a)_{a\in\Sigma}$ are a family of languages from $\cC$, each
over the alphabet $\Sigma_D=D \cup \Sigma \cup (D \times D)$,
$d_0\in D$ is an \emph{initial state}, and $\mmap_0\in \multiset{\Sigma}$ is a multiset 
of \emph{initial process instances}. 
We assume that each  $L_a$, $a\in \Sigma$, satisfies 
the condition $L_a \subseteq aD \big(\Sigma \cup (D \times D)\big)^* (D \cup D \times D)$.
This is slightly different from the definition in~\cite{DBLP:conf/icalp/0001GMTZ23},
but as we see below, this is not a meaningful change, and it helps for our proofs
in \Cref{sec:applications}.

A \emph{configuration} $c=(d,\mmap) \in D \times \multiset{\Sigma_D^*}$ consists of a global state $d\in D$ and a multiset $\mmap$ of words, each representing a process instance. 
Given a configuration $c=(d,\mmap)$, we write $c.d$ and $c.\mmap$ to denote the elements $d$ and $\mmap$, respectively. 
The size of a configuration $c$ is $|c.\mmap|$, i.e.\ the number of current process instances. 
We distinguish between processes that have been spawned but not executed (\emph{pending processes}) and processes that have already been partially executed (and may be resumed later).
The pending process instances are represented by single letters $a \in \Sigma$ (which corresponds to the name of the process),
while the partially executed processes corresponding to a name $a\in \Sigma$
are represented by words in the prefix language $\pref(L_a)$ which end in a letter from $D \times D$.

Formally, the semantics of a $\CP$ $\ap$ are defined as a labelled transition system over the set of configurations with the transition relation 
$\Rightarrow \subseteq (D \times \multiset{\Sigma_D^*}) \times (D \times \multiset{\Sigma_D^*})$.
The initial configuration is $c_0=(d_0,\mmap_0)$. 

The transition relation is defined using four different types of rules, as shown below. 
All four types are of the general form $d \xrightarrow{\multi{w},\nmap'}d'$.
A rule of this form allows the program to transition
from a configuration $(d,\mmap)$ to a configuration $(d',\mmap')$, 
written as $(d,\mmap) \Rightarrow (d',\mmap')$, iff 
$d\xrightarrow{\multi{w}.\nmap'} d'$ matches a rule and $(\mmap \ominus \multi{w}) \oplus \nmap'=\mmap'$.
Note that by the definition of $\ominus$, $\mmap$ has to contain $w$ for such a rule to be applicable.
We also write $\xRightarrow{w}$ to specify the particular $w$ used in the transition. 
As usual, $\Rightarrow^*$ denotes the reflexive transitive closure of $\Rightarrow$. 
A configuration $c$ is said to be \emph{reachable} if $c_0 \Rightarrow^* c$.

\begin{description}
  \item[(R1)] $d \xrightarrow{\multi{a},\parikh(w) \oplus \multi{adw(d',d'')}} d' \quad \text{ if }\quad
  \exists w \in \Sigma^* \colon ad w (d',d'') \in \pref(L_a)$.
\end{description}
Rule (R1) allows us to pick some pending process $a$ from $\mmap$ and atomically execute it until its next context switch $(d',d'')$. 
Note that the final letter $(d',d'')$ indicates that this process transitions the global state to $d'$ and can be resumed later when the global state is $d''$ (unless the prefix $ad w (d',d'')$ cannot be extended further). 
\begin{description}
  \item[(R2)] $d \xrightarrow{\multi{a},\parikh(w)} d'
  \quad \text{ if }\quad
  \exists w \in \Sigma^* \colon ad w d' \in L_a$.
\end{description}
Rule (R2) allows us to pick a pending process $a$ from $\mmap$ and atomically execute it to completion.
\begin{description}
  \item[(R3)]$d \xrightarrow{\multi{aw'(d'',d)},\parikh(w)\oplus \multi{aw'(d'',d)w(d',d''')}} d'  	\quad \text{ if }\quad
  \begin{aligned} &\exists w \in \Sigma^* \colon\\ &\quad aw'(d'',d) w (d',d''') \in \pref(L_a)\end{aligned}$
\end{description}
Rule (R3) allows us to pick a partially executed process and execute it atomically until its next context switch.
\begin{description}
  \item[(R4)]$d \xrightarrow{\multi{aw'(d'',d)},\parikh(w)} d' \quad \text{ if }\quad
  \exists  w \in \Sigma^* \colon aw'(d'',d) w d' \in L_a$
\end{description}
Rule (R4) allows us to pick some partially executed process and execute it to completion.

Note that our change from \cite{DBLP:conf/icalp/0001GMTZ23}, allowing words from $L_a$ to end in
$D \times D$ in addition to $D$, does not meaningfully change the semantics here.
The only thing that might happen is a process may stick around in the multiset, even though its execution cannot be continued.
This is because the only ways to remove processes are rules (R2) and (R4), which require words
ending in $D$.
However, a process whose prefix cannot be extended does not contribute to further transitions,
so in particular it does not change global state reachability in any way.

\myparagraph{Bounded context switching}
Most decision problems are undecidable for concurrent programs,
if we impose no restrictions on their behavior \cite{QadeerR05,ramalingam2000context}.
As such we consider the popular restriction of bounded context switching.
To properly introduce this restriction, we need to define a few notions
related to runs of concurrent programs.

A \textit{prerun} of a concurrent program $\ap=(D,\Sigma,(L_a)_
{a \in \Sigma},d_0,\mmap_0)$ is a finite sequence
$\rho = (e_0,\nmap_0),$ $w_1,$ $(e_1,\nmap_1),$ $w_2,$ $\ldots,$ $(e_\ell,\nmap_\ell)$
of alternating elements of configurations
$(e_i,\nmap'_i) \in D \times \multiset{\Sigma_D^*}$
and words $w_i \in \Sigma^*$. 

We use $\Preruns{\ap}$ to denote The set of preruns of $\ap$.
Note that if two concurrent programs $\ap$ and $\ap'$ have
the same global states $D$ and process names $\Sigma$, then $\Preruns{\ap}=\Preruns{\ap'}$.
The length $|\rho|$ of a prerun $\rho$ is the number of configurations in $\rho$.

A \textit{run} of a $\CP$ $\ap=(D,\Sigma,(L_a)_{a\in\Sigma},d_0,\mmap_0)$ is a prerun
$\rho = (d_0,\mmap_0), w_1, (d_1,\mmap_1), w_2, \ldots, (d_\ell,\mmap_\ell)$
starting with the initial configuration $(d_0,\mmap_0)$, such that for each $i \geq 0$, we have 
$(d_i,\mmap_i) \xRightarrow{w_{i+1}} (d_{i+1},\mmap_{i+1})$. 
We use $\Runs{\ap}$ to denote The set of runs of $\ap$.

For a number $k$, the run $\rho$ is said to be $k$-context-bounded ($k$-CB for short) 
if for each transition $c \xRightarrow{w} c'$ in $\rho$ we have $|w|_{D\times D} \leq k$.
This definition of is slightly changed from \cite{DBLP:conf/icalp/0001GMTZ23} as not to count
symbols from $D \times D$ at the end of a process' execution, if the process is not continued.
The set of $k$-context-bounded runs of $\ap$ is denoted by $\Runsk{\ap}$. 
We say a configuration $c$ is \emph{$k$-reachable} if there is a finite
$k$-context-bounded run $\rho$ ending in $c$.

\subsection{Closure Properties of \mcfg{}s}\label[appsec]{app:mcfg-closure}

\newcommand{\vecq}{\vec{\mathbf{q}}}
\newcommand{\vecp}{\vec{\mathbf{p}}}

For an \mcfg, let its \emph{dimension} $d$ be the maximum arity among its nonterminals,
and its \emph{rank} $r$ be the maximum number of nonterminals appearing on the right-hand side
of any of its productions.

\myparagraph{Intersecting \mcfg{}s and NFAs}
We prove the following:

\begin{theorem}\label{thm:mcfg-nfa-intersection}
  Given an \mcfg\ $\grammar$ of dimension $d$ and rank $r$, as well as an NFA $\cA$,
  one can compute a \mcfg\ $\grammar_\cap$ for the language
  $\langof{\grammar} \cap \langof{\cA}$ in time
  $(\sizeof{\grammar} \cdot \sizeof{\cA})^{\poly(d \cdot r)}$.
\end{theorem}

If $d$ and $r$ are fixed, we therefore obtain:

\begin{corollary}\label{cor:mfcg-nfa-intersection}
  Given an \mcfg\ $\grammar$ of constant dimension and rank, as well as an NFA $\cA$,
  one can compute a \mcfg\ for the language $\langof{\grammar} \cap \langof{\cA}$
  in polynomial time.
\end{corollary}

This is part~\ref{item:mcfg-nfa-intersect} of \Cref{lem:mcfg-closure-poly}.

The proof of \Cref{thm:mcfg-nfa-intersection}
is an adaptation of the well-known \emph{triple construction} for intersecting CFGs and NFAs.
Therein for every nonterminal $A$ of the grammar and every pair of states $p,q$ of the automaton,
one introduces a new nonterminal $A_{p,q}$ that derives only those words that are derivable by $A$, but also induce a transition sequence from $p$ to $q$.

To adapt this to the setting of \mcfg{}s, we need to consider a state pair for every variable
of a nonterminal $A$, i.e.\ arity of $A$ many state pairs.
Let $\grammar = (\Sigma, \hat\Nonterminals, \Variables, \Productions, A_\text{start})$ be
an \mcfg, 
and let $\cA$ be an NFA over $\Sigma^*$ with set of states $Q$, initial state $q_0$,
and final states $F$.

We construct the
\mcfg\ 
$\grammar_\cap = (\Sigma, \hat\Nonterminals_\cap, \Variables,
\Productions_\cap, A_{\text{start}\cap})$ as follows:
\begin{itemize}
  \item $A_{\vecq}^{(\ell)} \in \hat\Nonterminals_\cap$ iff $A^{(\ell)} \in \hat\Nonterminals$,
  and $\vecq = \big((q_1,q'_1),(q_2,q'_2),\ldots(q_\ell,q'_\ell)\big) \in (Q \times Q)^\ell$.
  \item $\pi_\cap = A_{\vecq}(s_1, \ldots, s_\ell) \leftarrow S_\cap \in \Productions_\cap$
  with $S_\cap = \setof{B_{1,\vecp_1}(x_{1,1}, \ldots, x_{1,\ell_1}), \ldots}$ iff
  the following all hold:
  \begin{itemize}
    \item $\pi = A(s_1, \ldots, s_\ell) \leftarrow S \in \Productions$
    with $S = \setof{B_1(x_{1,1}, \ldots, x_{1,\ell_1}), \ldots}$.
    \item $\vecq = \big((q_1,q'_1),(q_2,q'_2),\ldots(q_\ell,q'_\ell)\big)$,
    $\vecp_i = \big((p_{i,1},p'_{i,1}),(p_{i,2},p'_{i,2}),\ldots(p_{i,\ell_i},p'_{i,\ell_i})\big)$
    for all $i$, and we say $(p_{i,j},p'_{i,j})$ is the \emph{state pair} for $x_{i,j}$.
    
    Further, for all $i'$, $s_{i'} = w_{0,i'} y_{1,i'} w_{1,i'} \cdots y_{n_{i'},i'} w_{n_{i'},i'}$,
    where each $w_{j',i'}$ is a (possibly empty) word
    from $\Sigma^*$,
    and each $y_{j',i'}$ is a variable from $\Variables$.
    
    Then $q_{i'} \xrightarrow{w_{0,i'}} \tilde{q}_{1,i'}$,
    $\tilde{q}'_{j'-1,i'} \xrightarrow{w_{j'-1,i'}} \tilde{q}_{j',i'}$,
    and $\tilde{q}'_{n_{i'},i'} \xrightarrow{w_{n_{i'},i'}} q'_{i'}$ in $\cA$ for each $i',j'$,
    where $(\tilde{q}_{j',i'},\tilde{q}'_{j',i'})$ is the state pair for $y_{j',i'}$.
  \end{itemize}
  \item $A_{\text{start}\cap}^{(1)}(x) \leftarrow \{A_{\text{start},(q,q')}(x)\}
  \in \Productions_\cap$ iff $q = q_0$ and $q' \in F$.
\end{itemize}
By the notation $p \xrightarrow{w} q$ for a (possibly empty) word
$w \in \Sigma^*$
we mean that
there is a transition sequence from $p$ to $q$ over
$w$
in $\cA$.
If $w$ is empty, then we just consider $\varepsilon$-transitions.

Intuitively, an index $\vecq$ for a nonterminal $A_{\vecq}$ gives a state pair for
each member of a tuple derived by $A_{\vecq}$, that needs to be realized as a transition sequence
in $\cA$ by that member.
The productions of $\grammar_\cap$ are constructed in such a way that state pairs on the right-hand
side concatenate properly to ones on the left-hand side.

Let us now prove the theorem.

\begin{proof}[Proof of \cref{thm:mcfg-nfa-intersection}]
  We construct $\grammar_\cap$ as above.
  Towards the theorem, we prove the following statement by induction:
  A tuple $(v_1,\ldots,v_\ell) \in (\Sigma^*)^\ell$ can be derived from
  $A^{(\ell)}_{\vecq}$ in $\grammar_\cap$ iff
  $(v_1,\ldots,v_\ell)$ can be derived from $A^{(\ell)}$ in $\grammar$, and
  for $\vecq = \big((q_1,q'_1),(q_2,q'_2),\ldots(q_\ell,q'_\ell)\big)$ we have
  $q_i \xrightarrow{v_i} q_i'$ in $\cA$ for all $i$.
  Then by definition of the productions of $A_{\text{start}\cap}$, a terminal word is derivable
  by $\grammar_\cap$ iff it is derivable by $\grammar$, and induces a run on $\cA$ from the initial
  state to some final state, thus proving the construction correct.
  \begin{itemize}
    \item Base case, $A_{\vecq}(s_1, \ldots, s_\ell) \leftarrow \emptyset \in \Productions_\cap$.
    Then each $s_i$ is a word over atoms $a$ and $\Gamma^*$, meaning we require
    $q_i \xrightarrow{s_i} q_i'$ by definition.
    Furthermore, $A(s_1, \ldots, s_\ell) \leftarrow \emptyset \in \Productions$ is also
    guaranteed by definition.
    
    For the other direction, we assume $A(s_1, \ldots, s_\ell) \leftarrow \emptyset \in \Productions$
    and $q_i \xrightarrow{s_i} q_i'$ for each $i$.
    Then, again, $A_{\vecq}(s_1, \ldots, s_\ell) \leftarrow \emptyset \in \Productions_\cap$
    follows by definition.
    
    \item Inductive case,
    $\pi_\cap = A_{\vecq}(s_1, \ldots, s_\ell) \leftarrow S_\cap \in \Productions_\cap$
    with $S_\cap = \setof{B_{1,\vecp_1}(x_{1,1}, \ldots, x_{1,\ell_1}), \ldots}$,
    where for the words derivable from each $B_{i,\vecp_i}$ we assume the statement already holds,
    as long as the derivation is shorter.
    By definition we have a production $\pi = A(s_1, \ldots, s_\ell) \leftarrow S \in \Productions$
    with $S = \setof{B_1(x_{1,1}, \ldots, x_{1,\ell_1}), \ldots}$,
    and it is easy to see how an analogous derivation is now possible in $\grammar$.
    Regarding the validity of $\vecq$ and $\vecp_i$, we also refer to the definition
    of $\Productions_\cap$:
    For $s_{i'} = w_{0,i'} y_{1,i'} w_{1,i'} \cdots y_{n_{i'},i'} w_{n_{i'},i'}$, we have
    $q_{i'} \xrightarrow{w_{0,i'}} \tilde{q}_{1,i'}$,
    $\tilde{q}'_{j'-1,i'} \xrightarrow{w_{j'-1,i'}} \tilde{q}_{j',i'}$,
    and $\tilde{q}'_{n_{i'},i'} \xrightarrow{w_{n_{i'},i'}} q'_{i'}$
    by definition.
    Furthermore, for each $y_{j',i'}$ by induction we have a word $v_{j',i'}$
    witnessing $\tilde{q}_{j',i'} \xrightarrow{v_{j',i'}} \tilde{q}'_{j',i'}$.
    Stitching all these transition sequences together yields a run from $q_{i'}$ to $q'_{i'}$
    for $s_{i'}$ and each $i'$ as required.
    
    The other direction is similar.
  \end{itemize}
  
  It remains to argue about the time it takes to construct $\grammar_\cap$.
  For each nonterminal $A^{(\ell)}$ of $\grammar$, we need to add
  $|Q|^{2\ell}$ many nonterminals to $\grammar_\cap$, which takes time exponential
  in $d$, since $d$ is the upper bound for $\ell$, but otherwise requires time polynomial
  in $\sizeof{\grammar}$ and $\sizeof{\cA}$.
  
  Similarly for each Production $\pi = A(s_1, \ldots, s_\ell) \leftarrow S$ of $\grammar$,
  we may need to add up to $|Q|^{2\ell} \cdot |Q|^{2d \cdot |S|}$ many productions to
  $\grammar_\cap$, depending on whether such a production could lead to a valid path in $\cA$
  or not.
  Basically, in the worst case, we need to go over all possibilities of indices $\vecq$ and $\vecp_i$
  on both sides of the production, of which there are $|Q|^{d \cdot (r+1)}$ many in the worst case,
  since $r$ is the bound for $|S|$.
  This is of course exponential in $d$ and $r$, but otherwise polynomial in
  $\sizeof{\grammar}$ and $\sizeof{\cA}$.
  
  Finally, for each production we add, it may be necessary multiple times to check existence
  of a path $p \xrightarrow{w} q$ in $\cA$ for some terminal word $w \in \Sigma^*$ appearing
  as an infix in a production.
  However, this check is clearly polynomial in $|w|$ and $|\cA|$, and $|w|$
  is clearly dependant on the length of said production.
  So this just adds a polynomial factor to the already described runtime.
\end{proof}

\myparagraph{Prefix language of an \mcfg\ using rational transductions}
We prove the following statement:

\begin{theorem}\label{thm:mcfg-transduction}
  Given an \mcfg\ $\grammar$ of constant dimension and rank, as well as a transducer $\cT$,
  one can compute an \mcfg\ for the language $\cT(\langof{\grammar})$ in polynomial time.
\end{theorem}

Let us first explain what a transducer is and what the notation $\cT(-)$ means.
For an alphabet $\Sigma$ we define $\Sigma_\varepsilon = \Sigma \cup \{\varepsilon\}$.
A \emph{transducer} is basically an NFA, where rather than one label from $\Sigma_\varepsilon$,
each transition carries two such labels, one for the input and one for the output.
An run of a transducer reaching a final state thus accepts a pair of words
$(u,v) \in \Sigma^* \times \Sigma^*$, where $u$ is the concatenation of all input labels,
and $v$ is the concatenation of all the output labels.
Thus, a transducer computes a relation $\subseteq \Sigma^* \times \Sigma^*$,
called a \emph{rational transduction}.

The notation $\cT(L)$ for a transducer $\cT$ and a language $L$ refers to the image of $L$
under the rational transduction computed by $\cT$.
More formally, $v \in \cT(L)$ iff there is a $u \in L$ such that $(u,v)$ is accepted by $\cT$.
Note that the notion of ``full trio'' alluded to in \cref{app:programs-formal-def}
refers to exactly the language classes that are closed under rational transductions
(see e.g.\ \cite{DBLP:conf/icalp/0001GMTZ23}).
It is known that this closure property holds for MCFL \cite{SekiMFK91},
but we also need a specific time complexity.

With transducers explained, it is clear that \Cref{thm:mcfg-transduction} implies the following:

\begin{corollary}\label{cor:mcfg-prefix-lang}
  Given an \mcfg\ $\grammar$ of constant dimension and rank,
  we can compute an \mcfg\ for its prefix language $\pref{\langof{\grammar}}$ in polynomial time.
\end{corollary}

This is part \ref{item:mcfg-prefix-lang} of \Cref{lem:mcfg-closure-poly}.

\begin{proof}
  The following $2$-state transducer maps a language $L \subseteq \Sigma^*$
  to the set of all prefixes of its words.
  For each $a \in \Sigma$ there is loop labelled by $(a,a)$ on the initial state $q_0$,
  and a loop labelled by $(a,\varepsilon)$ on the only final state $q_f$.
  Then there is one more transition $q_0 \xrightarrow{(\varepsilon,\varepsilon)} q_f$.
  The transducer outputs each letter in the input word, until a nondeterministically chosen
  point, after which all letters are dropped.
  This is clearly the same as computing any prefix.
  The statement then follows from \Cref{thm:mcfg-transduction}.
\end{proof}

Let us now describe how to prove \Cref{thm:mcfg-transduction}.
We write $\shuffle$ to denote the shuffle operator.
Formally, given two languages $L_1,L_2 \subseteq \Sigma^*$,
we have $w = u_0 v_1 u_1\cdots v_n u_n \in L_1 \shuffle L_2$ iff
$u = u_0 \cdots u_n \in L_1$ and $v = v_1 \cdots v_n \in L_2$,
where $u_0,\ldots,u_n, v_1,\ldots,v_n$ are (possibly empty) infixes in $\Sigma^*$.
Now, given a transducer $\cT$ and a language $L$, we argue
$\cT(L) = h\big((L \shuffle \Delta^*) \cap R \big)$,
where $\Delta$ is the set of transitions of $\cT$, $h$ is a homomorphism
mapping each $\delta \in \Delta$ to its output and removing all letters $\notin \Delta$,
and $R$ is a regular language that checks that each $\delta$ is preceded by its input and
that the letters from $\Delta$ form an accepting run of $\cT$.
Then we compute the intersection with $R$ using \Cref{cor:mfcg-nfa-intersection},
whereas applying $h$ and shuffling with $\Delta^*$ are simple constructions.

Note that this way of writing $\cT(L)$ is an adaptation of Nivat's theorem
(see e.g. \cite[Chapter III, Theorem~3.2]{berstel2013transductions}).

\begin{proof}[Proof of \Cref{thm:mcfg-transduction}]
  Let $\grammar = (\Sigma, \hat\Nonterminals, \Variables, \Productions, A_\text{start})$ be
  an \mcfg\ of constant dimension and rank,
  and let $\cT$ be a transducer over $\Sigma$ with set of transitions $\Delta$.
  
  We define the homomorphism $h\colon \Sigma \cup \Delta \to \Sigma$
  as $h\colon a \mapsto \varepsilon$ for each $a \in \Sigma$
  and $h\colon \delta \mapsto \mathsf{output}(\delta)$ for each $\delta \in \Delta$.
  Furthermore, we construct the NFA $\cA$ from $\cT$ as follows:
  For each transition $\delta$ with input $\sigma \in \Sigma_\varepsilon$ in $\cT$
  we replace its label by the concatenation $\sigma \delta$, splitting it into two
  transitions and adding a fresh state in-between, if necessary.
  The resulting NFA clearly accepts all words over $\Sigma \cup \Delta$,
  where each transition in $\Delta$ is preceded by its input, and the maximal subword
  in $\Delta^*$ forms an accepting run of $\cT$.
  
  Let $L = \langof{\grammar}$.
  Now we argue how to compute an \mcfg\ for $h\big((L \shuffle \Delta^*) \cap R \big)$
  in polynomial time.
  For $L \shuffle \Delta^*$, we transform $\grammar$ as follows.
  First, we introduce a new nonterminal $C^{(1)}$, and add each $\delta \in \Delta$ as a terminal.
  Then, for every production $\pi$, we insert a fresh variable before and after every variable and terminal symbol appearing in the left side of $\pi$, resulting in say $\ell$ new variables.
  Then we add $\ell$ instances of $C$ to the right side of $\pi$, each pertaining to a different one
  of the new variables.
  Finally we add for each $\delta \in \Delta$
  the production $C(x\delta) \leftarrow \{C(x)\}$ for some $x \in \Variables$,
  and we also add the production $C(\varepsilon) \leftarrow \emptyset$.
  Clearly the resulting \mcfg\ derives words from $L$ with arbitrary many elements of $\Delta$
  shuffled between their letters.
  We at most triple the length of each production, and add $|\Delta| + 1$ new productions of fixed length, so the whole process takes polynomial time.
  
  We continue by intersecting the resulting \mcfg\ with the NFA $\cA$.
  This takes polynomial time via \Cref{cor:mfcg-nfa-intersection}, and results in another \mcfg.
  Finally, applying $h(-)$ to the latter grammar is easy.
  We simply replace every terminal symbol in every production of the grammar according to $h$,
  which clearly takes polynomial time.
  
  It remains to argue that the final resulting \mcfg\ indeed derives the language $\cT(L)$.
  To see this, consider what our constructions do on the side of formal languages.
  We start with $L$ and shuffle an arbitrary transition sequence of $\cT$ into each of its words.
  Then we check that each transition sequence is a valid accepting run, each transition is preceeded by its input, and each symbol in $\Sigma$ is succeeded by a transition that it inputs into.
  Finally, we remove all input symbols and replace all transitions by their output symbols.
  In sequence this is exactly applying $\cT(-)$.
\end{proof}

\subsection{Partially Permuting Downward Closure}\label[appsec]{app:parikh-dcl}

Let us recall the definition of our partially permuting downward closure, and its accompanying lemma.
We will then prove the latter.

For a language $L \subseteq \Gamma^*$, a subalphabet $\Theta \subseteq \Gamma$, and a number $k$
let $\alphdownck{L}{\Theta}$ denote its alphabet-preserving downward closure,
where subwords are not allowed to drop letters from $\Theta$,
and have to contain exactly $k$ of them.
This is slightly different from the definition in \cite{DBLP:journals/pacmpl/BaumannMTZ22},
where words with fewer than $k$ letters are also included,
but our version fits our purpose slightly better.
Formally, $\alphdownck{L}{\Theta} :=
\{u \in \Gamma* \mid \exists v \in L \colon u \subword v \wedge |u|_\Theta = |v|_\Theta = k\}$.
Here, $|w|_\Theta$ denotes the number of occurrences of letters from the alphabet $\Theta$
in the word $w$.

Furthermore, we define the partially permuting downward closure:
Let $\parikhalphdownck{L}{\Theta}$ denote the variant of $\alphdownck{L}{\Theta}$,
where every maximal infix in $(\Gamma\setminus\Theta)^*$
only needs to be preserved up to its Parikh-image~$\parikh(-)$.
Formally $u = u_0 \theta_1 u_1 \cdots \theta_k u_k \in \parikhalphdownck{L}{\Theta}$
if and only if $u_0,\ldots,u_k \in (\Gamma\setminus\Theta)^*$,
$\theta_1,\ldots,\theta_k \in \Theta$, and
there is $v = v_0 \theta_1 v_1 \cdots \theta_k v_k \in \alphdownck{L}{\Theta}$
such that $\parikh(u_i) = \parikh(v_i)$ for every $0 \leq i \leq k$.

\parikhdcl*

\begin{proof}%
  Let $\grammar= (\Gamma, \hat\Nonterminals, \Variables, \Productions, \StartNonterminal)$
  be an accelerated \mcfg, let $\Theta \subseteq \Gamma$ be a subalphabet,
  and let $\Sigma = \Gamma \setminus \Theta$.
  We construct an NFA $\cA$ for $\parikhalphdownck{\langof{\grammar}}{\Theta}$
  in the following manner.
  Each state of $\cA$ encodes a path in a derivation tree of $\cG$, storing at each step the production used at that node and the $\yield$ of all left siblings.
  Starting with the empty branch,
  $\cA$ guesses a production for $\StartNonterminal$, then builds a branch
  according to left-first depth-first search, until a leaf is encountered.
  At a leaf, or whenever a nonterminal $A$ is assigned a full tuple containing no more variables, $\cA$ moves onto its parent node and replaces the variables of $A$ in the production at said parent
  according to this assignment, subject to the following modifications:
  Instead of a letter $a \in \Sigma$, its Parikh image $\parikh(a)$ is stored,
  and whenever several Parikh images appear concatenated, we instead add them together.
  So if we identify $\multiset{\Sigma}$ and $\Nat^{|\Sigma|}$,
  then for each nonterminal of arity $\ell$ in a production,
  we store a tuple consisting of $\ell$ words over
  $\Variables \cup \Nat^{|\Sigma|} \cup \Theta$.
  Variables from $X$ remain in an instance of a production until all its children
  in the derivation tree have been explored.
  
  Since $\grammar$ is not a normal \mcfg, $\cA$ also needs to handle terminals of the form
  $\tilde{\Gamma}^*$ for some $\tilde{\Gamma} \subseteq \Gamma$.
  To this end we extend the values in vectors for Parikh images to
  $\Natomega = \Nat \cup \{\omega\}$, where $\omega$ denotes arbitrarily many occurrences.
  For adding vectors together we use the obvious semantics
  $\omega + n = n + \omega = \omega + \omega = \omega$ for any $n \in \Nat$.
  Since $\tilde\Gamma$ may contain letters from $\Theta$, which we do not want to drop, whenever $\cA$ encounters a terminal $\tilde{\Gamma}^*$, it first guesses an arbitrary word
  $w$ of length $\leq k$ with $w \in \tilde{\Gamma}^* \cap \Theta^*$. 
  Then before and after each letter of $w$, $\cA$ inserts a vector containing $\omega$ for
  each $a \in \Gamma \cap \Sigma$. The resulting word is then used in place of $\tilde{\Gamma}^*$.
  
  During the traversal of the tree, the NFA checks that the computed tuple of words contains
  at most $k$ letters of $\Theta$.
  If any assignment contains more than $k$ such letters, $\cA$ rejects.
  Finally, once every variable for $\StartNonterminal$ is assigned, $\cA$ checks that the
  resulting word contains exactly $k$ letters in $\Theta$.
  Then $\cA$ starts outputting letters according to the stored assignment.
  For vectors in $\Natomega^{|\Sigma|}$, letters are output in a nondeterministically chosen permutation, and for each value of $\omega$, a nondeterministic amount is output on a loop, before moving on.
  Once the whole word for $\StartNonterminal$ has been output, $\cA$ moves to an accepting state.
  From the construction, it is clear that $\cA$ computes the language
  $\parikhalphdownck{\langof{\grammar}}{\Theta}$.
  
  It remains to analyze the size of $\cA$.
  Let $d$ be dimension of $\grammar$ and let $r$ be its rank.
  Any assignment to a single nonterminal contains at most
  $2 + k + d \cdot r$ many non-vector symbols.
  Since two vectors will always get added together if there is no other symbol between them,
  we have at most $4 + 2k + 2d \cdot r$ vectors in $\Natomega^{|\Sigma|}$ stored per production.
  By our requirements, we only need to consider derivation trees of height $c\sizeof{\grammar}$.
  Let $m$ be the maximum number of terminal symbols in a production in $\Productions$.
  Then the non-$\omega$ values in the vectors across the entire tree are bounded by
  $m \cdot \ell^{c\sizeof{\grammar}}$,
  which for vectors of dimension $|\Sigma|$ require only polynomially many bits to store.
  Since each path is at most $c\sizeof{\grammar}$ long, each state of $\cA$ therefore has
  polynomial description size, meaning there are exponentially many possibilities for the states.
\end{proof}

}

\label{afterbibliography}
\newoutputstream{pagestotal}
\openoutputfile{main.pagestotal.ctr}{pagestotal}
\addtostream{pagestotal}{\getpagerefnumber{afterbibliography}}
\closeoutputstream{pagestotal}

\newoutputstream{todos}
\openoutputfile{main.todos.ctr}{todos}
\addtostream{todos}{\arabic{@todonotes@numberoftodonotes}}
\closeoutputstream{todos}

\end{document}